\documentclass[11pt]{article}
\usepackage{lowe}
\usepackage{microtype}
\usepackage[margin=1in]{geometry}

\newcommand*\samethanks[1][\value{footnote}]{\footnotemark[#1]}

\begin{document}

\title{Randomized truncation of quantum states}

\author{Aram W. Harrow~\thanks{Center for Theoretical Physics -- a Leinweber Institute, MIT, Cambridge, MA, 02139, USA}\and Angus Lowe~\samethanks[1]\and Freek Witteveen~\thanks{QuSoft and CWI, Science Park 123, 1098 XG Amsterdam, Netherlands}}
\date{\today}
\maketitle

\begin{abstract}
  A fundamental task in quantum information is to approximate a pure
  quantum state in terms of sparse states or, for a bipartite system,
  states of bounded Schmidt rank. The optimal deterministic
  approximation in each case is straightforward, and maximizes the
  fidelity: keep the largest entries or singular values.  On the other
  hand, random mixtures of sparse states can achieve quadratically
  improved trace distances, and yield nontrivial bounds on other
  distance measures like the robustness.

  In this work, we give efficient algorithms for finding mixtures of
  sparse states that optimally approximate a given pure state in
  either trace distance or robustness.  These algorithms also yield
  descriptions of efficiently samplable ensembles of sparse, or
  less-entangled, states that correspond to these optimal mixed
  approximations. This can be used for the truncation step of
  algorithms for matrix product states, improving their accuracy while
  using no extra memory and no asymptotic increase in runtime, and we
  demonstrate this improvement numerically.

  Our proofs rely on tools from convex optimization and zero-sum
  games, as well as rigorous gurarantees for computing and sampling
  from maximum-entropy distributions. In particular, a key step in our
  construction is {\em conditional Poisson sampling}, which is the
  well-studied problem of drawing fixed-size subsets from a population
  with known inclusion probabilities. We give improved algorithms for
  carrying out this sampling, which may be of independent interest.
\end{abstract}
\tableofcontents
\vspace{0.5em}
\emph{Editorial note: This is arxiv v3 of our paper.   It is posted for the benefit of the QIP PC to meet the arxiv requirement for QIP 2027.  We anticipate further updates soon to incorporate the fact that a proof has recently been posted of \cref{conj:diag}.  As a preview of those changes, we anticipate further runtime improvements to the algorithms in Section 8.
}

\section{Introduction}\label{sec:intro}

This paper is motivated by a deceptively simple question: given a target unit
vector, how well can we approximate it with a sparse, but potentially random, unit vector? We say a vector is $k$-sparse if it has at most $k$ nonzero
entries in the standard basis. Given a $d$-dimensional target vector $v$ and a sparsity parameter $k\leq d$, it is clear that keeping just the $k$
entries of $v$ with the largest absolute values gives the best
approximation among all deterministic choices for the $k$-sparse vector. Similarly, it is
well known that the optimal low-rank approximation to a matrix, as
measured by any unitarily invariant norm, is obtained by truncating
(keeping just a few of) its singular values. Another instance of this
type of problem arises when simulating quantum many-body systems using
tensor networks, where multipartite states are represented using
states that have low Schmidt rank with respect to several partitions
of the subsystems.  In each case---vectors, matrices, and tensors---we
refer to the problem of finding a sparse
approximations as {\em truncation}.

We focus on the setting of quantum
information, although truncation has far-reaching
applications in other areas of computer science as well.  Here, a unit vector $v$ corresponds to a pure quantum state
$\rho=vv^\dag$, where $v$ is interpreted as a column vector\footnote{We
eschew bra-ket notation in this paper, partly to emphasize that our
quantum results apply more generally to any normalized complex
vectors.}. The quantum setting suggests a few important features of
this problem.  First, observables are quadratic in $v$ (or linear in
$\rho$) and appoximation error is defined primarily in terms of deviations of these observables from their true values.  This can be measured using the standard {\em trace
distance} if we meaure additive error, or the less common {\em
robustness} if we consider multiplicative error.  Second, there are many natural physical settings in which superposition or entanglement are expensive resources, but randomness is cheap.  Thus it is natural to define the truncation task as approximating $\rho$ with a mixture of $k$-sparse states, or analogously a mixture of states with Schmidt rank at most $k$.

The main result of our paper is that when we define truncation in this
way, the simple answer of taking the top $k$ entries of $v$ is no
longer optimal: Even when we are approximating pure states, the
optimal truncation for either trace distance or robustness will
generally be a mixed state. As one might expect, finding these optimal mixed truncations is more complicated than simply sorting the vector and keeping the top $k$ entries; nevertheless, we show it can be done efficiently (see \Cref{thm:informal-approx}).

A key application of our result is the truncation of matrix product states (MPSs).  Traditionally this is
done by running a series of singular value decompositions (SVDs) and keeping the top $k$ singular values at each step.  Our randomized method improves the accuracy obtained in each of these truncations as measured by, say, the trace distance, at the cost of some overhead. However, as we show in \Cref{sec:fast_algorithms}, the asymptotic runtime is no worse than that of the SVD on its own. We remark that this question is really about classical optimization and computation; the input $v$ is given as a classical description of the quantum state and the algorithms are classical.

\begin{figure}
  \centering
  \begin{subfigure}[b]{0.3\textwidth}
    \centering
    \includegraphics[width=\linewidth]{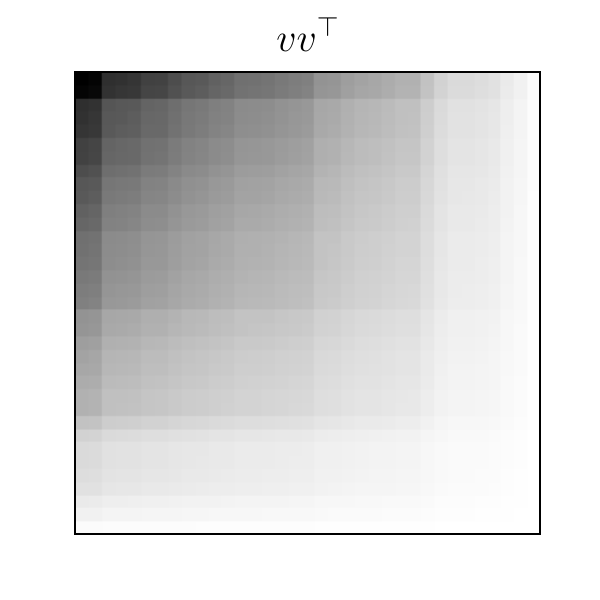}
    \caption{\label{fig:density_matrix_plot_1}}
  \end{subfigure}
  \begin{subfigure}[b]{0.3\textwidth}
    \centering
    \includegraphics[width=\linewidth]{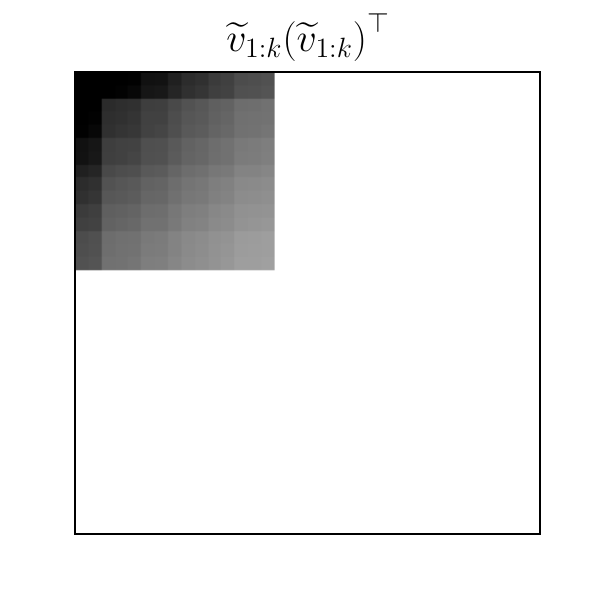}
    \caption{\label{fig:density_matrix_plot_2}}
  \end{subfigure}
  \begin{subfigure}[b]{0.3\textwidth}
    \centering
    \includegraphics[width=\linewidth]{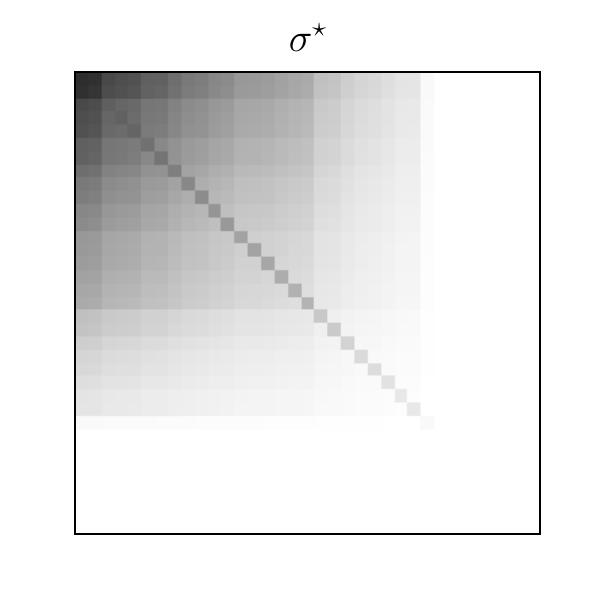}
    \caption{\label{fig:density_matrix_plot_3}}
  \end{subfigure}
  \caption{\label{fig:density_matrix_plot} With $d=35$, $k=15$, and darker cells representing larger values: (a) Quantum state given by a random unit vector $v\in\bbR^d$ with positive entries sorted in decreasing order. (b) Approximation obtained by keeping the first $k$ entries of $v$ and normalizing, which yields $\widetilde{v}_{1:k}$. (c) Optimal approximation $\sigma^\star$ of $vv^\top$, with respect to trace distance, by a convex combination of states given by $k$-sparse unit vectors.}
\end{figure}

\subsection{Simple example}\label{sec:simple-example}
To make the benefits of mixed approximations concrete, we will describe a simple but instructive example.
Let $v = (\sqrt{1- 2\eps}, \sqrt{\eps}, \sqrt{\eps} )^\top$
and suppose we want to approximate this by a 2-sparse state.
The best approximation in fidelity is given by either
$w_1 \propto (\sqrt{1- 2\eps}, \sqrt{\eps}, 0)^\top$ or $w_2 \propto (\sqrt{1- 2\eps}, 0, \sqrt{\eps})^\top$,
both giving a fidelity $F(v,w_i) = \sqrt{1 - \eps}$, and hence a trace distance of $T(v,w_i) = \sqrt{\eps}$.
If we look at the density matrix
\begin{align}
  vv^\top = \begin{pmatrix}
           1 - 2\eps              & \sqrt{\eps(1 - 2\eps)} & \sqrt{\eps(1 - 2\eps)} \\
           \sqrt{\eps(1 - 2\eps)} & \eps                   & \eps                   \\
           \sqrt{\eps(1 - 2\eps)} & \eps                   & \eps
         \end{pmatrix}
\end{align}
then we see that when truncating to one of the $w_i$, the dominant error in trace distance comes from the off-diagonal terms on the order of $\sqrt{\eps}$.
We can mitigate this by taking a uniform mixture of states $\tilde w_1 \sim (\sqrt{1- 2\eps}, 2\sqrt{\eps}, 0)^\top$ and $\tilde w_2 \sim (\sqrt{1- 2\eps}, 0, 2\sqrt{\eps})^\top$.
This way, we see that the mixed state
\begin{align}
  \sigma = \frac12 \tilde w_1\tilde w_1^\top+ \frac12 \tilde w_2 \tilde w_2^\top
  = \frac{1}{{1+2\eps}} \begin{pmatrix}
                          1 - 2\eps              & \sqrt{\eps(1 - 2 \eps)} & \sqrt{\eps(1 - 2\eps)} \\
                          \sqrt{\eps(1 - 2\eps)} & 2\eps                   & 0                      \\
                          \sqrt{\eps(1 - 2\eps)} & 0                       & 2\eps
                        \end{pmatrix}
  \label{eq:example-sigma}
\end{align}
only differs from $v$ in trace distance to order $\eps$ instead of $\sqrt{\eps}$.
Intuitively, the scarce resource here is the superposition between the 1st entry and the 2nd and 3rd entries, and we can boost this by increasing the amplitude of the 2nd and 3rd entries.  This introduces inaccuracies of order $\eps$ along the diagonals while fixing the error of order $\sqrt\eps$. The tradeoff between these two sources of error is also visible in the larger example illustrated in \Cref{fig:density_matrix_plot}.

The approximation in \cref{eq:example-sigma} turns out to achieve the optimal robustness of $2\eps$, meaning $vv^\top \preceq (1+2\eps)\sigma$. In \cref{sec:robustness} we will see more generally how to construct this sort of robustness-optimal ensemble.

\subsection{Overview of results}
\paragraph{Efficient algorithms.} We study the task of optimally approximating a pure state by a mixture of $k$-sparse states.  This is both a computational and a mathematical problem. It is important that our target states be pure since the more general question of approximating mixed states with mixtures of $k$-sparse states --- or even determining whether a given mixed state can be decomposed into $k$-sparse states --- is $\NP$-hard.  These, and other hardness results, are discussed in \cref{sec:hardness}.

Our main result is a collection of efficient algorithms that can find these randomized approximations when the state is pure. As we explain in further detail below, the runtime is dominated by the time required to find the right parametrization of a certain Gibbs distribution. 
\begin{thm}[Informal]\label{thm:informal-approx}
  Given a pure state $v\in\bbC^d$ and a sparsity bound $k\leq d$, there is an efficient algorithm to find a description of the optimal randomized approximation to $v$ in terms of either trace distance $T$ or robustness $R$.  This description includes:
  \begin{enumerate}
    \item the optimal value of $T$ or $R$;
    \item the corresponding density matrix (with some additional polynomial-time overhead); and
    \item a procedure to sample $k$-sparse pure states from this distribution in time $O(d)$.
    \end{enumerate}

    The same results hold when approximating a pure bipartite state $v\in\bbC^{d_1}\otimes \bbC^{d_2}$ by a mixture of states with Schmidt rank at most $k$, the relevant task for MPS truncation.
\end{thm}

These algorithms are given in \cref{sec:trace-dist,sec:robustness}, building on sampling algorithms described in \cref{sec:sampling}. Specifically, the trace distance version of this statement follows from Theorems~\ref{thm:td_optimal_value_computation}, \ref{thm:opt_td_state}, and \ref{thm:opt_td_sampling}. For the robustness, the optimal value computation appears in prior work~\cite{Regula18,kcoh18}, and we add to this by giving an efficient method to describe and sample from the optimal solution in Theorems~\ref{thm:robust_tau} and \ref{thm:robust_ensemble}.

We now describe the form of our approximating ensembles.  We find that for both trace distance and robustness, the approximating sparse states involve keeping some of the top entries of $v$ and sampling from some or all of the remaining entries.
Concretely, assume without loss of generality that in the standard basis
\ba
v = \sum_{i=1}^d v_i e_i,
\ea
where $e_i$ are the standard basis vectors, the entries $v_i$ are nonnegative and ordered non-increasingly.
Then the optimal approximation to $v$ samples a subset $S$ with probability $p(S)$ and takes a state $v_S$, corresponding to a density matrix
\ba
\rho = \sum_S p(S) w(S) w(S)^\dagger
\ea
Here,
\ba\label{eq:form state intro}
w(S) = \underbrace{\alpha \sum_{i=1}^r v_i e_i}_{\text{keep top } k - r - 1} \; + \underbrace{\frac{\beta}{\sqrt{r+1}} \sum_{i \in S} e_i}_{\substack{\text{uniform superposition} \\ \text{on } S}}
\ea
where $\alpha$ and $\beta$ are some appropriately chosen values, and $r \leq k-1$ an index that depends on $v$. The set $S$ is sampled from the indices that are not kept deterministically $\{k-r,k-r+1,\dots,d\}$. For the trace distance approximation, $S$ is sampled from a subset $\{k-r,k-r+1,\dots,\ell\}$ for $\ell \leq d$, so some indices are truncated deterministically.
The probability distribution on subsets $S$ of size $r+1$ needs to be such that the marginals $q_i = \Pr(i \in S)$ and $Q_{ij} = \Pr(i,j \in S)$ satisfy certain constraints.
The precise values of $\alpha, \beta$ and the distribution $p$ depend on whether one targets an optimal approximation for trace distance or robustness. Let us give a complete description for the robustness; a precise specification for the optimal approximation for the trace distance can be found in \cref{thm:opt_td_state} and \cref{eq:optimal_k_sparse_state_form_trace_distance}.
For the robustness, let $r$ be the unique index such that (recall we assumed $v_i \geq 0$ and non-increasing)
\ba
v_{k-r-1} > \frac{s_{k-r}}{r+1} \geq v_{k-r} \qquad \text{where} \qquad s_{k-r} := \sum_{i=k-r}^d v_i
\ea
and let
\ba
R_k = \sum_{i=1}^{k-r-1} v_i^2 + \frac{s_{k-r}^2}{r+1} - 1.
\ea
Then $R_k$ is the optimal robustness approximation.
The states in \cref{eq:form state intro} are written as
\ba
w(S) = \frac{1}{\sqrt{1 + R_k}}\left( \sum_{i=1}^r v_i e_i + \frac{s_{k-r}}{r+1} \sum_{i \in S} e_i \right)
\ea
and the probability distribution over subsets needs to be such that the marginals satisfy $q_i = v_i(r+1)/s_{k-r}$ and $Q_{ij} \leq q_i q_j$.

The type of sampling we need is known as \emph{conditional Poisson sampling} in the survey sampling literature. We also refer to it as {\em importance sampling of sets} in this work and discuss the task in more detail in \cref{sec:sampling}.   After defining the problem mathematically, there are several algorithmic questions; primarily {\em inference}, or determining the distribution parameters given the marginals, and {\em sampling} given the distribution parameters.  Since we would like our randomized truncation algorithms to be of use for numerical linear algebra tasks, such as MPS calculations, we give new and faster algorithms for inference, sampling and other tasks, which may be of independent interest.

\paragraph{Quantifying the advantage of randomized truncation.}
While our results give efficient algorithms for finding the optimal approximations in trace distance and robustness, we would also like to have an intuitive explanation of when these improve significantly over the deterministic truncation. As we explain in \Cref{sec:distance_measures}, the best possible improvement in terms of the accuracy is quadratic. To get a sense of when the quadratic improvement in trace distance is or is not possible, we discuss a range of representative examples in \cref{sec:examples}.  This turns out to depend on the decay rate of the sorted entries of the target vector, and specifically on the relation between the $\ell_1$ and $\ell_2$ norms of its tails.

We also note that, while the improvement in trace distance is at most quadratic, in order to facilitate a comparison along the lines of recent work in randomized compiling (see \Cref{sec:related}), which aims to reduce gate count, one may wish to determine the advantage of randomized truncation in terms of the relevant resource being conserved; namely, the sparsity parameter $k$. This advantage depends on the functional form of $k$ with $\sqrt{\eps}$, the best deterministic trace distance, and $T$, the best randomized trace distance, through $k(\sqrt{\eps})/k(T)$, and may in principle be super-quadratic.

\paragraph{Numerics and application to MPS.} An important motivation for our work is the application to randomized methods for simulating quantum systems, which we explore numerically in \cref{sec:MPS}. In particular, MPS are a widely used class of states where our randomized truncation schemes serve as a drop-in replacement for the usual deterministic approach to reducing bond dimension.
Code for our paper is available at \url{https://github.com/angusjlowe/rtrunc}. It includes a package for computing optimal approximations for pure states as well as code to reproduce the plots in this paper.

\subsection{Related work}\label{sec:related}

Randomized approximations to classical and quantum objects have appeared in several different contexts.

\paragraph{Randomized compiling.}  A related question to the one we study in this work is how well unitary operators can be approximated by mixtures of unitaries. Here, the power of randomness is somewhat greater for channels than for states, which we discuss further in \cref{sec:intuition}.   Prior work by Campbell~\cite{Campbell17} and Hastings~\cite{Hastings16} first showed (nonconstructively) that mixing unitaries can quadratically reduce diamond-norm error under general conditions.  This led to randomized product formulas for Hamiltonian simulation~\cite{COS19,qDRIFT,OWC20,HW23}, rotations~\cite{Low21,koczor2024probabilistic}, and Quantum Signal Processing~\cite{MR25}. The general question of approximating the set of all unitaries by convex combinations of an allowed set of unitaries was studied by Akibue, Kato and Tani~\cite{AKT24-unitary}, who found that mixing quadratically improves the worst-case error, i.e., given an $\sqrt{\eps}$-net of unitaries, taking convex combinations yields an $O(\eps)$-net in diamond norm.

\paragraph{Probabilistic state synthesis.} Randomized approximations to states have also been considered previously in \cite{AKT22,AKT24-state}. Given a set of allowed pure states in the random mixture, these works show a quadratic improvement in accuracy for the problem of \textit{worst-case approximation} (i.e.~the hardest possible target state $v$) as well as when certain symmetry conditions are met. Their setting does not substantially overlap with ours since we consider approximations for arbitrary pure states, which in general do not obey their symmetry conditions. Indeed, in
\cref{sec:examples}, we give examples where the quadratic accuracy
improvement of \cite{AKT24-state} does not apply.  Refs.~\cite{AKT22,AKT24-state} also
provide algorithms based on semidefinite programming for finding the optimal randomized approximation, whose runtime scales polynomially in the number of allowed states in the mixture. In our setting, this number is at least exponential in the sparsity parameter $k$. See \cref{sec:hardness} for further
discussion on this point.  Although this prior work does not resolve our problem, their findings of widespread quadratic advantage serve as inspiration for pursuing efficient algorithms applicable to our setting. Finally, we note that simultaneous work~\cite{wang2025faster} studies the application of a randomized approach to state preparation on quantum devices, and considers a similar set of examples to those appearing in \Cref{sec:examples}.

\paragraph{Resource theory of coherence.} The idea of a ``resource theory'' is to define a set of ``free'' states and/or operations, and then to use this to define distances from these objects which quantify some resource of interest~\cite{CG19}.  The original example of this was to study entanglement as a resource when local operations and classical communication (LOCC) is free~\cite{BBPS96}.  A natural way to quantify a resource is the robustness (related to the max-relative entropy) of a given state $\rho$ with respect to the free states~\cite{VT99,LBT19}.  The resource theory that directly relates to our work is that of ``$k$-incoherence" (also known as ``$k$-coherence''), which treats mixtures of $k$-sparse states as free, and quantifies the distance of a given state from this set; see \cite{SAP17} for a review of this topic. 
The $k=1$ case was studied in \cite{Chen_2016_quantifying}, using both the trace distance and robustness as resource measures. The $k>1$ case was analyzed in detail in \cite{Regula18,kcoh18}.  Among other results, these works show that for pure states, different notions of robustness are equal to each other and can be expressed in terms of a norm called the $k$-support norm (see \Cref{def:k-support}). Our result on the robustness in \Cref{sec:robustness} can be seen as a constructive version of the upper bound on the robustness given in Theorem~1 of \cite{kcoh18}.

\paragraph{The $k$-support norm.}
Both the optimal fidelity and the optimal robustness can be understood in terms of two norms which are dual to one another: the top-$k$ norm, and the $k$-support norm.   These are defined in \cite{AFS12} and reviewed in our paper in \cref{sec:distance_measures}.  In \cite{AFS12} these norms were used to develop convex relaxations for the task of finding sparse solutions to regression problems.  They also gave an expression for the proximal operator for the $k$-support norm~\cite[{Section~4}]{AFS12} which is similar to our trace distance optimization in \Cref{sec:trace-dist}.
Gauge norms are also studied in optimization settings~\cite{convex-inverse12} where the goal of finding a simple
model can be formalized through constraints on sparsity or rank, or more generally, being a linear combination of a small number of atoms.

\paragraph{Tensor network Monte Carlo.}
An important application of truncation, both deterministic and randomized, is to reduce the bond dimension in matrix product states and tensor networks.   An immediate application of our work is to approximate a tensor by a lower-rank tensor, as we will discuss in more detail in \Cref{sec:MPS}.  Ferris~\cite{Ferris15} has proposed a tensor truncation scheme with some similarities to ours, and this idea has also been studied more recently in~\cite{todo2024markov}.  We defer the detailed comparison of our work with \cite{Ferris15} to \Cref{sec:MPS}.  As a brief summary of the comparison, our methods provably minimize robustness or trace distance for bipartite states, while Ferris' method samples indices  using a heuristic which is not based on the solution to an optimization problem.

\paragraph{Randomized SVD.}
Truncating pure entangled states can be thought of as finding a low-rank approximation to a matrix, which is a basic task in numerical linear algebra.  The standard method to solve this problem is to compute the singular value decomposition (SVD) and keep only the top singular values. This is precisely the deterministic truncation scheme that we generally compare our methods against.  In defining error for randomized methods, there are two main possibilities: we can evaluate the error either before or after the random choices are made. Existing algorithms evaluate the error {\em after} the random choices are made.  If the original matrix is $M$ and the random low-rank approximation is $\tilde M$, then they evaluate something like $\bbE[\text{dist}(M,\tilde M)]$ for some distance measure.  In this framework, randomized algorithms are, by definition, no more accurate than deterministic algorithms, although they may improve the efficiency. This is the setting typically considered for randomized SVDs; see \cite{KT24} for a review.

Our framework corresponds instead to evaluating the error {\em before} the random choices are made.  Suppose that $\text{dist}(M,\tilde M) = \sup_\varphi \varphi(M-\tilde M)$ where $\varphi$ ranges over some set of linear functionals.  Norms such as operator norm or Schatten norms fall into this framework.  Our approach then consists of choosing a distribution to minimize
\be
\sup_\varphi \bbE_{\tilde M}[\varphi(M-\tilde M)].
\ee
This setting is more realistic if $\varphi$ represents a test that is unknown in advance, but is not chosen adversarially based on our choice of approximation $\tilde M$.
As a result of our framework, it is possible to be more accurate (in terms of bias) than the low-rank deterministic approximation found using the SVD.  We believe this is likely to have applications beyond quantum information.


\subsection{Proof techniques and intuition}\label{sec:intuition}
Our work shows that in order to approximate a pure state, given a library of simpler pure states, the best accuracy comes from introducing randomness. This may at first seem counterintuitive, but there are a few reasons why we might expect this to be the case.

One explanation comes from zero-sum games. The optimal trace distance approximation can be viewed as such a game where, given a target pure state $v$, we seek the value of
\be
\min_{\sigma\in\cI_k}
\max_{0\preceq M\preceq \mathds{1}}\tr(M(vv^\dag-\sigma)).
\ee
Here, $\cI_k$ is the set of mixtures of $k$-sparse states.  This has the form of a zero-sum game, with the $\sigma$ player trying to minimize the distinguishability and the $M$ player trying to maximize it.  A common feature of equilibria is that the optimal strategies are mixed.  In  game terms, the $\sigma$ player often does better with a mixed strategy because they may then perform simultaneously well against many possible responses by the $M$ player.  This picture also tells us that a fixed $M$ obviates the need for a mixed strategy by the $\sigma$ player.  This happens when the goal is to minimize an observable such as the energy of a system, or when maximizing the fidelity with a fixed state.  In these cases, optimal pure strategies exist, while optimal mixed strategies may or may not also exist.

This principle also explains why randomness helps with approximating unitaries.  The diamond norm involves a maximization over both the final distinguishing measurement {\em and} the choice of input state.   The choice of final measurement is similar to the state case that we consider in this paper, and leads to analogous up-to-quadratic improvements.  The maximum over initial states can lead to more dramatic, dimension-dependent, improvements.  For example, the $d$-dimensional unitaries $U_i = I - 2 e_i e_i^\dag$ are each far from the identity but the average channel $\cE(\rho) = \bbE_{i\in[d]} U_i \rho U_i^\dag$ is $O(1/d)$ close to the identity.  Similar ideas are used in qDRIFT and multiqubit Toffoli gates~\cite{qDRIFT,CHKT21,GKZ25}.

As we saw in the state case, the advantage from randomness comes from achieving good performance against many different tests simultaneously.  Thus  if we fix  both the input state and the final observable, then again deterministic approximation can achieve the optimal error.


Another way to understand the quadratic improvement in trace distance comes from the geometry of the state space, an idea which is due to \cite{AKT22,AKT24-unitary,AKT24-state}.  This is illustrated in  \Cref{fig:convex_combination_diagram}, where the best pure approximations to $v$ are $w_1,w_2$, each at distance $\sqrt\eps$, but a convex combination of these states has distance $O(\eps)$.

\begin{figure}
  \centering
    \begin{tikzpicture}
  \def\R{10} 
  \coordinate (C) at (0,0);

  \coordinate (Pone)    at ($(C)+(60:\R)$);   
  \coordinate (Ptwelve) at ($(C)+(90:\R)$);   
  \coordinate (Pelev)   at ($(C)+(120:\R)$);  


  \draw[line width=1.2pt]
    ($(C)+(50:\R)$) arc[start angle=50, delta angle=80, radius=\R];

  \foreach \ang in {120,90,60} {
    \draw[very thick] ($(C)+(\ang:\R+0.25)$) -- ($(C)+(\ang:\R-0.25)$);
  }

  \draw[dashed] (Pelev) -- (Pone);

  \path let \p1=(Pelev), \p2=(Pone) in coordinate (M) at ($(\p1)!0.5!(\p2)$);
  \draw (Ptwelve) -- (M) node[midway, left] {$O(\eps)$};

  \draw (Pelev) -- (Ptwelve) node[midway, below] {$\sqrt{\eps}$};
  \draw (Pone) -- (Ptwelve) node[midway, below] {$\sqrt{\eps}$};
  
  \node[above left]  at ($(Pelev)+(0,0.2)$) {$w_1w_1^\top$};
  \node[above]       at ($(Ptwelve)+(0,0.2)$) {$vv^\top$};
  \node[above right] at ($(Pone)+(0,0.2)$) {$w_2w_2^\top$};

  \node[below] at (M) { $\sigma$};
\end{tikzpicture}
  \caption{\label{fig:convex_combination_diagram} Illustration of the geometric intuition underlying cases where randomized truncation offers a significant improvement in approximating a pure state $vv^\top$.}
\end{figure}
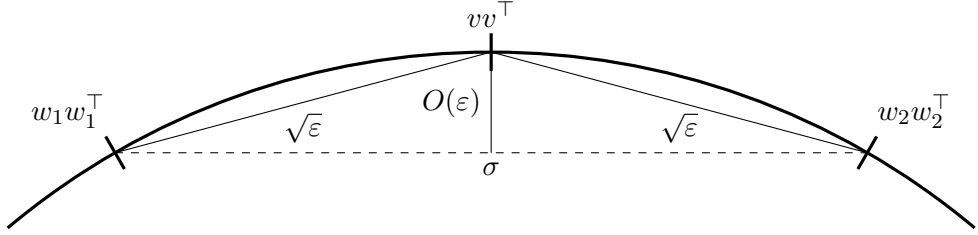

Finally, our results illustrate the value of conditional Poisson sampling, the main idea for which we now sketch.  Suppose $v_1,\ldots,v_d$ and $x_1,\dots, x_d$ are positive reals with $\sum_i v_i=V$ and the goal is to estimate $s=\sum_{i=1}^d v_i x_i$.   A natural unbiased estimator, familiar from Monte Carlo algorithms, is obtained by sampling $i$ with probability $v_i/V$ and outputting $\hat s_i \coloneq V x_i$. This procedure can alternatively be viewed as constructing a randomized $1$-sparse approximation to a vector $v=(v_1,\dots, v_d)$. Namely, if $(e_j:j=1,\dots, d)$ are the standard basis vectors and $i$ is drawn randomly with probability $v_i/V$ then the vector-valued random variable $Ve_i$ is an unbiased estimator of $v$.

If we wish to now generalize to a $k$-sparse approximation then we need to sample a random subset of the indices $S = \{i_1,\ldots,i_k\} \subset [d]$.  As before, we can choose to impose that this $k$-sparse estimator be unbiased for the target vector $v$.  This means that if $i\in S$ with probability $q_i\in (0,1]$ then the random vector ought to be of the form
\be\label{eq:k_sparse_estimate_intuition}
w(S) = \sum_{i\in S} \frac{v_i}{q_i}e_i.
\ee
Since $|S|=k$ and the $q_i$ represent marginal probabilities, we must have $\sum_{i=1}^d q_i = k$. The task at hand --- say, minimizing the variance of some function of this random vector --- informs the desired values of the marginal probabilities $q_{1},\dots, q_d$.

For the randomized approximations in robustness and trace distance, we do not explicitly enforce that the $k$-sparse states be unbiased in the sense described above, but rather are of the form in \cref{eq:form state intro}, where we keep some terms deterministically, and sample from the other indices.
This can be seen as a refinement of the approximation by importance sampling of sets in \cref{eq:k_sparse_estimate_intuition}.


Conditional Poisson sampling, which is needed to compute and sample from the optimal ensemble, is a nontrivial algorithmic challenge. Algorithms for this task are known from the literature on maximum entropy sampling. We draw on this past work and  improve upon it as well in \cref{sec:sampling}, which may be of independent interest.


\subsection{Open questions}
We give a comprehensive analysis of the computational problem of finding optimal randomized approximations of pure quantum states by $k$-sparse states or states with Schmidt rank $k$. However, there are various interesting questions for future work.

\paragraph{Application to tensor networks.} We propose tensor network truncation as a natural application of randomized truncation. How useful this is in practice, and for \emph{which} applications precisely, will require detailed numerical study, which we leave to follow-up work.

\paragraph{Computing the sampling weights.}  We give a provable runtime bound for finding $\mu$ in \cref{thm:inference-algs} but we conjecture that this can be improved significantly.  One plausible conjecture which would imply this is \cref{conj:diag}.  These algorithms are discussed in more detail in \Cref{sec:fast_algorithms}

\paragraph{Applications beyond quantum states.} We focus on the case where the vector we try to approximate is a quantum state. A natural follow-up direction is to investigate applications beyond this interpretation. This may require generalizing to approximations in norms besides the (Schatten) 1-norm, such as the operator norm or general $p$-norms. Moreover, it may not be natural to impose a normalization condition on the approximating vectors, which changes the optimization problem.

\paragraph{Relation to Hamiltonian simulation.} As described in
\cref{sec:related}, there are several randomized Hamiltonian
simulation methods that can be viewed as solving a related problem.
Namely, given a Hamiltonian $H=\sum_{i=1}^L v_ih_i$, can we
approximate $e^{-iHt}$ with a random mixture of
$\exp(-i\sum_i w_i h_i)$ for a collection of $k$-sparse vectors $w$?
The qDRIFT algorithm can be viewed as doing this for 1-sparse $w$.
Applied to a starting state $\ket \psi$, and assuming
$h_i^2 \propto I$, this means approximating $e^{-iHt}\ket\psi$ with a
mixture of superpositions of $\ket\psi$ and $h_i\ket\psi$.  Applying
our robustness-optimal scheme with $k=2$, and neglecting higher-order
terms, turns out to yield precisely qDRIFT.  Improving this to higher
values of $k$ is a compelling open problem.  One of the key
difficulties is that qDRIFT addresses only first-order and
second-order error while modern Hamiltonian simulation methods
generally control higher-order errors as well.  An exploration of methods for partially randomized Hamiltonian simulation is initiated in recent work~\cite{Q4Bio-PhaseEstimation}.

\subsection{Notation}\label{sec:notation}
For a positive integer $d$, we write $[d]=\{1,2,\dots,d\}$.
The standard basis of $\bbC^d$ is denoted by $\{e_i : i \in [d]\}$. For any operator $X$ on $\bbC^d$ we let $X^\dag$ denote the Hermitian conjugate. We adopt the physics convention where the Hermitian inner product $\langle\cdot, \cdot\rangle$ is linear in the second argument, so that $\langle u,v\rangle=u^\dag v$ for any $u,v\in\bbC^d$. We write $v = \sum_{i=1}^d v_i e_i$, where $v_i$ denotes the coordinate of $v$ in the standard basis.
We let $|v_i^\downarrow|$ denote the $i$th largest element (in nonincreasing order) of the sequence $(|v_j| : j \in [d])$, and define
$
  v^\downarrow \coloneq \sum_{i=1}^d |v_i^\downarrow| e_i.
$
For $v\in\bbC^d$ and $1\leq j \leq d$, set
\begin{align}\label{eq:s_defn}
  s^{v}_j \;\coloneq\; \sum_{i=j}^d |v_i^\downarrow|.
\end{align}
The set of linear operators on $\bbC^d$ is denoted $\cL(\bbC^d)$, and the set of density matrices on $\bbC^d$ (positive semidefinite with unit trace) is denoted $\cD(\bbC^d)$.
We use $v,w,\dots$ for vectors in $\bbC^d$ (which also represent pure states), and $\rho,\sigma,\dots$ for density matrices in $\cD(\bbC^d)$. Matrix norms are Schatten $p$-norms throughout.
Given $v \in \bbC^d$ and $\cA \subseteq [d]$, we let $v_\cA$ denote the vector obtained by keeping coordinates $v_i$ for $i \in \cA$ and setting the rest to zero.
If $\cA$ is a contiguous block $\cA_{i,j} = \{i,i+1,\dots,j\}$, we also write $v_{i:j} \coloneq v_{\cA_{i,j}}$.
For $p \in \bbZ_{>0}$ we denote the $\ell_p$ norm of $v$ by $\norm{v}_p$, and $\norm{v}_0$ is the number of nonzero coordinates of $v$ in the standard basis.
Finally, for a set $S$ and scalar $a$, the rescaled set is $aS = \{as : s \in S\}$, and for a point $g$ we write $g+S = \{g+s : s \in S\}$.
Other notation will be introduced as needed.

\section{Quantifying coherence and entanglement}\label{sec:distance_measures}
Diagonal density matrices represent classical states and can be viewed as mixtures of 1-sparse pure states: the standard basis vectors. By considering mixtures of $k$-sparse states, we can generate increasingly non-classical states, as measured by the resource theory of coherence.
\begin{dfn}[$k$-sparse mixture, $\cI_k$]
  A pure state $v \in \bbC^d$ is \emph{$k$-sparse} if $\norm{v}_0\leq k$.
  A mixed state $\rho \in \cD(\bbC^d)$ is a \emph{$k$-sparse mixture} if
  $
    \rho = \sum_{j}p_j u_j u_j^\dagger
  $
  for some $k$-sparse pure states $u_j\in \bbC^d$ and probabilities $p_j>0,\ \sum p_j = 1$\footnote{In more conventional quantum notation we write $\rho = \sum_j p_j \dyad{u_j}$.}. We denote the set of all $k$-sparse mixtures by $\cI_k(\bbC^d)$, or simply $\cI_k$ when the Hilbert space is evident.
\end{dfn}
The notation $\cI_k$ is in reference to previous work which refers to this as the set of ``incoherent" states~\cite{kcoh18}. Although this definition is basis dependent, it is closely related to other basis-independent objects.
For example, another familiar resource in quantum information is bipartite entanglement, which can be measured by the Schmidt rank.
A bipartite state $v\in \bbC^a\otimes \bbC^b$ has a Schmidt decomposition given by
\begin{align}\label{eq:schmidt_decomp}
  v = \sum_{i=1}^r \lambda_i\, x_i \ot y_i,
\end{align}
for some positive integer $r\leq \min\{a,b\}$, orthonormal bases $\{x_i\}\subset \bbC^a$ and $\{y_i\}\subset \bbC^b$, and Schmidt coefficients $\lambda_i > 0$, $\sum\lambda_i^2 = 1$.
\begin{dfn}[Schmidt number $k$, $\cS_k$]
  A pure state $v$ has \emph{Schmidt rank $k$} if it has $k$ nonzero Schmidt coefficients. A mixed state $\rho \in \cD(\bbC^a \ot \bbC^b)$ has \emph{Schmidt number at most $k$} if
  $
    \rho = \sum_{j} p_j u_j u_j^\dagger
  $
  for some pure states $u_j\in\bbC^a\otimes\bbC^b$ with Schmidt rank at most $k$ and $p$ a probability distribution.  We denote the set of mixed states with Schmidt number at most $k$ by $\cS_k(\bbC^a,\bbC^b)$ or $\cS_k$ when the Hilbert spaces are evident.
\end{dfn}
The Schmidt number for density matrices was originally introduced and studied in~\cite{terhal2000schmidtnumber}. We next briefly recall standard quantum distance measures; see e.g. \cite[{Ch.~5}]{nielsen_chuang}.
The \emph{fidelity} between states $\rho,\sigma$ is
$
  F(\rho,\sigma) = {\lVert\sqrt{\rho}\sqrt{\sigma}\rVert}_1.
$
If $\sigma = vv^\dagger$ is pure, then $F(\rho,\sigma) = \sqrt{\langle v, \rho v\rangle}$, and if $\rho=ww^\dag$ is also pure we have $F(vv^\dag, ww^\dag) = |\langle v,w\rangle|$.
The \emph{trace distance} is
$
  T(\rho,\sigma) = \tfrac12 \norm{\rho-\sigma}_1.
$
When convenient, we abbreviate $F(v,w) := F(vv^\dagger, ww^\dagger)$ and similarly for $T$.
We also use $F(v,\rho)$ and $T(v,\rho)$ for comparisons between a mixed state $\rho$ and a pure state given as a unit vector $v$.  For two pure states $v,w$ these are related by $F(v,w)^2 + T(v,w)^2=1$.  More generally, the Fuchs-van de Graaf inequalities relate $F$ and $T$ by
\be\label{eq:fvdg}
1-F \leq T \leq \sqrt{1-F^2}.
\ee
If one of the two states is pure, the lower bound can be improved, and we have
\be\label{eq:fvdg_pure}
1 - F^2 \leq T \leq \sqrt{1-F^2}.
\ee
We are interested in the maximum fidelity (or minimum trace distance) achievable by approximating a state with one that has limited coherence or entanglement. The lower bound in \cref{eq:fvdg_pure} implies that, for pure states, we cannot hope to attain more than a quadratic advantage using a mixed approximation. Here and throughout, we let $d,a,b\in\bbZ_{>0}$ such that $a,b\leq d$.
\begin{dfn}[Fidelity with $\cI_k$ or $\cS_k$]
  For $\rho \in \cD(\bbC^d)$, let
  $
    F_k(\rho) := \max_{\sigma \in \cI_k} F(\rho,\sigma).
  $
  For $\rho \in \cD(\bbC^a \ot \bbC^b)$, let
  $
    F_k^{(E)}(\rho) := \max_{\sigma \in \cS_k} F(\rho,\sigma).
  $
\end{dfn}
As in the case of $F(\cdot,\cdot)$, if $v$ is a unit vector corresponding to a pure state we abbreviate $F_k(vv^\dag)$ by $F_k(v)$ and similarly for $F^{(E)}_k(\cdot)$. These measures are straightforward to compute for pure states.
\begin{lem}\label{lem:fidelity_equals_top_k}
  If $v\in\bbC^d$ is pure then for any positive integer $k\leq d$ it holds that
  $
    F_k(v)^2 = \sum_{i=1}^k |v_i^{\downarrow}|^2.
  $
\end{lem}
Similarly, if $v\in\bbC^a\otimes \bbC^b$ is a bipartite state with a Schmidt decomposition of the form in \cref{eq:schmidt_decomp}, we have the following.
\begin{lem}\label{lem:fidelity_ent_equals_top_k}
  If $v$ is of the form given in \cref{eq:schmidt_decomp} then for any positive integer $k\leq r$ it holds that
  $
    F_k^{(E)}(v)^2 = \sum_{i=1}^k \lambda_i^2.
  $
\end{lem}
\begin{proof}
  By definition, $F_k^{(E)}(v)^2 = \max_{\sigma\in\cS_k}\langle v, \sigma v\rangle$. The function $\sigma\mapsto \langle v,\sigma v\rangle$ is linear and the set $\cS_k$ is convex so it suffices to consider a maximization over the extreme points of $\cS_k$. Thus, $F_k^{(E)}(v)^2 = \max_{\varphi}|\langle \varphi, v\rangle|^2$ where the maximization is over pure states on $\bbC^a\otimes \bbC^b$ having Schmidt rank at most $k$. The optimal solution to this problem is attained by the state proportional to keeping the top $k$ Schmidt coefficients, by Eckart-Young Theorem~\cite{Eckart1936,mirsky1960symmetric}.
\end{proof}
These facts agree with the intuition that if we want to maximize the overlap with a $k$-sparse state, we should keep the $k$ largest coordinates. We can also phrase this in terms of an appropriate norm, whose appearance in the context of truncation is natural.
\begin{dfn}[Top-$k$ norm]
  For $v\in\bbR^d$ and $1 \leq k \leq d$, the \emph{top-$k$ norm} is
  \begin{align}
    \norm{v}_{(k)} \;\coloneq\; \Bigl(\sum_{i=1}^k |v_i^\downarrow|^2\Bigr)^{1/2}.
  \end{align}
\end{dfn}
With this terminology, Lemmas~\ref{lem:fidelity_equals_top_k} and~\ref{lem:fidelity_ent_equals_top_k} state that
$
  F_k(v) = {\lVert v\rVert}_{(k)}
$
and, if $v$ is bipartite, $F_k^{(E)}(v) = {\lVert\lambda\rVert}_{(k)}$ where $\lambda$ is a vector of Schmidt coefficients.
The story is more interesting for optimal trace-distance approximations.
\begin{dfn}[Trace distance with $\cI_k$ or $\cS_k$]
  For $\rho \in \cD(\bbC^d)$, let
  $
    T_k(\rho) \coloneq \min_{\sigma \in \cI_k} T(\rho,\sigma).
  $
  For $\rho \in \cD(\bbC^a \ot \bbC^b)$, let
  $
    T_k^{(E)}(\rho) \coloneq \min_{\sigma \in \cS_k} T(\rho,\sigma).
  $
\end{dfn}


\medskip

Finally, we recall the \emph{robustness} of coherence and entanglement \cite{VT99, kcoh18}, which can be thought of as measuring a multiplicative approximation error by states in $\cI_k$ or $\cS_k$.

\begin{dfn}[Robustness with $\cI_k$ or $\cS_k$]\label{def:robustness}
  For $\rho \in \cD(\bbC^d)$,
  \begin{align}
    R_k(\rho) := \min \Bigl\{ s\geq 0 :  \frac{\rho+s\sigma}{1+s} \in \cI_k,\ \sigma \in \cD(\bbC^d), \Bigr\}.
    \label{eq:Rk-def}
  \end{align}
  For $\rho \in \cD(\bbC^a \ot \bbC^b)$,
  \begin{align}
    R_k^{(E)}(\rho) := \min \Bigl\{ s\geq 0 : \frac{\rho+s\sigma}{1+s} \in \cS_k,\ \sigma \in \cD(\bbC^a \ot \bbC^b) \Bigr\}.
  \end{align}
\end{dfn}
The robustness gives a bound on the trace distance through
\begin{align}\label{eq:robustness_upper_bound}
  T_k(\rho) \leq \frac{R_k(\rho)}{1 + R_k(\rho)}
\end{align}
which follows by taking the trace distance between $\rho$ and the optimum in~\cref{eq:Rk-def}. There is an alternative notion of the robustness which arises when one restricts the feasible set in the minimization in \Cref{def:robustness} to density matrices $\sigma$ which are themselves in $\cI_k$, or $\cS_k$. However, in \cite{kcoh18} (building on \cite{AFS12,Regula18}) it was proven that for pure states, these notions coincide. Moreover, in this case the robustness can be computed using the dual norm to the top-$k$ norm. We now define this norm, and collect some of its relevant properties. Below, we write $\binom{\cA}{k}$ for the family of $k$-element subsets of a finite set $\cA$.
\begin{dfn}[$k$-support norm]\label{def:k-support}
  The \emph{$k$-support norm} is the dual of the top-$k$ norm. An equivalent definition is
  \begin{align}
    \norm{v}_{(k,*)} \;\coloneq\; \min
    \ \Bigl\{\ \sum_{S\in {[d]\choose k}}\norm{w(S)}_2 \ :\ v=\sum_{S\in {[d]\choose k}} w(S)\ \Bigr\}
    \label{eq:k-support-norm}
  \end{align}
  for any $v\in\bbR^d$ and $1 \leq k \leq d$.  Here we minimize over collections of vectors $\{w(S)\}_S$ such that the support of each $w(S)$ is contained in $S$.
\end{dfn}
This norm was introduced in the context of sparse prediction problems in~\cite{AFS12}.  See also \cite{Regula18} and~\cite[Ch.~4]{Bhatia1997} for additional background. That $\norm{\cdot}_{(k)}$ and $\norm{\cdot}_{(k,*)}$ are dual to one another follows from the observation that the $k$-support norm is a gauge function for the set
\begin{align}
  C_k\;\coloneq\;\conv\bigl\{\,w\in\bbR^d:\ \norm{w}_0=k,\ \norm{w}_2\leq1\,\bigr\},
  \label{eq:Ck-def}
\end{align}
This implies $C_k$ is the unit ball for the $k$-support norm, and the top-$k$ norm is indeed the maximum inner product with a vector from this set.
For $k=1$ and $k=d$ one recovers the $\ell_1$ and $\ell_2$ norms, respectively.  \cref{eq:k-support-norm,eq:Ck-def} do not suggest an efficient algorithm for computing $\norm{v}_{(k,*)}$.
However, there is a closed-form expression for $\norm{v}_{(k,*)}$ which 
can be computed in $O(d\log d)$ time.
\begin{lem}[{\cite[Prop.~2.1]{AFS12}}]\label{lem:k_supp_efficient_expr}
  For every $v\in\bbR^d$,
  \begin{align}
    \norm{v}_{(k,*)}^2 \;=\; \sum_{i=1}^{k-r-1} |v_i^\downarrow|^2 \;+\; \frac{\bigl(s^{v}_{k-r}\bigr)^2}{r+1},
  \end{align}
  where $r\in\{0,1,\dots,k-1\}$ is the unique\footnote{The uniqueness of the choice of $r$ satisfying the condition in \cref{eq:condition k support norm} is stated as Lemma~3.1 in \cite{govind1985unitarily}.}  index satisfying
  \be\label{eq:condition k support norm}
  |v_{k-r-1}^\downarrow|\;>\;\frac{s^{v}_{k-r}}{r+1}\;\ge\;|v_{k-r}^\downarrow|
  \ee
  (with the convention $|v_0^\downarrow|:=+\infty$), and where $s^{v}_j$ is the partial sum defined in \cref{eq:s_defn}.
\end{lem}

The following optimization view underlies \cref{lem:k_supp_efficient_expr}.
For $v\in\bbR^d$, the optimal value for
\begin{equation}\label{eq:k_supp_opt_problem}
  \begin{aligned}
    \textnormal{max.}\quad & \langle v,x\rangle - \tfrac12 \norm{x}_{(k)}^2 \\
    \textnormal{s.t.}\quad & x\in\bbR^d
  \end{aligned}
\end{equation}
is equal to $\tfrac{1}{2}\norm{v}_{(k,*)}^2$. This is a standard fact about dual norms, stated explicitly in \Cref{lem:fenchel_conjugate_of_norm}.  (See \Cref{sec:convex} for a brief overview of the convex analysis relevant in this work.) \Cref{lem:k_supp_efficient_expr} is then an easy consequence of the following.
\begin{lem}\label{lem:optimal_soln_to_k_supp}
  Let $v\in\bbR^d$ and let $\ell$ be the smallest index with $v_\ell^\downarrow=0$ (set $\ell=d+1$ if none). Assume $\ell\geq k+1$.
  A vector $x$ is optimal for the maximization in~\cref{eq:k_supp_opt_problem} if and only if, with $r$ as in \Cref{lem:k_supp_efficient_expr},
  \begin{align}\label{eq:k_supp_opt_soln_condn}
    x_j^\downarrow=\begin{cases}
                     v_j^\downarrow,             & j=1,\dots,k-r-1,    \\[2pt]
                     \;\dfrac{s^{v}_{k-r}}{r+1}, & j=k-r,\dots,\ell-1.
                   \end{cases}
  \end{align}
\end{lem}
The proof is essentially shown in the proof of Proposition~2.1 in \cite{AFS12}. We give the full argument in \Cref{sec:k_supp_lem_proof}.

For pure states, the robustness is directly related to the $k$-support norm.  This was proven in \cite{Regula18} (see also \cite{kcoh18}) and our construction in \Cref{sec:robustness} also yields a simple alternate proof.
\begin{thm}[$R_k$ for pure states \cite{Regula18}] \label{thm:Rk-value}
  If $v \in \bbR^d$ is a pure state,
  \be
  R_k(v) = \norm{v}_{(k,*)}^2 - 1.
  \ee
  Also, if $v \in \bbC^a \ot \bbC^b$ has a vector of Schmidt coefficients $\lambda$,
  $
    R_k^{(E)}(v) = \norm{\lambda}_{(k,*)}^2 - 1.
  $
\end{thm}

For our applications we will usually want to calculate $\norm{\cdot}_{(k,*)}$ precisely.  But an $O(1)$-approximation is possible using the 1- and 2-norms. From \cite{AFS12} we have that for a unit vector $v$,
\ba
\frac{\norm{v}_1}{\sqrt{k}} \leq \norm{v}_{(k,*)} \leq \sqrt{2}\frac{\norm{v}_1}{\sqrt{k}},
\ea
which in particular gives $R_k(v) \geq \frac{\norm{v}_1^2}{k} - 1$.

\paragraph{Sorting entries.} A useful symmetry is that the distance measures we use, as well as our notions of sparsity, are invariant under permutation and phases.  Specifically, for the three distance measures introduced in this section and,  for any $v \in \bbC^d$,
\be\label{eq:perm_diag_invariance}
F_k(v) = F_k(v^{\downarrow}), \quad T_k(v) = T_k(v^{\downarrow}), \quad \textnormal{and}\quad R_k(v) = R_k(v^{\downarrow}).
\ee
Thus in all computations, we can assume without loss of generality that the entries are real, nonnegative, and ordered. Obtaining the optimal solutions for the optimization problems defining each of these quantities is then a straightforward matter of tracking the phases and ordering of entries in the original vector $v$, and we neglect this detail throughout this work for the sake of clarity.


\section{Computational hardness for mixed states}\label{sec:hardness}
How hard is it to compute approximations by sparse states or states with small Schmidt numbers?
For an arbitrary state $\rho \in \cD(\bbC^d)$ the problem is computationally intractable.
In fact, we will see that even weak membership of $\cI_k$ and $\cS_k$ is $\NP$-hard\footnote{Weak membership of a convex set $S$ is the task of deciding whether $\rho$ is in the $\eps$-interior or $\eps$-exterior of $S$ in Euclidean distance. Weak linear optimization is the task of approximately solving $\max_{\sigma\in S} \tr[M\sigma]$ up to additive error $\eps$.  These terms are defined in detail in \cite{GLS12}.}.
As a consequence, computing $F_k(\rho)$, $F_k^{(E)}(\rho)$, $T_k(\rho)$, and $T_k^{(E)}(\rho)$ is $\NP$-hard in general.

For entanglement, weak membership of $\cS_1$ is precisely the separability problem, which is known to be $\NP$-complete \cite{gurvits2003classical, gharibian2008strong}.
In contrast, checking 1-sparsity is trivial: $\rho\in\cI_1$ if and only if $\rho$ is diagonal in the standard basis.
However, for larger $k$, the problem again becomes $\NP$-hard.  
\begin{prop}
  Given $\rho \in \cD(\bbC^d)$ and $k$, the problem of weak membership of $\rho$ in $\cI_k$ is $\NP$-complete.
\end{prop}
\begin{proof}
By Carath\'{e}odory's theorem, any element in $\cI_k$ can be written as a mixture of at most $d^2 + 1$ $k$-sparse states. This shows the existence of an efficiently verifiable witness for weak membership, so the problem is in $\NP$.
  To prove hardness, we use the fact that weak membership, weak separation, and weak linear optimization are polynomial-time equivalent for convex sets \cite{GLS12}.
  Thus it suffices to consider weak linear optimization over $\cI_k$, i.e.
  \begin{align}\label{eq:optimize over incoherent}
    m_* \;=\; \max_{\sigma\in \cI_k} \tr[M\sigma].
  \end{align}
  The maximum is achieved at an extreme point of $\cI_k$, namely a pure $k$-sparse state.
  Equivalently,
  \begin{align}
    m_* \;=\; \max_{\substack{S\subseteq[d] \\ |S|=k}}\; \norm{M_S}_{\infty},
  \end{align}
  where $M_S$ is the $k\times k$ principal submatrix of $M$ indexed by $S$.

  We consider the case where $M = A_G$ is the adjacency matrix of a graph $G$ on $d$ vertices.
  Then each $M_S$ is the adjacency matrix $A_H$ of the induced subgraph $H=G[S]$.
  For any $k$-vertex graph $H$, we have $\norm{A_H}_{\infty} \le k-1$, with equality if and only if $H$ is the complete graph $K_k$.
  Moreover Proposition~1.3.10 of \cite{cvetkovic2010introduction} implies that
  if $H$ is not complete then $\norm{A_H}_{\infty}$ is bounded away from $k-1$ by an inverse polynomial in $k$.
  Thus the $k$-Clique problem reduces to weak linear optimization over $\cI_k$, proving $\NP$-hardness.
\end{proof}

We emphasize that this hardness requires $k$ to be a  variable.  If $k$ is fixed then $\cI_k$ can be expressed as a semidefinite program of size $d^{O(k)}$, since any $\sigma\in \cI_k$ can be expressed as a mixture of states in $\cD(\bbC^S)$ for $S$ ranging over all subsets of $[d]$ of size $k$.  Here $\bbC^S$ means the $|S|$-dimensional subspace of states with support contained in $S$.  The hardness results shows that this SDP is essentially optimal if finding $k$-cliques in $d$-vertex graphs takes time $d^{\Omega(k)}$.  This latter conjecture is implied by the Exponential Time Hypothesis and other plausible complexity-theory conjectures~\cite{ABW15}.

For \emph{pure} states, membership in $\cI_k$ and $\cS_k$ is easy to decide: in the first case by support size, in the second by Schmidt rank.
We have already seen that for pure $v$, the fidelities $F_k(v)$ and $F_k^{(E)}(v)$ are easy to compute.
In \Cref{sec:trace-dist} we will show that the trace distances $T_k(v)$ and $T_k^{(E)}(v)$ can also be computed efficiently, although these algorithms are considerably more complicated than those for the fidelity.

Our algorithms for $T_k(v)$ and $R_k(v)$ thus rely crucially on $v$ being a pure state and involve details of the properties of trace distance and robustness.  For a general distance measure on quantum states $D(\rho,\sigma)$, one can define $D(v,\cI_k) \coloneq \min_{\sigma\in\cI_k} D(vv^\dag, \sigma)$.  We do not have a general reduction from $D(v,\cI_k)$ to the pairwise distances $D(\rho,\sigma)$, nor do we know examples where the pairwise distances are easy but calculating  $D(v,\cI_k)$ is hard.

\section{Approximations in trace distance}\label{sec:trace-dist}
In this section we give efficient algorithms to compute the optimal randomized truncation to a pure state with respect to trace distance. More specifically, given a target unit vector $v \in\bbC^d$, our first algorithm, given in \Cref{sec:comp_opt_meas}, computes the value
\begin{align}\label{eq:main_td_optimization_problem}
  T_k(v):=\min_{\sigma\in \cI_k}\frac{1}{2}{\lVert vv^\dag-\sigma\rVert}_1=\min_{\sigma\in \cI_k}\max_{0\preceq M\preceq \mathds{1}}\tr(M(vv^\dag-\sigma)).
\end{align}
It also returns a description of a measurement $M^\star = mm^\dag$ with $m\in\bbC^d$ such that $(\sigma^\star,M^\star)$ is a saddle point, for some $\sigma^\star\in\cI_k$. A second algorithm, given in \Cref{sec:optimal_density_matrix_td}, then returns a description of the optimal solution $\sigma^\star$, either explicitly in the standard basis or in the form of an ensemble of $k$-sparse pure states $\{(p_i,v_i)\}_i$ where the distribution $p$ can be efficiently sampled.

\subsection{Computing the trace distance and optimal measurement}\label{sec:comp_opt_meas}
The main result of this section is:
\begin{thm}\label{thm:td_optimal_value_computation}
  Let $k,d\in\bbZ_{>0}$ such that $k\leq d$ and $v\in\bbC^d$ be a unit vector. There is an algorithm (\Cref{alg:td_approx}) running in time $O(dk + d\log d)$ that computes $T_k(v)$ as well as a measurement operator $M^\star$ such that $(\sigma^\star, M^\star)$ is a saddle point of the minimax in \cref{eq:main_td_optimization_problem}, for some $\sigma^\star\in\cI_k$.
\end{thm}


Our proof of \cref{thm:td_optimal_value_computation} will proceed by first showing that the optimal $M$ can be taken to be of the form $mm^\dag$, i.e.~a rank-1 measurement.  Then we derive the structure of $m$, proving that it consists of two pieces that are proportional to the entries of $v$ with a flat region in the middle (when sorted).  At this point, there is only a single parameter and it can be found by solving a cubic equation.  Then in \cref{sec:optimal_density_matrix_td} we will use this to find an optimal approximating ensemble of $k$-sparse states.  The three regions in $m$ correspond to three regions in our random $k$-sparse approximation $w$: the first region is proportional to the entries of $v$, the second region is uniform over a set (from which we will later importance sample), and the third region is zero.

The fact that $M$ can be taken without loss of generality to be rank-one is due to the following lemma.
\begin{lem}\label{lem:rank_1_suffices}
  For any $d\in\bbZ_{>0}$, density matrix $\sigma\in \cD(\bbC^d)$, and unit vector $v \in\bbC^d$, the optimum in
  $
    \max_{0\preceq M\preceq \mathds{1}} \tr(M(v v^\dagger-\sigma))
  $
  is attained by a rank-one projection operator.
\end{lem}
\begin{proof}
  By \Cref{fact:spectral_stability}, which is a standard consequence of Weyl's Inequality, the largest eigenvalue of the matrix $v v^\dagger-\sigma$ lies in the interval $[0,1]$, while the second largest eigenvalue lies in $[-1,0]$. If $\sigma = v v^\dagger$ then the optimal value of the maximization is zero, and $M= m m^\dagger$ is an optimal solution for any unit vector ${m}\in\bbC^d$. Otherwise, $v v^\dagger-\sigma$ has just a single positive eigenvalue, and an optimal solution is given by the rank-one projector onto this eigenspace.
\end{proof}
We also have the following related lemma.
\begin{lem}\label{lem:density_matrix_suffices}
     For any $d\in\bbZ_{>0}$, density matrix $\sigma\in \cD(\bbC^d)$, and unit vector $v \in\bbC^d$, it holds that \ba\label{eq:maximization_equivalence_density_matrices}\max_{0\preceq M\preceq \mathds{1}} \tr(M(v v^\dagger-\sigma))=\max_{\rho\in \cD(\bbC^d)} \tr(\rho(v v^\dagger-\sigma)).\ea Moreover, if $\sigma\neq vv^\dag$, then the optimum for the maximization in the right-hand side of \cref{eq:maximization_equivalence_density_matrices} is rank-one and unique.
\end{lem}
\begin{proof}
  The $\geq$ direction of \cref{eq:maximization_equivalence_density_matrices} holds because $\cD(\bbC^d)\subset \{M\in\bbC^{d\times d}: 0\preceq M\preceq \mathds{1}\}$.
  For the $\leq $ direction, the left-hand side of \cref{eq:maximization_equivalence_density_matrices} is maximized by a rank-one solution, due to \cref{lem:rank_1_suffices}, and this is also a feasible solution for the maximization on the right-hand side. This proves \cref{eq:maximization_equivalence_density_matrices}.
    
    The uniqueness part of the claim follows from the fact that, as shown in the proof of \cref{lem:rank_1_suffices}, if $\sigma\neq vv^\dag$ then the matrix $A\coloneq vv^\dag - \sigma$ has a single positive eigenvalue $\lambda_1(A) > 0$ which is the optimal value for the maximization. Indeed, let $P$ be the orthogonal projection onto this eigenspace, and set $Q\coloneq \mathds{1}-P$. Then for any $\rho\in\cD(\bbC^d)$ we have
    \begin{align}
        \tr(\rho A) &= \tr(\rho PA) + \tr(\rho QA)\leq \tr(\rho PA) = \lambda_1(A)\tr(\rho P)
    \end{align}
    and the right-hand side of the above is strictly less than $\lambda_1(A)$ unless $\rho=P$.
\end{proof}
The proof of the next lemma uses the above results as well as some of the facts in \Cref{sec:minimax}, following a similar reasoning to the proof of Lemma~3 in \cite{AKT24-state}.
\begin{lem}\label{lem:main_td_opt_problem_real}
  Let $k,d\in\bbZ_{>0}$ such that $k\leq d$ and $v \in\bbR^d$ be a unit vector such that $v_1\geq v_2\geq\dots\geq v_d> 0$.
  Then $T_k(v)$ is equal to
  \begin{align}\label{eq:t_k_opt_simplified}
    \begin{aligned}
\max & \quad \langle m,v\rangle^2-\norm{{m}}_{(k)}^2 \\
      \textnormal{s.t.} & \quad {m}\in\bbR^d                                       \\
                        & \quad \norm{{m}}_2 = 1.
    \end{aligned}
  \end{align}
  Moreover, the minimax problem defining $T_k(v)$ in \cref{eq:main_td_optimization_problem} has a saddle point $(\sigma^\star,M^\star)$ for some $\sigma^\star\in\cI_k$ and $M^\star = m^\star(m^\star)^\dag$ where $m^\star$ is the optimal solution to the maximization in \cref{eq:t_k_opt_simplified}.
\end{lem}
\begin{proof}
  If $k=d$ then $T_k(v)=0$ and the objective function in \cref{eq:t_k_opt_simplified} is just $m\mapsto \langle m, v\rangle^2 - 1$, which clearly has an optimal value of zero. Also, the statement about saddle points is trivial in this case. Hence, from now on we assume $k<d$. Let $\cM$ denote the set of Hermitian matrices $M\in \bbC^{d\times d}$ such that $0\preceq M\preceq \mathds{1}$ and $\cM_1$ the subset of such matrices which are rank-one projection operators. We have
  \begin{align}
    \max_{M\in \cM}\min_{\sigma\in\cI_k}\tr(M(vv^\dag -\sigma)) &= \min_{\sigma\in\cI_k}\max_{M\in\cM}\tr(M(vv^\dag-\sigma))\label{eq:original_minimax}\\
    &= \min_{\sigma\in\cI_k}\max_{\rho\in\cD(\bbC^d)}\tr(\rho(vv^\dag-\sigma))\label{eq:density_matrix_minimax}\\
    &=\max_{\rho\in\cD(\bbC^d)}\min_{\sigma\in\cI_k}\tr(\rho(vv^\dag-\sigma)),\label{eq:optimization_to_rewrite}
  \end{align}
  where in the first and third lines we used Sion's Minimax Theorem (\Cref{thm:fan_minimax}), and in the second line we used \Cref{lem:density_matrix_suffices}. By \Cref{lem:saddle_point_existence} there exists a saddle point $(\sigma^\star, \rho^\star)\in\cI_k\times \cD(\bbC^d)$ for the minimax in \cref{eq:density_matrix_minimax} such that
  \begin{align}\label{eq:max_at_eq_point}
    \tr(\rho^\star(v v^\dagger-\sigma^\star)) & = \max_{\rho\in \cD(\bbC^d)} \tr(\rho(v v^\dagger-\sigma^\star)).
  \end{align}
  Appealing once again to \Cref{lem:density_matrix_suffices}, the maximum in \cref{eq:max_at_eq_point} is uniquely attained by a rank-one projection operator because $k < d$, which implies $\sigma^\star\neq vv^\dag$. Hence, $\rho^\star\in\cM_1$ and the right-hand side of \cref{eq:optimization_to_rewrite} is equal to
  \begin{align}
   \max_{M\in \cM_1}\min_{\sigma\in\cI_k}\tr(M(v v^\dagger-\sigma))
     & = \max_{{m}\in\bbR^d:\ \norm{{m}}_2=1} \left\{\lvert \langle m, v \rangle \rvert^2 - \max_{\sigma\in\cI_k}\ \langle m, \sigma m \rangle\right\}\label{eq:rank_1_m_formulation} \\
     & = \max_{{m}\in\bbR^d:\ \norm{{m}}_2=1} \left\{\lvert \langle m,v\rangle\rvert^2 - \sum_{i=1}^k|m_i^\downarrow|^2\right\},\label{eq:max_with_complex_vectors}
  \end{align}
  where the second line follows from \Cref{lem:fidelity_equals_top_k}. This proves the first part of the claim.

  For the second part, it suffices to show that $(\sigma^\star, \rho^\star)$ is also a saddle point for the minimax in \cref{eq:original_minimax}. But this is clear by comparing the left-hand side of \cref{eq:rank_1_m_formulation} with that of \cref{eq:original_minimax}, which we have shown to be equal. Explicitly, we have
  \begin{align}
      \rho^\star \in \argmax_{M\in\cM} \min_{\sigma\in \cI_k}\tr(M(vv^\dag-\sigma))\quad \textnormal{and}\quad \sigma^\star\in \argmin_{\sigma\in\cI_k}\max_{M\in\cM}\tr(M(vv^\dag-\sigma))
  \end{align}
  where the second inclusion follows from \Cref{lem:rank_1_suffices}.
\end{proof}
The optimization in \cref{eq:t_k_opt_simplified} bears some resemblance to the computation of the proximal operator for the $k$-support norm, as studied in~\cite{AFS12}, and can be solved in a similar way. To begin executing this strategy, we require the following fact.
\begin{lem}\label{lem:opt_soln_norm_inner_product_condn}
  A point ${m}\in \bbR^d$ is an optimal solution for the optimization problem in \cref{eq:t_k_opt_simplified} with $k<d$ if and only if there exists a $\lambda\geq 0$ such that
  \begin{align}\label{eq:m_opt_condn}
    \norm{{u}}_{(k,*)}^2 + \norm{{m}}_{(k)}^2 & = 2\langle u, m\rangle\quad \textnormal{where}\quad {u}\coloneq \langle v , m\rangle v -\lambda {m}.
  \end{align}
\end{lem}
\begin{proof}
  We make use of basic facts from convex analysis, given in \Cref{sec:convex}. The optimization problem in \cref{eq:t_k_opt_simplified} is equivalent to minimizing $f({m})$ over all ${m}\in \Omega$, where $f({m})\coloneq-\frac{1}{2}\langle m , v\rangle^2 + \frac{1}{2}\norm{{m}}_{(k)}^2$ and $\Omega\coloneq \{{m}\in \bbR^d: \norm{{m}}_2 \leq 1\}$.
  Relaxing the norm constraint here is valid by homogeneity of the objective function in \cref{eq:t_k_opt_simplified}, along with the fact that the optimal value is greater than zero when $k < d$. By \cref{lem:fooc} the point ${m}$ is an optimal solution if and only if $0\in \partial f({m}) + \cN_{\Omega}({m})$, where $\partial f(m)$ is the subdifferential of $f$ at $m$ and $\cN_{\Omega}(m)$ is the normal cone at $m$. When ${m}$ is a unit vector, the normal cone $\cN_{\Omega}({m})$ is just the ray generated by the vector ${m}$. Also, the subdifferential can be computed as
  $
    \partial f({m}) = - \langle m , v\rangle v + \partial\big(\frac{1}{2}\norm{\cdot}_{(k)}^2\big)({m}).
  $
 Therefore the optimality condition for a given unit vector ${m}$ is equivalent to
  \begin{align}
    \langle v , m\rangle v-\lambda {m}\in \partial\big(\frac{1}{2}\norm{\cdot}_{(k)}^2\big)({m})
  \end{align}
  for some $\lambda\geq 0$. The equivalence between this condition and \cref{eq:m_opt_condn} is a direct consequence of \Cref{cor:dual_norm_fenchel_conj}.
\end{proof}

We are now ready to show that the optimal measurement vector ${m}$ has a simple functional form. An example of such an $m$ is given in \Cref{fig:m-example}.  This restricted form will allow us to efficiently compute it, as well as the optimal trace distance value. We first establish some notation.
\begin{dfn}\label{dfn:theta_def}
  Let $r<k<\ell\leq d$ be nonnegative integers and ${v}\in\bbC^d$ be a unit vector.
  For any $\lambda\geq 0$ we set
  \begin{align}
    \theta^{{v}}_{r,\ell}(\lambda) & \coloneq \frac{s_{k-r}^{{v}}-s_{\ell}^{{v}}}{r+1+(\ell-k+r)\lambda}
  \end{align}
  where $s^{{v}}_j$ is defined as in \cref{eq:s_defn}.
\end{dfn}
This definition excludes the case where $k=d$, since in that case the $v$ is already $d$-sparse and the trace distance is 0.
Also define $v_0\coloneq +\infty$ and $v_{d+1}\coloneq 0$.
\begin{lem}\label{lem:m_restricted_form}
  Let $k,d\in\bbZ_{>0}$ such that $k < d$ and $v \in\bbR^d$ be a unit vector such that $v_1\geq v_2\geq\dots\geq v_d> 0$. A solution ${m}\in\bbR^d$ to the optimization problem in \cref{eq:t_k_opt_simplified} is optimal if and only if:
  \begin{enumerate}
    \item ${m} = \frac{{\widetilde{m}}}{\norm{{\widetilde{m}}}_2}$ where
          \begin{align}\label{eq:tilde_m_defn}
            {\widetilde{m}_i} = \begin{cases}
                                  \frac{v_i}{1+\lambda}, & i=1,\dots,k-r-1    \\
                                  \theta_{r,\ell}^{{v}}(\lambda),   & i=k-r,\dots,\ell-1 \\
                                  \frac{v_i}{\lambda},   & i=\ell,\dots,d
                                \end{cases}
          \end{align}
          for some $r\in \{0,1,\dots,k-1\}$, $\ell\in\{k+1,\dots,d+1\}$, and $\lambda> 0$ such that
          \begin{align}\label{eq:theta_condns}
            \theta^{{v}}_{r,\ell}(\lambda)\in \frac{1}{1+\lambda}[v_{k-r}, v_{k-r-1})\cap \frac{1}{\lambda}(v_{\ell}, v_{\ell-1}],
          \end{align} and
    \item $\langle v , \widetilde{m}\rangle=1$.
  \end{enumerate}
  Moreover, for such a unit vector ${m}$, it holds that
  $
    \lambda = T_k(v).
  $
\end{lem}
\begin{proof}
  Since $k < d$ and ${v}$ has no entries equal to zero then complementary slackness implies that $\lambda>0$.  Specifically, since $\lambda$ is a Lagrange multiplier for the constraint $\norm{m}_2\leq 1$, if $\lambda=0$ then there is an $m$ with norm less than one which attains the optimal value. But this can only happen if that value is $0$. 
  
Next,  by \Cref{lem:opt_soln_norm_inner_product_condn}, ${m}$ is optimal for \cref{eq:t_k_opt_simplified} if and only if $\norm{{m}}_2=1$ and there is a $\lambda> 0$ such that $\norm{{u}}_{(k,*)}^2 + \norm{{m}}_{(k)}^2 = 2\langle u , m\rangle$ where ${u}=\langle v,m\rangle v -\lambda{m}$. Let us assume that $m$ is nonnegative and sorted in nonincreasing order; we will justify this assumption shortly. By \Cref{cor:dual_norm_fenchel_conj}, the unit vector ${m}$ is optimal for the problem
  \begin{align}\label{eq:another_opt_in_lem_proof}
    \begin{aligned}
\max & \quad \langle u, m^\prime\rangle - \frac{1}{2}\norm{m^\prime}_{(k)}^2 \\
      \textnormal{s.t.} & \quad m^\prime\in\bbR^d.
    \end{aligned}
  \end{align}
  Here, we can assume ${u}$ is also nonnegative and sorted without loss of generality. Indeed, if $u_j< u_{j+1}$ for some index $j\in [d-1]$, we must have $m_j = m_{j+1}$. Otherwise, ${m}$ is not optimal for the maximization in \cref{eq:t_k_opt_simplified}. But then $v_j < v_{j+1}$ from the definition of $u$, which contradicts the assumption that $v$ is sorted.
  
  Let us first prove necessity of the two conditions in \Cref{lem:opt_soln_norm_inner_product_condn}. Suppose $\ell$ is the least positive integer such that $u_\ell=0$, with $\ell\coloneq d+1$ if such an integer does not exist, and $r$ is the unique integer such that
  \begin{align}
    u_{k-r-1} > \frac{s^{{u}}_{k-r}}{r+1} \geq u_{k-r}.
  \end{align}
  Observe that we must have $\ell \geq k+1$; otherwise, $u_j = m_j$ for all $j\in [d]$ is an optimal solution to the problem in \cref{eq:another_opt_in_lem_proof}. But then this implies $u_j = \langle v,u\rangle v_j/(1+\lambda)$ for all $j\in [d]$, and therefore $v_{\ell}=0$, contradicting the assumption that ${v}$ has no entries equal to zero. Next, by the definition of $\ell$, we have $m_{j}=\langle v, m\rangle v_j/\lambda$ for every $j\in \{\ell,\dots, d\}$. By \Cref{lem:optimal_soln_to_k_supp}, ${m}$ must also satisfy $m_j = u_j$ for every $j\in \{1,\dots,k-r-1\}$ and $m_j= s^{{u}}_{k-r}/(r+1)$ for every $j\in \{k-r,\dots,\ell-1\}$. The first of these relations is equivalent to $m_j = \langle v, m\rangle v_j/(1+\lambda)$ for every $j\in\{1,\dots,k-r-1\}$. The second implies
  \begin{align}
    m_{k-r} & = \frac{1}{r+1}\left(\sum_{j=k-r}^{\ell-1}\langle v, m\rangle v_j - \lambda m_{j}\right) = \frac{1}{r+1}\left(\langle v, m\rangle(s_{k-r}^{{v}}-s_{\ell}^{{v}}) - \lambda (\ell-k+r)m_{k-r}\right)
  \end{align}
  which is true only if $m_{j}=\langle m,v\rangle\theta^{{v}}_{r,\ell}(\lambda)$ for each $j\in\{k-r,\dots, \ell-1\}$. Furthermore, $\theta_{r,\ell}^{{v}}(\lambda)$ must satisfy the condition in \cref{eq:theta_condns} for ${m}$ to be sorted in nonincreasing order. (Some of the inequalities are strict because of the definitions of $r$ and $\ell$.) Thus, we necessarily have ${m} = \langle v, m \rangle{\widetilde{m}}$ at the optimum, where ${\widetilde{m}}$ is of the same form as in \cref{eq:tilde_m_defn}. Additionally, since ${m}$ is a unit vector, we obtain $\langle v, m \rangle=\norm{{\widetilde{m}}}_2^{-1}$, which implies the second condition.

  We now return to the assumption that $m$ is sorted. We can justify this assumption in an identical manner to the final step in the proof of \Cref{lem:optimal_soln_to_k_supp} in \Cref{sec:k_supp_lem_proof}. Suppose there is an optimal $m$ for which $m_j < m_{j+1}$ for some $j\in [d-1]$. Then we must have $v_j = v_{j+1}$ and, consequently, the solution obtained by permuting the indices $j$ and $j+1$ in $m$ is also optimal; otherwise, the objective function in \cref{eq:t_k_opt_simplified} is less than that obtained by such a permutation. But this contradicts the necessary condition on the form of an optimal sorted $m$ derived above. For example, for $j\in [k-r-2]$ we have shown above that $m_j=m_{j+1}$ when $v_j=v_{j+1}$. Other cases may be handled similarly. 

  To see that the two conditions are sufficient, we compute
  \begin{align}
    s^{{u}}_{k-r}
     & = \sum_{i=k-r}^{\ell-1} \left(\langle v, m \rangle v_i - \lambda \frac{\widetilde{m}_i}{\norm{{\widetilde{m}}}_2}\right) & \\
     & = \langle v, m \rangle \left(s^{{v}}_{k-r} - s^{{v}}_\ell - \lambda (\ell-k+r)\theta_{r,\ell}^{{v}}(\lambda)\right)      & \\
     & = \langle v, m \rangle(r+1)\theta_{r,\ell}^{{v}}(\lambda) . \label{eq:s_k_r_u_expr}
  \end{align}
  Thus,
  \begin{align}
    u_{k-r}=\langle v, m \rangle(v_{k-r}-\lambda \theta_{r,\ell}^{{v}})\leq \langle v, m \rangle\frac{v_{k-r}}{1+\lambda}\leq \langle v, m \rangle\theta^{{v}}_{r,\ell}(\lambda)= \frac{s_{k-r}^{{u}}}{r+1}
  \end{align}
  where the first equality uses Condition $2$ in the statement of the lemma, both inequalities follow from \cref{eq:theta_condns} and the final equality is implied by \cref{eq:s_k_r_u_expr}. Similarly, from $\theta_{r,\ell}^{{v}}(\lambda)< \frac{1}{1+\lambda}v_{k-r-1}$, we have
  \begin{align}
    u_{k-r-1} & = \langle v, m \rangle\left(v_{k-r-1} - \lambda \frac{v_{k-r-1}}{1+\lambda}\right) = \langle v, m \rangle\frac{v_{k-r-1}}{1+\lambda} > \langle v, m \rangle\theta_{r,\ell}^{{v}}(\lambda) = \frac{s_{k-r}^{{u}}}{r+1}.
  \end{align}
  Noting that the constraints in \cref{eq:theta_condns} also ensure that ${m}$ is sorted in nonincreasing order, we may write
  \begin{align}
    \norm{{u}}_{(k,*)}^2 & = \sum_{i=1}^{k-r-1} u_i^2 + \frac{(s_{k-r}^{{u}})^2}{r+1}                                                                  \\
                         & = \langle v, m \rangle^2\left(\frac{\norm{{v_{1:k-r-1}}}_2^2}{(1+\lambda)^2} + (r+1)\theta_{r,\ell}^{{v}}(\lambda)^2\right) \\
                         & = \norm{{\widetilde{m}}}_2^{-2}\norm{{\widetilde{m}}}^2_{(k)}                                                               \\
                         & = \norm{{m}}_{(k)}^2,
  \end{align}
  where in the second line we used the fact that $u_i=m_i=\frac{\langle v, m \rangle v_{i}}{1+\lambda}$ for every $i\in [k-r-1]$. By the above and \Cref{lem:opt_soln_norm_inner_product_condn}, if $\norm{{m}}_{(k)}^2 = \langle u, m \rangle$ then ${m}$ is optimal. Since $m_i = \langle v, m \rangle\theta_{r,\ell}^{{v}}(\lambda)$ for $i\in \{k-r,\dots,\ell-1\}$ we have
  \begin{align}
    \langle u, m \rangle
     & = \sum_{i=1}^{k-r-1}m_i^2 + \langle v, m \rangle\theta_{r,\ell}^{{v}}(\lambda)s_{k-r}^{{u}}
    = \sum_{i=1}^{k-r-1}m_i^2 + \langle v, m \rangle^2(\theta_{r,\ell}^{{v}}(\lambda))^2(r+1) = \norm{{m}}_{(k)}^2
  \end{align}
  where the second equality follows from \cref{eq:s_k_r_u_expr}.

  Finally, it is easy to see that $\lambda=T_k(v)$ for an optimal solution ${m}$ of this form, since
  \begin{align}
    \langle v, m \rangle^2-\norm{{m}}_{(k)}^2 & = \langle v, m \rangle^2 - \langle u, m \rangle = \lambda.
  \end{align}
  This concludes the proof.
\end{proof}


\paragraph{Algorithm description.} Condition 2 in \Cref{lem:m_restricted_form}, which pertains to the norm of the vector ${\widetilde{m}}$ in that lemma, is equivalent to
\begin{align}\label{eq:td_norm_condn}
  1 & = \frac{\lVert v_{1:k-r-1}\rVert_2^2}{1+\lambda} + \frac{\norm{v_{k-r:\ell-1}}_1^2}{r+1+(\ell-k+r)\lambda} + \frac{\lVert v_{\ell:d}\rVert_2^2}{\lambda},
\end{align}
a cubic equation in terms of $\lambda$. This leads naturally to an algorithm for computing the trace distance and optimal measurement, given in \Cref{alg:td_approx}.

\algrenewcommand\algorithmicrequire{\textbf{Input:}}
\algrenewcommand\algorithmicensure{\textbf{Output:}}
\floatname{algorithm}{Algorithm}
\begin{algorithm}
  \caption{Optimal trace distance \& measurement}\label{alg:td_approx}
  \begin{algorithmic}[1]
    \Require Description of unit vector ${v}\in\mathbb{C}^d$; $k\in \{1,\dots, d\}$
    \Ensure $T_k(v)$ (optimal trace distance); ${m}$ (optimal distinguishing vector)
    \If{$k=d$}
    \State \Return $T_k\gets 0$ and ${m}\gets {v}$
    \EndIf
    \For{$r\in\{0,\dots,k-1\}$, $\ell\in\{k+1,\dots, d+1\}$}
    \State $\Lambda \gets$ positive solutions to the cubic equation in $\lambda$ implied by \cref{eq:td_norm_condn}
    \For{$\lambda\in \Lambda$}
    \If{\cref{eq:theta_condns} holds for these values of $r$, $\ell$, and $\lambda$}
    \State $T_k\gets \lambda$
    \State Set ${\widetilde{m}}$ as in \cref{eq:tilde_m_defn}
    \State ${m}\gets {\widetilde{m}}/\norm{{\widetilde{m}}}_2$
    \State Go to Line~\ref{lst:sort_line}
    \EndIf
    \EndFor
    \EndFor
    \State Re-order and re-introduce phases to ${m}$\label{lst:sort_line}
    \State \Return $T_k$ and ${m}$
  \end{algorithmic}
\end{algorithm}

A direct inspection of \Cref{alg:td_approx} leads to a runtime on the order of $d\log d + k(d-k)d$. The first term is due to the runtime of sorting the vector ${v}$. The second term is due to the $k(d-k)$ for-loop iterations. Within each of these iterations, the runtime is dominated by the computation of the constant (not $\lambda$) terms involved in \cref{eq:td_norm_condn}, each of which can be computed in time on the order of $d$. By a straightforward memoization of the partial sums and norms --- i.e., keeping track of the intermediate values of the sums between the for-loop iterations --- we can remove this order $d$ overhead in the second term. This proves \Cref{thm:td_optimal_value_computation}.

An interesting question is whether the runtime could be further improved to $\tilde O(d)$, perhaps by restricting the range of possible $r,\ell$.  However, the current runtime is already less than what is needed for other steps: the singular value decomposition is used for truncating entangled states or matrix-product states and needs time $d^3$ or $d^{2.37\ldots}$, and even the heuristic version of \Cref{thm:efficient_computation_of_weights} needs time $dk^2$.

\begin{figure}
  \centering
  \begin{subfigure}[b]{0.495\textwidth}
    \centering
    \includegraphics[width=\linewidth]{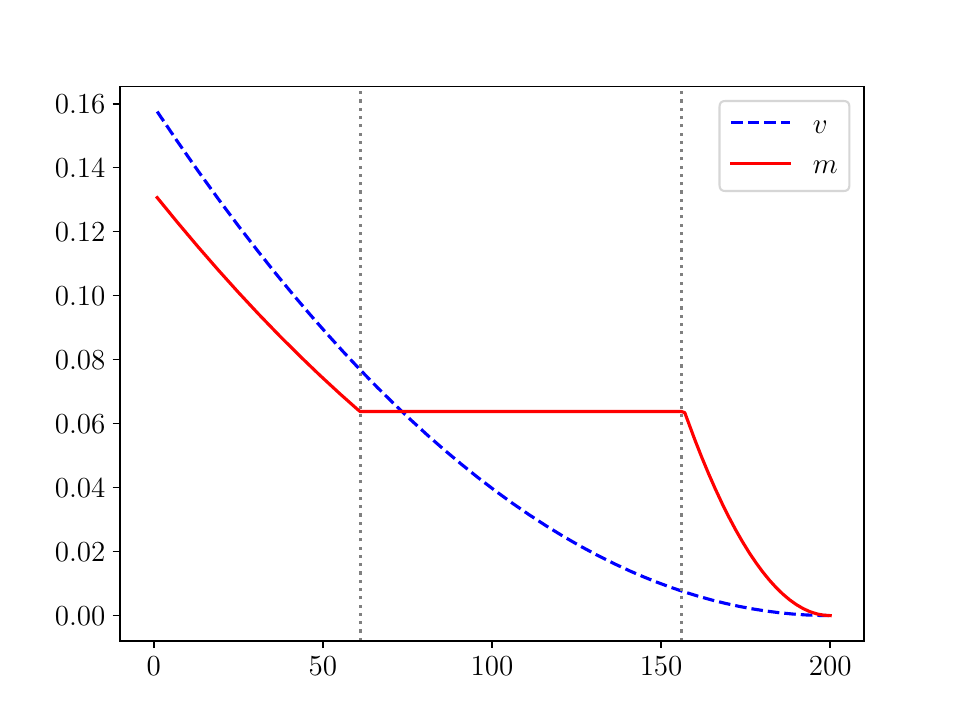}
    \caption{\label{fig:m-example}}
  \end{subfigure}
  \begin{subfigure}[b]{0.495\textwidth}
    \centering
    \includegraphics[width=\linewidth]{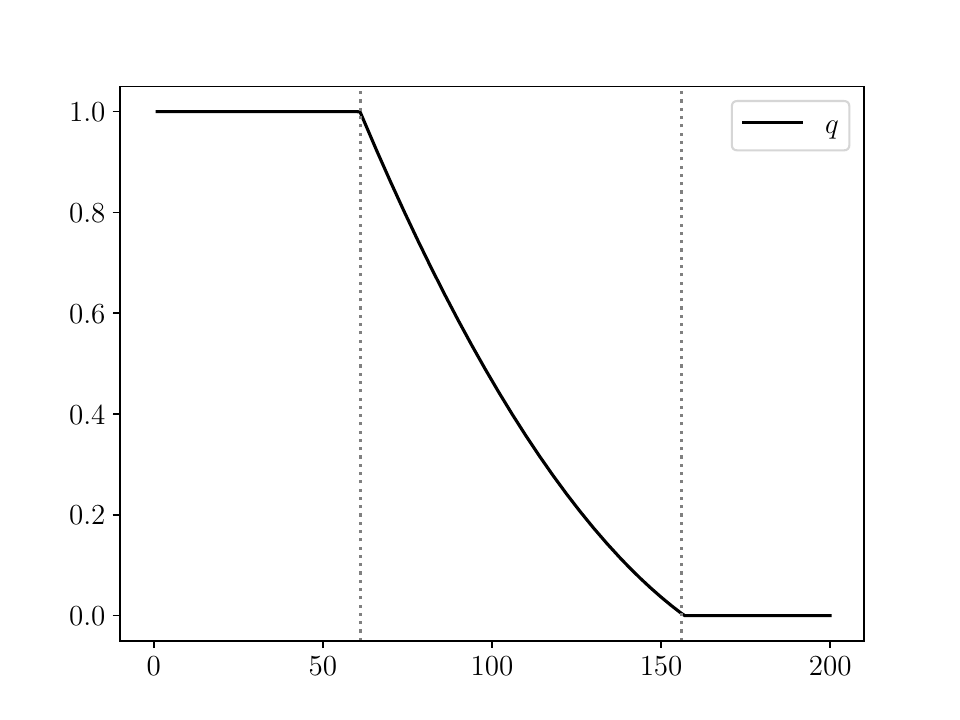}
    \caption{\label{fig:marginals}}
  \end{subfigure}
  \caption{With $d=200$ and $k=100$: (a) Example of a positive unit vector $v\in \bbR^d$ and the corresponding optimal $m\in\bbR^d$ derived in \cref{lem:m_restricted_form}. This $m$ has all three regions described in \cref{eq:tilde_m_defn}, but for other choices of $v$ and $k$, one might see only one or two of these regions. Vertical dotted lines are at $k-r$ and $\ell-1$, where in this case $r=39$ and $\ell=157$. (b) Marginal probability of drawing a $k$-sparse vector with a nonzero weight on the $i^\text{th}$ index, as a function of $i$, with the optimal ensemble of states described in \Cref{sec:optimal_density_matrix_td}. In this case, the first $60$ entries are kept deterministically.}
\end{figure}

\subsection{Subset sampling}
When we want to compute and sample the optimal density matrix, we need to solve a certain sampling problem.
We here describe this problem, and postpone a detailed discussion to \cref{sec:sampling}.
Consider a set $[n]=\{1,2,\dots,n\}$. We want to sample to subset of $\ell$ distinct integers from $[n]$.
What we will need for computing the density matrix and sampling from it, both for the optimal trace distance and robustness approximations, is the following.
In each case, the input is the alphabet size $n$, the size of the set $\ell$, and the marginal probabilities $q_1,\ldots,q_n$ that each element is included in the set. Normalization requires that $\sum_{j=1}^n q_j = \ell$.
Given these inputs, we want a distribution on subsets $S \subset [n]$, with $\abs{S} = \ell$ which can be represented efficiently and satisfies four desired criteria:
\begin{enumerate}
  \item marginal probability constraints are satisfied, i.e., $q_j=\sum_{S, j \in S}p(S)$ for every $j\in [n]$;
  \item two-element marginals
        $Q_{i,j}\coloneq \sum_{S, i,j \in S}p(S)$
        can be computed in time $\poly(n)$;
  \item the distribution can be sampled from in time $\poly(n)$; and
  \item negative correlation (of indicator random variables): $Q_{i,j}\leq q_iq_j$ for all $i,j\in [n]$.
\end{enumerate}
Note that there are ${n \choose \ell}$ subsets of size $\ell$, so it is not obvious that one can efficiently represent, find and sample from such a distribution.
Sampling subsets with constrained marginals is known as \emph{conditional Poisson sampling}.
In \cref{sec:sampling}, we explain that the \emph{maximum-entropy distribution} is a solution to this problem.
This distribution is a distribution of the form
\begin{align}
  p_\mu(S) &\coloneq \frac{\ee^{-\sum_{j \in S} \mu_j}}{Z}\ \forall S \subset [n], \abs{S} = \ell \quad Z_\mu\coloneq \sum_{S \subset [n], \abs{S}=\ell} \ee^{-\sum_{j\in S}\mu_j}
\end{align}
for some $\mu \in \bbR^n$, where the maximum-entropy distribution is given by $p^\star := p_{\mu^\star}$ with $\mu^\star$ the solution of an optimization problem. Note that these probability distributions are fully specified by $\mu$, so they have an efficient representation. To satisfy the requirements, we need to be able to perform the following algorithmic tasks:
\bit
\item
\emph{Inference:} given $q$, compute $\mu^\star$.
\item
\emph{Sampling:} given $\mu$, sample from $p_{\mu}$.
\item
\emph{Computing marginals:} given $\mu$, compute $Z_\mu$, $q_i$ and $Q_{ij}$.
\eit
We give efficient algorithms for these tasks in \cref{sec:sampling}.

\subsection{Computing and sampling the optimal mixed approximation}\label{sec:optimal_density_matrix_td}
In this section, we describe the density matrix $\sigma^\star\in\cI_k$ which attains the minimum value $T_k(v)$ in the optimization problem in \cref{eq:main_td_optimization_problem}, i.e., the optimal randomized truncation of $v$. We then show that this density matrix corresponds to an ensemble of $k$-sparse pure states which can be efficiently sampled, upon obtaining $T_k(v)$ and $m$ from \Cref{alg:td_approx}. The ensemble consists of pure states which are proportional to the vector obtained by keeping the top $k-r-1$ entries of $v$ deterministically, truncating indices greater than $\ell-1$, and drawing $r+1$ nonzero indices at random from the remaining $\ell-k+r$ indices. An example of the marginal inclusion probabilities $q_j$ for $j=k-r,\dots,\ell-1$ is given in \Cref{fig:marginals} and, as we will see in the proof of \Cref{thm:opt_td_state}, it holds that
\begin{align}
  q_j = \frac{v_j}{\theta} - \lambda\quad \textnormal{for every}\quad j\in\{k-r,\dots,\ell-1\}
\end{align}
where $\lambda=T_k(v)$ and $\theta=\theta_{r,\ell}^v(\lambda)$. The algorithm can also straightforwardly be applied to approximate an entangled quantum state by sampling states having Schmidt-rank at most $k$.

In Theorems~\ref{thm:opt_td_state} and \ref{thm:opt_td_sampling} below, the stated runtimes inherit the runtime of the inference problem for the maximum-entropy sampling. Hence, $\text{poly}(d,\log(1/\eta))$ can be further refined using any of the runtimes quoted in~\Cref{tab:inference}, e.g., $\widetilde{O}(d^{1.5}k)$ for constant $\eps$ using the accelerated projected gradient descent.
\begin{thm}\label{thm:opt_td_state}
  Let $k,d\in\bbZ_{>0}$ satisfy $k \leq d$, let $\eta\in (0,1)$, and let $v\in\bbC^d$ be a unit vector. There is an algorithm, requiring at most $\textnormal{poly}(d,\log(1/\eta))$ steps, which computes the entries in the standard basis of a density matrix $\sigma\in\cI_k$ such that
  $
    T(v,\sigma)\leq T_k(v) + \eta.
  $
\end{thm}
\begin{proof}
  We use the shorthand $T_k=T_k(v)$ and assume that ${v}$ is real, nonnegative, and sorted in nonincreasing order. As explained at the end of \Cref{sec:distance_measures}, the relevant phases and ordering can easily be handled upon solving the problem with this simplifying assumption. We make use of the minimax equality
  \begin{align}
    T_k = \min_{\sigma\in \cI_k}\max_{x\in\bbR^d: \norm{x}_2=1} \langle x, v\rangle^2 - \langle x, \sigma x\rangle = \max_{x\in\bbR^d: \norm{x}_2=1}\min_{\sigma\in \cI_k}\ \langle x, v\rangle^2 - \langle x, \sigma x\rangle.
  \end{align}
  The first equality is a direct consequence of \Cref{lem:rank_1_suffices} while the second was established in the proof of \Cref{lem:main_td_opt_problem_real}. Let $(\sigma^\star, m)$ be a saddle point for the this minimax problem. Observe that we must have
  \begin{align}\label{eq:sigma_star_maximizes}
    \max_{\sigma\in\cI_k} \ \langle m,\sigma m\rangle = \langle m, \sigma^\star m\rangle.
  \end{align}
  Otherwise,
  \begin{align}
      T_k = \min_{\sigma\in\cI_k}\ \langle m, v\rangle^2-\langle m,\sigma m\rangle = \langle m, v\rangle^2-\max_{\sigma\in\cI_k}\ \langle m,\sigma m\rangle < \langle m, v\rangle^2 - \langle m, \sigma^\star m\rangle
  \end{align}
  which contradicts $(\sigma^\star, m)$ being a saddle point.
  \cref{eq:sigma_star_maximizes} implies that $\sigma^\star$ is a mixture of pure states attaining the maximum in
  \begin{align}\label{eq:condition on w}
    \max_w\{|\langle m, w\rangle|^2: \norm{w}_2=1,\ \norm{w}_0\leq k\}.
  \end{align}
  Additionally, by \Cref{lem:m_restricted_form}, $m\propto \widetilde{m}$ with
  \begin{align}\label{eq:opt_m_form_again}
    \widetilde{m}_i & = \begin{cases}
                          \frac{v_i}{1+\lambda}, & i=1,\dots,k-r-1     \\
                          \theta,                & i=k-r,\dots,\ell-1, \\
                          \frac{v_i}{\lambda},   & i=\ell,\dots,d
                        \end{cases}
  \end{align}
  for some $r\in \{0,\dots,k-1\}$, $\ell\in \{k+1,\dots,d+1\}$, and $\theta\in\bbR$ such that $v_{k-r-1}/(1+\lambda) > \theta > v_{\ell}/\lambda$. Here, $\lambda=T_k$, and all these quantities can be efficiently computed. Eqs.~\eqref{eq:condition on w} and \eqref{eq:opt_m_form_again} then imply that each state in the mixture comprising $\sigma^\star$ is of the form
  \begin{align}\label{eq:optimal_k_sparse_state_form_trace_distance}
    w(S)\coloneq \frac{1}{\norm{\widetilde{m}}_{(k)}}\left(\sum_{i=1}^{k-r-1}\widetilde{m}_ie_i + \theta\sum_{j\in S} e_j\right)
  \end{align}
  for some $S\subseteq \{k-r,\dots,\ell-1\}$ satisfying $|S|=r+1$. Denote by $\Omega$ the collection of all such $S$. Then we can write $\sigma^\star = \bbE_S w(S)w(S)^\dag$ for some set-valued random variable $S\in \Omega$. Next, observe that
  \begin{align}
    T_k & = \max_{x: \norm{x}_2=1} x^\dag \left(vv^\dag - \sigma^\star\right)x
  \end{align}
  and $m$ is an optimal solution for this problem, which leads to the eigenvector equation
  $
    (vv^\dag - \sigma^\star)m = T_km.
  $
  Consequently,
  \begin{align}
    T_k\widetilde{m} & = \frac{T_k m}{\langle m, v\rangle}                                                                      \\
                     & = \frac{1}{\langle m, v\rangle}\left[\langle v, m\rangle v - \bbE_S \ \langle w(S), m\rangle w(S)\right] \\
                     & = v-\bbE_S \left[\sum_{i=1}^{k-r-1}\widetilde{m}_ie_i + \theta\sum_{j\in S} e_j\right]                   \\
                     & = v - \sum_{i=1}^{k-r-1}\widetilde{m}_ie_i - \theta\sum_{j=k-r}^{\ell-1}q_j e_j
  \end{align}
  where $q_j\coloneq \Pr_{S}\left[j\in S\right]$ is the marginal probability that the index $j$ is included in the random subset $S$, for each $j\in\{k-r,\dots,\ell-1\}$. Hence, for each $j\in \{k-r,\dots,\ell-1\}$ it holds that $T_k\theta = v_j - \theta q_j$, or equivalently,
  \begin{align}\label{eq:td_marginals}
    q_j = \frac{v_j}{\theta} - T_k.
  \end{align}
  One may verify that $\sum_{j=k-r}^{\ell-1} q_j = r+1$ by making use of the fact that $\theta=\theta^v_{r,\ell}(\lambda)$, as in \Cref{dfn:theta_def}, which implies these are valid single-element marginal probabilities. 

  We have successfully derived the fact that $\sigma^\star$ is a mixture of pure states $w(S)$ where $S$ is a random subset of $\{k-r,\dots,\ell-1\}$ of fixed size $r+1$, whose marginal inclusion probabilities satisfy \cref{eq:td_marginals}. Since each quantity on the right-hand side of that equation can be efficiently computed (see \Cref{alg:td_approx}) we can also compute these marginal probabilities efficiently.

  Next, we claim that sampling $S$ from the max-entropy distribution $p^\star$ subject to these marginal constraints results in an optimal density matrix $\sigma^\star\in\cI_k$.
  This follows from the observation that for any two distributions $p_1,p_2: \Omega\to [0,1]$ with the same single-element marginal probabilities, we necessarily have
  \begin{align}
    \left(\bbE_{S_1\sim p_1} w(S_1)w(S_1)^\top - \bbE_{S_2\sim p_2}w(S_2)w(S_2)^\top\right)m & = \norm{m}_{(k)} \left( \bbE_{S_1\sim p_1} w(S_1) - \bbE_{S_2\sim p_2} w(S_2) \right) = 0,
  \end{align}
  and this holds in particular for $p_1 = p^\star$ and $p_2$ any choice of distribution over the subsets $S\in \Omega$ which results in an optimal density matrix.

  It remains to argue that we can compute the entries of an approximation $\sigma$ efficiently, and this is a consequence of the algorithms in \cref{sec:sampling} for conditional Poisson sampling.
  We saw that taking $S$ to have the max-entropy distribution $p^\star$ subject to the marginal constraints given by the $q_j$, results in a density matrix $\sigma^\star = \bbE_{S\sim p^\star}w(S)w(S)^\top$ which is optimal. By \Cref{thm:efficient_computation_of_weights}, since $\ell-k+r\leq d$, using at most $\textnormal{poly}(d, \log(1/\eta))$ operations we can find a set of weights $\mu\in\bbR^{\ell-k+r}$, specifying a distribution $p_\mu$ such that $\mathrm{d}_{\textnormal{TV}}(p^\star,p_\mu)\leq \eta$. Setting $\sigma\coloneq \bbE_{S\sim p_\mu}w(S)w(S)^\top$, this implies $T(\sigma, \sigma^\star)\leq \eta$ and, by the triangle inequality, $T(v,\sigma)\leq T_k + \eta$. Moreover, setting $y\coloneq m_{1:k-r-1}$ we have
  \begin{align}
    \sigma = \frac{1}{\norm{m}_{(k)}^2}\left[yy^\top + \theta \sum_{i=k-r}^{\ell-1}q(\mu)_i(ye_i^\top + e_iy^\top) + \theta^2\sum_{j, j^{\prime}=k-r}^{\ell-1} Q(\mu)_{jj^\prime} e_je_{j^\prime}^\top\right]
  \end{align}
  where $q(\mu)_i \coloneq \Pr_{S\sim p_\mu}[i\in S]$ and $Q(\mu)_{ij}\coloneq \Pr_{S\sim p_\mu}\left[i\in S \land j\in S\right]$ for every $i,j\in \{k-r,\dots,\ell-1\}$ are the single- and two-element marginal inclusion probabilities. By \Cref{lem:marginal_computation}, each term in the square brackets can be computed efficiently, and the claim follows.
\end{proof}

\begin{thm}\label{thm:opt_td_sampling}
  Let $\sigma\in\cI_k(\bbC^d)$ be the $\eta$-accurate approximation to the optimal density matrix obtained by the algorithm in \Cref{thm:opt_td_state}. There is an algorithm running in time $\textnormal{poly}(d,\log(1/\eta))$ which samples a random $k$-sparse pure state $w\in \bbC^d$ from a distribution such that $\bbE ww^\dag = \sigma$. Moreover, subsequent samples can be obtained in $O(d)$ steps.
\end{thm}
\begin{proof}
  Take the ensemble of pure states to be the random pure state $w(S)$ defined in \cref{eq:optimal_k_sparse_state_form_trace_distance}, where $S$ is a random subset of indices having a distribution which is $\eta$-accurate in total variation distance to the max-entropy distribution. As stated in the proof of \Cref{thm:opt_td_state}, it takes $\textnormal{poly}(d,\log(1/\eta))$ operations to find the parameters $\mu$ corresponding to this distribution. The claim follows from \Cref{lem:marginal_computation}, which implies that it takes $O(d)$ steps to sample from that distribution once $\mu$ is found.
\end{proof}

\subsection{Coherence versus entanglement for the trace distance}\label{sec:coherence_vs_entanglement}
In this section we show that for approximating pure states by density matrices with low Schmidt number, the results on the $k$-sparse approximations carry over directly.
Recall that given any state $\rho \in \cD(\bbC^a \ot \bbC^b)$, the optimal approximation by a state with Schmidt rank at most $k$ is the following optimization problem:
\begin{align}\label{eq:ent_min_problem}
  T_k^{(E)}(\rho) := \min_{\sigma \in \cS_k} T(\rho,\sigma).
\end{align}
We prove the following.
\begin{lem}
  Let ${v} \in \bbC^a \ot \bbC^b$ have a Schmidt decomposition
  \begin{align}
    {v} = \sum_{i=1}^r \lambda_i\; {x_i} \ot {y_i}\quad \textnormal{with}\quad \lambda=(\lambda_1,\dots,\lambda_r)^\top.
  \end{align}
  It holds that $T_k^{(E)}(v) = T_k(\lambda)$. Moreover, an optimal solution $\sigma^\star\in\cS_k$ to the minimization in \cref{eq:ent_min_problem} is given by applying the linear map with the action $e_i\mapsto x_i\otimes y_i$ for each $i\in [r]$ to the optimum in the problem defining $T_k(\lambda)$.
\end{lem}

\begin{proof}
  Without loss of generality, choose the Schmidt basis to be the standard basis, so
  \begin{align}
    {v} = \sum_{i=1}^r \lambda_i\; e_i \ot e_i,
  \end{align}
  and define the $k$-sparse states with respect to the basis of states $e_i \ot e_i$.
  First of all it is clear that $T_k^{(E)}(v) \leq T_k(v)$ since, with this definition of sparsity, $\cI_k \subseteq \cS_k$.

  In order to conclude that we have equality, we need to show that we can restrict the optimization over $\cS_k$ which defines $T_k^{(E)}(v)$ to $\cI_k$. To this end, observe that the argument of \cref{lem:main_td_opt_problem_real} applies with $\cI_k$ replaced by $\cS_k$ without alteration, which shows that
  \begin{align}
    T_k^{(E)}(v) & = \max_{M: \norm{M}_{\infty} \leq 1} \min_{\sigma \in \cS_k} \tr\left(M (v v^\dagger - \sigma)\right) \\
                 & = \max_{m: \norm{m}_2 \leq 1} \abs{\langle m, v\rangle}^2 - \norm{m}_{(E,k)}^2,
  \end{align}
  where $\norm{m}_{(E,k)}^2 \coloneq \max_{\sigma \in \cS_k} \tr(m m^\dagger \sigma)$ equals the sum of the squares of the $k$ largest Schmidt coefficients of ${m}$, by \Cref{lem:fidelity_ent_equals_top_k}.
  Once we show that the optimal ${m}$ is of the form ${m} = \sum_i m_i e_i \ot e_i$, we can conclude that $T_k^{(E)}({v}) = T_k({v})$, since then it is clear we can optimize over $\cI_k$ instead of $\cS_k$.
  Consider an arbitrary ${m} = \sum_{ij} m_{ij} e_i \ot e_j$.
  We denote by $e_{ij}$ the matrix $e_i e_j^\top$.
  Let $P = \sum_i e_{ii} \ot e_{ii}$ and let ${\tilde m} = \sum_{i} m_{ii} e_i \ot e_i = P {m}$.
  Let $\Phi:\cL(\bbC^a\ot \bbC^b)\to\cL(\bbC^a\ot\bbC^b)$ be the quantum channel with the action $\Phi: X\mapsto \bbE U X U^\dag$ where $U\in\cL(\bbC^a\ot\bbC^b)$ is a random unitary operator of the form
  \begin{align}
    U=\sum_{i,j} z_i z_j^\star e_{ii} \ot e_{jj}
  \end{align}
  with uniformly random phases $z_i\in\bbC$, $|z_i|=1$. This implies that $\Phi$ acts on matrix elements through
  \begin{align}
    e_{ij} \ot e_{kl} \mapsto \begin{cases}
                                e_{ij} \ot e_{kl} & \text{if } i = j \text{ and } k = l  \text{ or } i=k \text{ and } j=l \\
                                0                 & \text{otherwise.}
                              \end{cases}
  \end{align}
  Furthermore, the channel is self-adjoint $\Phi = \Phi^\dagger$ and preserves Schmidt number.
  Acting on the rank-1 projection operator $mm^\dag$ gives
  \begin{align}
    \Phi(m m^\dagger) = \tilde m \tilde m^\dagger + \sum_{i \neq j} \abs{m_{ij}}^2 e_{ii} \ot e_{jj} \eqcolon \tilde m \tilde m^\dagger + B
  \end{align}
  for some operator $B$ which is positive semidefinite.
  In particular, for any ${w}\in\bbC^a\ot\bbC^b$ with Schmidt rank at most $k$, we have
  \begin{align}
    |\langle \tilde m, w\rangle|^2 = \tr\left(m m^\dagger \Phi(w w^\dagger)\right) - w^\dagger B w \leq \norm{m}_{(E,k)}^2
  \end{align}
  where the inequality follows from the fact that $\Phi(ww^\dag)\in\cS_k$ and the second term is at most zero. Maximizing over ${w} \in \cS_k$ we arrive at $\norm{\tilde m}_{(E,k)} \leq \norm{m}_{(E,k)}$. Since $\abs{\langle m, v \rangle}^2 = \abs{\langle \tilde m, v \rangle}^2$ we conclude that
  \begin{align}
    \abs{\langle m, v \rangle}^2 - \norm{m}_{(E,k)}^2 \leq \abs{\langle \tilde m, v \rangle}^2 - \norm{\tilde m}_{(E,k)}^2
  \end{align}
  and hence we can maximize over ${m}$ of the form ${m} = \sum_i m_i e_i \ot e_i$.
\end{proof}

\subsection{Variance estimates for ensembles}\label{sec:variance}
A randomized truncation yields a $k$-sparse state ${w(S)}$ with some probability. Given an observable $M$ (along with a method to compute its expectation with respect to a target state) one may use the random variable $w(S)^\dag Mw(S)$ as an estimator for the original expectation value of interest, incurring some bias and variance in the process. (For a deterministic truncation, there is no variance.) The trace distance value bounds the bias, while the variance can be related to the quality of an approximation using the individual states ${w(S)}$ comprising the ensemble. In this section, we record bounds on the variance in such a process, which may inform the feasibility of using these methods in numerical simulations, e.g., using MPS. (See \Cref{sec:MPS}.)  In fact, these bounds hold for generic ensembles of pure states.

Recall the trace distance is given by
\begin{align}
  T(v, \sigma) = \max_{\norm{M}_{\infty} \leq 1} \tr\left(M(v v^\dagger - \sigma)\right) = \max_{\norm{M}_{\infty} \leq 1} \E \tr\left(M(v v^\dagger - w(S) w(S)^\dagger)\right).
\end{align}
We have the following.
\begin{lem}\label{lem:variance}
  Let $S$ be a random variable, let $v$ be an arbitrary pure state, and let $w(S)$ be a random pure state such that $\E w(S)w(S)^\dag = \sigma$.
  \begin{enumerate}
    \item The expected value of the trace distance between $w(S)$ and $v$ is bounded by
          \be
          \E T(v, w(S)) \leq \sqrt{T(v,\sigma)}.
          \ee
    \item Let $M$ be any observable with $\norm{M}_{\infty} \leq 1$. It holds that
          \be
          \Var[w(S)^\dag M w(S)]\leq T(v,\sigma)\left(1 + \sqrt{T(v,\sigma)}\right)^2.
          \ee
  \end{enumerate}
\end{lem}

\begin{proof}
  We can bound the second moment of the random variable $T(v,w(S))$ through
  \begin{align}
      \E T(v,w(S))^2 &= 1-\E |\langle v,w(S)\rangle|^2 = 1-F(v,\sigma)^2\leq T(v,\sigma)
  \end{align}
  where the second equality uses $\sigma=\E w(S)w(S)^\dag$ and the inequality follows from \cref{eq:fvdg_pure}. Then the first bound in the lemma is a consequence of the relationship between $L_1$ and $L_2$ norms for random variables. For the second bound, note that
  \begin{align}
      \left|\tr\left( M(v v^\dagger - w(S) w(S)^\dagger)\right)\right|&\leq T(v, w(S))
  \end{align}
  which implies
  \begin{align}\label{eq:L_2_norm_bound}
      \E \tr\left( M(v v^\dagger - w(S) w(S)^\dagger)\right)^2 \leq \E T(v,w(S))^2\leq T(v,\sigma).
  \end{align}
  Let $\norm{X}_{L_2}$ denote the $L_2$ norm of a random variable $X$. By the triangle inequality and \cref{eq:L_2_norm_bound},
  \begin{align}
      \norm{w(S)^\dag M w(S) - \Tr(\sigma M)}_{L_2} &\leq \left|\Tr(M(\sigma - vv^\dag))\right| + \sqrt{T(v,\sigma)}.
  \end{align}
  The claim follows upon noting that the first term on the right-hand side above is at most $T(v,\sigma)$, and $\E w(S)^\dag M w(S) = \Tr(\sigma M)$.
\end{proof}

These bounds are fairly constraining: suppose that $F_k(v) = \sqrt{1 - \eps}$ and we are in a case where randomized truncation is close to the optimal value allowed by the Fuchs-van de Graaf inequalities, i.e., $T(v, \sigma) = O(\eps)$. Then the above argument shows that \emph{on average} the trace distance $T(v, w(S))$ must be $O(\sqrt{\eps})$. But we also know that the optimal pure state approximation has trace distance $\sqrt{\eps}$, so the states ${w(S)}$ must be close to the optimal deterministic approximations with high probability. Additionally, in this ``good" regime the standard deviation of the estimator is at most $O(\sqrt{\eps})$, which is on the same order as the error from the best deterministic truncation even after a single sample. Obtaining multiple samples using the randomized truncation then decreases the variance, from this already-competitive starting point.

\section{Approximations in robustness}\label{sec:robustness}
Recall that a simple expression for the robustness with $\cI_k$ of a pure state is derived in~\cite{Regula18,kcoh18}. (See \Cref{thm:Rk-value}.) Namely, for any unit vector $v\in\bbC^d$, we have
$
  R_k(v) = {\lVert v \rVert}_{(k,*)}^2-1.
$
By the characterization of the $k$-support norm from \Cref{lem:k_supp_efficient_expr}, this value is efficiently computable~\cite{AFS12}, in time $O(d\log d)$.
In this section, we show that the optimal density matrix in robustness is also efficiently computable and samplable, via the same max-entropy construction used in \Cref{sec:optimal_density_matrix_td}.
\subsection{Computing and sampling the mixed approximation in robustness}
The following result can be shown using a simpler version of the construction leading to \Cref{thm:opt_td_state}. As in \Cref{sec:optimal_density_matrix_td}, below in Theorems~\ref{thm:robust_tau} and \ref{thm:robust_ensemble} the runtime $\text{poly}(d,\log(1/\eta))$ can be further refined using any of the runtimes quoted in~\Cref{tab:inference}.
\begin{thm}\label{thm:robust_tau}
  Let $k,d\in\bbZ_{>0}$ satisfy $2 \leq k \leq d$, let $\eta\in (0,1)$, and let $v\in\bbC^d$ be a unit vector. There is an algorithm, requiring at most $\textnormal{poly}(d,\log(1/\eta))$ steps, which computes the entries in the standard basis of a density matrix $\tau\in\cI_k$ such that
  $
    T(\tau,\tau^\star)\leq \eta
  $
  where $\tau^\star\in\cI_k$ satisfies
  \begin{align}
    \tau^\star = \frac{vv^\dag + R_k(v)\sigma}{1+R_k(v)}\quad \textnormal{for some}\quad \sigma\in\cI_2.
  \end{align}
\end{thm}
\begin{proof}
  We use the shorthand $R_k = R_k(v)$ and, as in the proof of \Cref{thm:opt_td_state}, assume that $v$ is real, nonnegative, and sorted in nonincreasing order.

  Let $\cS$ denote the set of all subsets  $S\subseteq \{k-r,\dots, d\}$ satisfying $|S|=r+1$.
For any $S\in \cS$, define the state
  \begin{align}\label{eq:us}
    u(S) \coloneq \sum_{i=1}^{k-r-1}v_i e_i + \frac{s_{k-r}^v}{r+1}\sum_{j\in S}e_j .
  \end{align}
  These states are not generally unit vectors; in fact, $\langle u(S), u(S)\rangle = 1+R_k$ for any $S\in \cS$.  We use them to construct  the (not necessarily normalized) mixture
  \be \widetilde{\tau} = \bbE_S u(S)u(S)^\dag \ee
for a yet-to-be-specified distribution over $\cS$.  Then $\tr(\widetilde{\tau})=1+R_k$ and 
  \begin{align}
    \widetilde{\tau} & = v_{1:k-r-1}v_{1:k-r-1}^\top + \frac{s^v_{k-r}}{r+1}\sum_{i=k-r}^d q_i(v_{1:k-r}e_i^\top + e_i v_{1:k-r}^\top) + \left(\frac{s^v_{j-r}}{r+1}\right)^2\sum_{j,j^\prime=k-r}^dQ_{j,j^\prime}e_je_{j^\prime}^\top
  \end{align}
  where $q_i=\Pr_{S}[i\in S]$ and $Q_{ij}\coloneq \Pr_{S}\left[i\in S \land j\in S\right]$ for every $i,j\in \{k-r,\dots,d\}$, are the single- and two-element marginal inclusion probabilities.
  Now, suppose that the distribution over $S\in \cS$ is the max-entropy distribution whose single-element marginals are given by $q_i = v_i(r+1)/s_{k-r}^v$. Evidently, $\sum_{i=k-r}^d q_i = r+1$, so this is a valid set of marginal probabilities. Then
  \begin{align}
\Delta \coloneq    \widetilde{\tau}-vv^\dag
     & = \left(\frac{s_{k-r}^v}{r+1}\right)^2\sum_{i,j=k-r}^d\left(Q_{ij}-q_iq_j\right)e_ie_j^\top.
  \end{align}
  Since, for every $i=k-r,\dots,d$ we have $Q_{ii}=q_i$, the matrix $\Delta$ has diagonal elements
  \be \Delta_{ii} = v_i\qty(\frac{s_{k-r}^v}{r+1} - v_i),\ee
  which are positive by the condition in \cref{eq:condition k support norm}. On the other hand, a basic property of the max-entropy distribution (see \Cref{lem:negative_correlation} below) is that $Q_{ij}< q_iq_j$ for every $i,j\in \{k-r,\dots,d\}$ with $i\neq j$. Finally, we can compute the $i^\text{th}$ row sum of $\Delta$ as
  \begin{align}
    \sum_{j=k-r}^d\Delta_{i,j} \propto
    \sum_{j=k-r}^d(Q_{ij}-q_iq_j) & = \bbE_{S}\mathbf{1}_{i\in S}\left(\sum_{j=k-r}^d\mathbf{1}_{j\in S}\right) - q_i(r+1) = 0.
  \end{align}
  In summary, $\Delta$ has (i) nonnegative diagonals, (ii) nonpositive off diagonals, and (iii)  rows summing to 0.  Such matrices are known as diagonally dominant and  are proportional to an element of $\cI_2$, by Lemma~1 in \cite{kcoh18}. Thus, we have established:
  \begin{align}
    \widetilde{\tau} = vv^\dag + R_k\sigma\qquad \textnormal{for some} \qquad \sigma\in \cI_2.
  \end{align}
  Setting $\tau^\star = \widetilde{\tau}/(1+R_k)$ gives the desired optimal density matrix. Of course, we cannot compute exactly the one- and two-element marginals corresponding to the max-entropy distribution. But we can do so for an approximation that is $\eta$ close in total variation distance in time $\textnormal{poly}(d,\log(1/\eta))$. This leads to an approximation $\tau$ to $\tau^\star$ which is $\eta$-close in trace distance, using Lemmas~\ref{lem:marginal_computation} and \ref{thm:efficient_computation_of_weights}, in exactly the same way as in the proof of \Cref{thm:opt_td_state}.
\end{proof}

 The construction above also tells us how to efficiently sample from the ensemble corresponding to the robustness-optimal density matrix, as in the case for the trace distance.
\begin{thm}\label{thm:robust_ensemble}
  Let $\tau\in\cI_k(\bbC^d)$ be the $\eta$-accurate approximation to the optimal density matrix obtained by the algorithm in \Cref{thm:robust_tau}. There is an algorithm running in time $\textnormal{poly}(d,\log(1/\eta))$ which samples a random $k$-sparse pure state $u\in \bbC^d$ from a distribution such that $\bbE uu^\dag = \sigma$. Moreover, subsequent samples can be obtained in $O(d)$ steps.
\end{thm}
\begin{proof}
  The proof is identical to the proof of \Cref{thm:opt_td_sampling} upon taking the random pure state to be proportional to $u(S)$, defined in~\cref{eq:us}, rather than $w(S)$.
\end{proof}

\begin{remark}
  Our construction also yields a simple alternate proof of the fact from \cite{Regula18,kcoh18} that $\norm{v}_{k,*}^2 = 1 + R_k(v) = 1 + R_k^s(v)$ where $R_k^s(v)$ is the smallest $t\geq 0$ for which $\exists \sigma \in \cI_k$ such that $\frac{vv^\dag + t \sigma}{1+t} \in \cI_k$.    It is straightforward to show that
  \be
\norm{v}_{(k,*)}^2 \leq 1+ R_k(v) \leq 1+R_k^s(v).
\label{eq:k-supp-rob}  \ee
The second inquality is because $R_k^s$ imposes additional restrictions on $\sigma$ that do not appear in $R_k$.  For the first inequality, from the definition of $\norm{\cdot}_{(k,*)}$ there exists a vector $m$ with $\langle v,m\rangle = \norm{v}_{(k,*)}$ and $\norm{m}_{(k)} \leq 1$.  And from the definition of $R_k(v)$, there exists $\tau\in\cI_k$ with $vv^\dag \leq (1+R_k(v))\tau$.    Since $\tau$ is a mixture of $k$-sparse states, $\langle m, \tau m\rangle \leq 1$.  Combining these facts yields the first inequality in \cref{eq:k-supp-rob}.

Finally,  \Cref{thm:robust_tau} constructs $\sigma\in \cI_2 \subseteq \cI_k$ and $\tau\in \cI_k$ such that
\be
vv^\dag + R\sigma = (1+R)\tau,
\label{eq:rob-decomp}\ee
where $R = R_k(v)$.  Examining the proof of \Cref{thm:robust_tau}, we see that in fact it proves \cref{eq:rob-decomp} for $R = \norm{v}_{(k,*)}^2-1$.  Since \cref{eq:rob-decomp} is also a valid decomposition for $R_k^s$, this proves that $R_k^s(v) \leq \norm{v}_{(k,*)}^2-1$.  Thus all the inequalities in \cref{eq:k-supp-rob} are equalities.

Our proof is short and constructive.  However, it does rely on the existence of distributions over subsets of indices with the desired single-element marginals and with negative correlation.
\end{remark}

\section{Examples}\label{sec:examples}
We have derived the \emph{optimal} $k$-sparse trace distance approximations to pure states.
If $v$ is a pure state, and $v_k$ is the optimal pure $k$-sparse approximation, then we know that $F(v, v_k)^2 + T(v, v_k)^2 = 1$.
In other words, if $F_k(v) = F(v, v_k) = \sqrt{1 - \eps}$, we have $T(v, v_k) = \sqrt{\eps}$.
On the other hand, when we allow approximation by mixed states, the Fuchs-van de Graaf inequalities as in \cref{eq:fvdg_pure} only constrain the optimal trace distance as
$T_k(v) = T(v, \rho) \geq 1 - F(v, \rho)^2 \geq 1 - F_k(v)^2 = \eps$.
In \Cref{sec:simple-example}, we already saw an example where indeed $T_k(v) = O(\eps)$.
However, there are also cases where $T_k(v) = \Theta(\sqrt{\eps})$.

This raises the question: is such a quadratic improvement in approximation by allowing randomized truncation a very special case, or is it common?
We will now give a broad set of examples, which shows that a quadratic improvement happens in many cases.
The examples are discussed from the perspective of $k$-sparsity, but also all have bipartite entanglement equivalents.
In these examples, we will denote by $d$ the dimension of the Hilbert space, and we let $\eps$ denote truncation error in fidelity, so $F_k(v) = \sqrt{1 - \eps}$. If we fix an allowed error $\eps$, this also implicitly defines a value of $k$ that achieves this error (so we can write  $\eps(k)$ or $k(\eps)$).

The easiest way to upper bound $T_k(v)$ is by either giving an explicit approximation $\rho \in \cI_k$, or by using the robustness, which is computed by the $k$-support norm, using \cref{lem:k_supp_efficient_expr}.
Lower bounding $T_k(v)$ can either be done by the Fuchs-van de Graaf inequalities, or by explicitly choosing a vector $m \in \bbC^d$ which lower bounds the optimization problem in \cref{lem:main_td_opt_problem_real}.

\subsection{Additional simple examples}

\paragraph*{Uniform state.}
Suppose ${u} = d^{-1/2}\sum_{i\in [d]} e_i$ (i.e. in the bipartite case, this corresponds to a maximally entangled state).
We have $F_k(u) =\sqrt{k/d}$ so $k=d(1-\eps)$.
When computing the $k$-support norm of $u$, for the condition in \cref{eq:condition k support norm} we have $s_{k-r} = (d - k + r)/\sqrt{d}$, which is never smaller than $(r+1)u_{k-r-1} = (r+1)/\sqrt{d}$ for $r < k-1$ and $k < d$.
So in \cref{lem:k_supp_efficient_expr}, we get $r = k-1$ and
\be
\norm{{u}}_{(k,*)}^2=\frac{s_1^2}{k} = \frac{d}{k} = \frac{1}{1-\eps}
.\ee
This yields
\be
\begin{aligned}
  R_k(u) & = \frac{d}{k} - 1 = O(\eps)                                 \\
  T_k(u) & \leq \frac{R_k(u)}{1 + R_k(u)} = 1 - \frac{k}{d} = O(\eps).
\end{aligned}
\ee
We actually have $T_k(u) = 1 - \frac{k}{d}$, as can be seen by taking $m = u$ in \cref{lem:main_td_opt_problem_real}.\\

The uniform state is the state furthest away from $\cI_k$ for any $k$, when measured in fidelity, robustness or trace distance.
To see this, first $T(\rho,\cI_k)$ is a convex function of $\rho$, so is maximized by a $\rho$ pure.
Let $v$ be any pure state.
We have $\norm{v}_{(k)} \geq \sqrt{k/d}$, with equality for the uniform state.
This shows the uniform state $u$ is furthest from $\cI_k$ in fidelity.
Additionally, it implies that the dual norm satisfies $\norm{v}_{(k,*)} \leq \sqrt{d/k}$ and hence
\be\label{eq:maximal robustness}
R_k(v) = \norm{v}_{(k,*)}^2 - 1 \leq \frac{d}{k} - 1,
\ee
which further implies that for any state $v$
\be\label{eq:maximal trace dist}
T_k(v) \leq \frac{R_k(v)}{1 + R_k(v)} \leq 1 - \frac{k}{d}.
\ee
Both bounds are saturated by the uniform state $u$.

\paragraph*{A state for which there is no advantage.}
We start with a qubit example, with almost all of the weight on a single basis state.
Let
$
v = (\sqrt{1 - \eps}, \sqrt{\eps})^\top
$
and take $k = 1$, so $F_1(v) = \sqrt{1 - \eps}$.
The best 1-sparse approximation in trace distance is
\be
\rho = \begin{pmatrix} 1 - \eps & 0 \\ 0 & \eps \end{pmatrix} .
\ee
It has $F(\rho,v) = \sqrt{1-2\eps +2\eps^2}$, and gives a trace distance $T_1(v)=\sqrt{\eps(1-\eps)}$, which is only slightly better than the pure state approximation.
We can also compute $\norm{v}_{(1,*} = s_1 = \sqrt{1 - \eps} + \sqrt{\eps}$ and hence
\be
R_1(v) = \norm{v}_{(1,*)}^2 - 1 = 2\sqrt{\eps(1-\eps)}.
\ee

We can generalize this to states of arbitrary dimension $d$.
Fix $k \leq d$ and $\eps$.
Consider $v$ with
\begin{align}
  v_i = \sqrt{\frac{1 - \eps}{k}} \text{ for } i = 1, \dots, k \; \text{ and } \; v_i = \sqrt{\frac{\eps}{d-k}} \text{ for } i = k+1, \dots, d.
\end{align}
Then $F_k(v) = \sqrt{1 - \eps}$.
For the trace distance, let
\begin{align}
  m_i = \begin{cases}
          \frac{1}{\sqrt{2k}}     & 1 \leq i \leq k, \\
          \frac{1}{\sqrt{2(d-k)}} & k < i \leq d.
        \end{cases}
\end{align}
Then $\norm{m}_2 = 1$ and by \cref{lem:main_td_opt_problem_real}
\begin{align}
  T_k(v) \geq \abs{\langle m, v \rangle}^2 - \norm{m}_{(k)}^2 = \sqrt{\eps(1-\eps)}
\end{align}
so the trace distance does not improve much beyond the pure state approximation.
For the robustness, $s_{k-r} = (r + 1)\sqrt{\frac{1-\eps}{k}} + \sqrt{(d-k)\eps}$ is large for $d \gg k$,  and we get $r=k-1$ and
\be \norm{v}_{(k,*)}^2 = \frac{s_1^2}{k}  = \frac 1k (\sqrt{(1-\eps)k} + \sqrt{\eps(d-k)})^2 \approx \frac dk \eps. \ee
In this example, for fixed $k$ and $\eps$, the robustness $R_k(v)$ can get arbitrarily larger with increasing $d$. That is, even if the deterministic truncation error $\eps$ is small, $R_k(v)$ can be large.

\paragraph{The case $k=1$.}
A different generalization of the previous example is to consider a general $v\in \bbC^d$ and to take $k=1$, as studied in \cite{Chen_2016_quantifying}.   If $F_1(v)=\sqrt{1-\eps}$ then $v_1 = \sqrt{1-\eps}$ and so $\norm{v_{2:d}}_2=\sqrt{\eps}$.  

There is an explicit but complicated formula for $T_1$ in \cite{Chen_2016_quantifying}.  But without giving an exact expression, we can see that $T_1=O(\sqrt\eps)$ by taking
\be
m = \frac{1}{\sqrt 2} e_1 + \frac{1}{\sqrt{2\eps}} v_{2:k}.
\ee
From \cref{lem:main_td_opt_problem_real}, we have
\be
T_1(v) \geq \langle m,v\rangle^2 - \norm{m}_{(1)}^2
= \qty(\sqrt{\frac{1-\eps}{2}} + \sqrt{\frac{\eps}{2}} )^2 - \frac 12
=\sqrt{\eps(1-\eps)}.
\ee

Thus for $k=1$ we cannot asymptotically improve the trace distance using mixed states.  Of course, even in this case, mixed states are the only way to achieve nontrivial robustness, which here equals $\norm{v}_1^2-1$ .  And mixed states can still yield a non-asymptotic improvement in trace distance for $k=1$.

Intuitively the weakness of the $k=1$ approximation is because $\sigma$ has no off-diagonal entries, so there is not much space for interesting mixed strategies.

\subsection{Power laws}
An informative set of examples will be states in which the coefficients of $v$ decay as a power law.  Here we will see that the asymptotic improvement in trace distance (or lack thereof) will depend on the exponent of the power law. 

For $\gamma > 0$ choose
\be\label{eq:power_law_entries}
v_i = \frac{i^{-\gamma}}{\sqrt{Z}}
\qand
Z = \sum_{i=1}^d i^{-2\gamma}
\ee
so the coefficients of the state have a power law decay.
Does using randomized truncation help?
Recall that randomized truncation is closely related to the $k$-support norm, which interpolates between the $\ell_2$ and $\ell_1$ norms.
There are three relevant regimes. Consider the unnormalized vector $w$ with $w_i = i^{-\gamma}$, so $Z = \norm{w}_2$.
\begin{enumerate}[label=(\Alph*)]
  \item\label{it:small gamma} For $\gamma < \frac12$, both $\norm{w}_2$ and $\norm{w}_1$ diverge with increasing $d$.
  \item\label{it:medium gamma} For $\frac12 < \gamma < 1$, $\norm{w}_2$ converges, but $\norm{w}_1$ diverges with increasing $d$.
  \item\label{it:large gamma} For $\gamma > 1$, both $\norm{w}_2$ and $\norm{w}_1$ converge with increasing $d$.
\end{enumerate}
We will see that for cases \ref{it:small gamma} and \ref{it:large gamma}, we get a quadratic advantage, and for \ref{it:medium gamma} we get a smaller improvement.

For the following, we will assume that both $d$ and $k$ are large.
We can approximate\footnote{Note that the intention of this section is to illustrate a qualitative behaviour, rather than prove a rigorous claim. Hence, for brevity, we omit some details pertaining to the accuracy of approximating sums by integrals.} the normalization $Z$ by using the following asymptotic expansion of generalized harmonic numbers:
\be
H_d^{(\alpha)} := \sum_{j=1}^d j^{-\alpha} = \zeta(\alpha) + \frac{1}{1-\alpha}d^{1-\alpha} + O(d^{-\alpha}).
\ee
In particular, for $\gamma < \frac12$ the constant factor is small for large $d$ and we can approximate
\be
Z \approx \frac{d^{1 - 2\gamma}}{1 - 2\gamma}
\ee
and for $\gamma > \frac12$ and large $d$ we can approximate $Z \approx \zeta(2\gamma)$.
We can also approximate
\be
\eps = \sum_{i = k+1}^d \abs{v_i}^2 \approx \frac{1}{Z}\int_{k}^d x^{-2\gamma} \mathrm{d}x.
\ee
Finally, we can estimate
\be
s_{k-r} = \frac{1}{\sqrt{Z}} \sum_{i = k-r}^d i^{-\gamma} \approx \frac{1}{\sqrt{Z}} \int_{k-r-1}^d x^{-\gamma} \mathrm{d}x .
\ee
Ignoring some terms which are small as $d$ grows, this gives the approximations shown in \Cref{tab:power_law_quantities}.

\begin{table}
    \centering
    \begin{tabular}{cccccc}
    \toprule
    &$\gamma$                           & $Z$                                 & $\eps$                                      & $k$                                                         & $s_{k-r}$                                                                      \\
    \midrule
    \ref{it:small gamma}:& $(0,1/2)$    & $\frac{d^{1 - 2\gamma}}{1-2\gamma}$ & $1 - \qty(\frac kd)^{1-2\gamma}$            & $d(1-\eps)^{\frac{1}{1-2\gamma}}$                           & $\frac{d^{1-\gamma}-(k-r-1)^{1-\gamma}}{d^{\frac 12-\gamma}}$                  \\
    \ref{it:medium gamma}:& $(1/2,1)$   & $\zeta(2\gamma)$                    & $\frac{k^{-(2\gamma - 1)}}{\zeta(2\gamma)}$ & $\qty(\frac{\zeta(2\gamma)}{\eps})^{\frac{1}{2\gamma - 1}}$ & $\frac{d^{1-\gamma}-(k-r-1)^{1-\gamma}}{\sqrt{\zeta(2\gamma)}(1-\gamma)}$      \\
    \ref{it:large gamma}:& $(1,\infty)$ & same                                & same                                        & same                                                        & $\frac{(k-r+1)^{1 - \gamma}-d^{1-\gamma}}{\sqrt{\zeta(2\gamma)}(\gamma - 1)} $ \\
    \bottomrule
  \end{tabular}
    \caption{\label{tab:power_law_quantities}Approximations of relevant quantities for target states obeying power laws at different values of the parameter $\gamma$.}
\end{table}

We will now determine for these three cases how $T_k(v)$ scales with $\eps$, where $\sqrt{\eps}$ is the deterministic truncation error. We take $\gamma$ to be a fixed constant value (which we absorb into the constants in the big-$O$ notation).
We will see that we get a quadratic advantage for small $\eps$ for cases \ref{it:small gamma} and \ref{it:large gamma}, and a smaller advantage in between for \ref{it:medium gamma}. This is illustrated in \Cref{fig:power_law_plot}.

\paragraph*{Case \ref{it:small gamma}:}
For this case, where $0 < \gamma < \frac12$ we take $k$ as in the table, and use the general bound in \cref{eq:maximal trace dist} to get
\be
T_k(v) \leq 1 - \frac{k}{d} = 1 - (1 - \eps)^{\frac{1}{1 - 2\gamma}} = O(\eps),
\ee
so for large $k$ (and small $\eps$) randomized truncation gives a quadratic improvement in trace distance.

\paragraph*{Case \ref{it:medium gamma}:}
Here $\frac12 < \gamma < 1$.
Let $\beta = 1/(3 - 2\gamma)$, let $\ell = k^{2\beta}$, and define $w$ to be the deterministic truncation of $v$ to $\ell$ terms.
The error in trace distance can be bounded as
\be
T(v,w) = \sqrt{1 - F_{\ell}(v)^2} = \sum_{i = \ell+1}^d v_i^2 = O\mleft(\ell^{\frac{1}{2\gamma - 1}}\mright) = O(\eps^{\beta}).
\ee
We now bound $R_k(w)$. Using the quantities listed in \Cref{tab:power_law_quantities} for this case, we see
\be
\begin{aligned}
  \norm{w}_{(k,*)}^2 &= \sum_{i=1}^{k-r-1} \abs{w_i}^2 + \frac{s_{k-r}^2}{r+1} \\&\leq 1 + \frac{s_{1}^2}{k}\\
                     & = 1 + O\mleft(\frac{\ell^{2(1 - \gamma)}}{k}\mright) \\
                     &= 1 + O\mleft(k^{4\beta(1 - \gamma) - 1}\mright)\\
                     &= 1 + O\mleft(k^{-\frac{2\gamma - 1}{3 - 2\gamma}}\mright)\\
                     &= 1 + O(\eps^{\beta}),
\end{aligned}
\ee
which gives
\be
R_k(w) = O(\eps^{\beta}) \quad \text{ and } \quad T_k(w) = O(\eps^{\beta}).
\ee
Therefore,
\be
T_k(v) \leq T(v,w) + T_k(w) = O(\eps^{\beta}).
\ee
Since $\beta > \frac12$ this is an improvement over the pure state approximation, which achieves trace distance $\sqrt{\eps}$.

\begin{figure}
  \centering
  \includegraphics[width=0.45\linewidth]{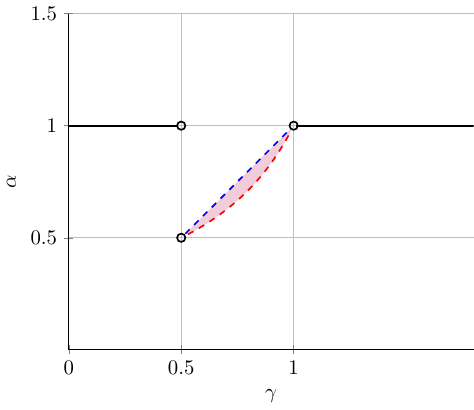}
  \caption{Leading order exponent $\alpha$ of the optimal randomized trace distance as a function of $\eps\equiv 1-F_k(v)^2$ for a $d$-dimensional power law state $v$, as described in the main text. In particular, the state $v$ is defined as in \cref{eq:power_law_entries} with parameter $\gamma$, and the value of $k$ is chosen to keep $\eps$ fixed as $d\to\infty$. Then the resulting trace distance is $c\eps^\alpha$ to leading order, for some constant $c$. For $\gamma\notin (1/2,1)$ the exponent is equal to $1$, corresponding to the maximum possible advantage using randomized truncation. For $\gamma\in (1/2,1)$ the value of the exponent lies within the shaded region.}
  \label{fig:power_law_plot}
\end{figure}

We now show that one can not do much better in this case, and that $T_k(v) = \Omega(\eps^{\gamma})$, which is worse than the optimal scaling with $\eps$ since $\gamma < 1$.
We give a lower bound by choosing a unit vector $m$ and bounding
\be
T_k(v) \geq \abs{\langle m, v \rangle}^2 - \norm{m}_{(k)}^2
\ee
by \cref{lem:main_td_opt_problem_real}.
We let $\ell > k$ and $0 < a < 1$, which we choose later, and let
\be
m_i = \begin{cases} a v_i                                 & \text{ for } i = 1, \dots, k        \\
              \sqrt{\frac{1 - a^2(1 - \eps)}{\ell}} & \text{ for } i = k+1, \dots, k+\ell
\end{cases}
\ee
which is chosen such that
\be
\norm{m}_2^2 = a^2 \norm{v}_{(k)}^2 + 1 - a^2(1 - \eps) = 1.
\ee
We now have
\be
\begin{aligned}
  \abs{\langle m, v \rangle}^2 & \geq a^2 \norm{v}_{(k)}^4 + 2a \sqrt{\frac{1 - a^2(1 - \eps)}{\ell}} \sum_{i=k+1}^{k+\ell} v_i \\
                               & = a^2(1-\eps)^2 + 2a \sqrt{\frac{1 - a^2(1 - \eps)}{\ell}} \sum_{i=k+1}^{k+\ell} v_i \, .
\end{aligned}
\ee
Also, as long as $m_{k+1} \leq m_k$,
\be
\norm{m}_{(k)}^2 = a^2 \norm{v}_{(k)}^2 = a^2(1 - \eps).
\ee
This gives
\be
T_k(v) \geq \abs{\langle m, v \rangle}^2 - \norm{m}_{(k)}^2 =  2a \sqrt{\frac{1 - a^2(1 - \eps)}{\ell}} \sum_{i=k+1}^{k+\ell} v_i - O(\eps).
\ee
We choose $a = \frac12 (1 - \sqrt{\eps})^{-\frac12}$, which gives $1 - a^2(1 - \eps) = \frac14(3 - \eps)$, so and $ab = \Theta(1)$. We choose $\ell = \Theta(k^{2\gamma})$,  which is good enough to make sure $\Omega(k^{-\gamma}) = m_{k} \geq m_{k+1} = \Theta(1/\sqrt{\ell})$, and for which $v_{k+1} + \dots + v_{k+\ell} = \Omega(\ell^{1 - \gamma})$.
We conclude that
\be
T_k(v) =\Omega(\ell^{\frac12 - \gamma}) = \Omega(k^{-\gamma(2\gamma - 1)}) = \Omega(\eps^{\gamma}).
\ee

\paragraph*{Case \ref{it:large gamma}:}
For large $d$, we have
\be
\eps = \Theta(k^{1 - 2\gamma}).
\ee
When computing the $k$-support norm, we have $s_{k-r} = \Theta((k-r)^{1 - \gamma})$, since for $\gamma > 1$, the $d^{1-\gamma}$ term is subleading.
The condition that $(r+1)v_{k-r} = \Theta(r (k-r)^{-\gamma})$ should be approximately equal to $s_{k-r}$ in \cref{eq:condition k support norm} then gives $r = \Theta(k)$ and as before we can bound
\be
\norm{v}_{(k,*)}^2 = \sum_{i=1}^{k-r-1} \abs{v_i}^2 + \frac{s_{k-r}^2}{r+1} \leq 1 + O(k^{-1} (k^{1 - \gamma})^2) = 1 + O(k^{1 - 2\gamma}) = 1 + O(\eps).
\ee
Hence, $R_k(v) = O(\eps)$, and also $T_k(v) = O(\eps)$.

\paragraph{Summary.} We can summarize the state of affairs for power law states with a diagram, given in \Cref{fig:power_law_plot}. For $v$ a fixed $d$-dimensional power law state, as in \cref{eq:power_law_entries}, consider the trace distance $T_d(\eps)$ defined as $T_k(v)$ for $k$ such that $\eps= 1-F_k(v)^2$. In other words, this is the optimal randomized trace distance as a function of $\eps$, where $\sqrt{\eps}$ is the optimal deterministic trace distance. Then the arguments above show that $\lim_{d\to\infty}T_d(\eps) = \Theta(\eps^{\alpha})$ for some $\alpha\in [1/2,1]$, the exponent of the leading-order contribution to the trace distance. The case where $\alpha=1$ represents the maximal possible advantage allowed by the bound in \cref{eq:fvdg_pure}, while $\alpha=1/2$ corresponds to no advantage over deterministic truncation.

\section{Application to matrix product states}\label{sec:MPS}

In this section, we discuss the application of randomized truncation to tensor network algorithms. We focus on matrix product states (MPS), but the principle applies to other tensor networks as well.
\begin{dfn}
  Let $\cH = (\bbC^d)^{\otimes n}$, let $ \{e_{i_1 \dots i_n} = e_{i_1} \ot \cdots \ot e_{i_n}: i_1,\dots,i_n \in[d]\}$ be the standard basis vectors, and let $r_0,r_1,\dots, r_n\in\bbZ_{>0}$ with $r_0=r_n$. A state $v\in\cH$ is a \emph{matrix product state} with bond dimensions $(r_0,\dots,r_n)$ if it can be written as
  \begin{equation}
    v
    \;=\;
    \sum_{i_1,\ldots,i_n\in[d]}
    \tr\Big( M^{(1,i_1)} M^{(2,i_2)} \cdots M^{(n,i_n)} \Big)\,e_{i_1\cdots i_n},
  \end{equation}
  where for each $\alpha\in [n]$ and $i\in [d]$ we have $M^{(\alpha,i)}\in\bbC^{r_{\alpha-1}\times r_\alpha}$ are some matrices.
\end{dfn}
We restrict our attention to MPS with open boundary conditions, where $r_0=r_n=1$. In many applications, the goal is to simulate a quantum many-body system (e.g., spin degrees of freedom on a lattice, or shallow quantum circuits) by representing its state as an MPS. In order to limit the computational and memory resources used in these simulations, a common subroutine is to truncate a bond, which usually refers to keeping the largest Schmidt coefficients with respect to the bipartition corresponding to that bond.

We propose using an optimal randomized truncation of the bonds instead. Although the results from \Cref{sec:trace-dist,sec:robustness} give efficient methods for this procedure, there are at least two unresolved issues in this setting. First, the trace distance is a worst-case error metric and therefore does not provide a guaranteed improvement over determinstic truncation for a specific observable of interest. Nevertheless, it is reasonable to expect quantum states which are closer in trace distance (or robustness) to lead to better estimates. 
The second issue is that performing multiple randomized truncations in succession may not yield a globally optimal approximation, i.e., an approximation of the target state as a convex combination of MPS with bounded bond dimensions. However, this observation applies equally well in the case where one uses deterministic truncations. In fact, finding the best product state approximation (corresponding to $\chi = 1$) to a multipartite state is $\NP$-hard to compute, even for $n=3$ parties \cite{hillar2013most}, so one should not expect efficient algorithms for finding globally optimal MPS approximations.

\paragraph{Details of numerics.} In light of the above discussion, we perform numerics to test the effectiveness of using randomized truncation for multipartite states. We begin by generating partially random $10$-qubit states as follows. First, we sample a $10$-qubit state from the uniform distribution. Next, we compute a Schmidt decomposition of the state with respect to the first bond and replace the resulting Schmidt coefficients with coefficients obeying a (normalized) power law as in \cref{eq:power_law_entries}, for fixed exponent $\gamma$. We then carry out an identical procedure on the remaining bonds in sequence from left to right. The result is a multipartite state whose entanglement spectrum with respect to any given bipartition is constrained depending on the value of $\gamma$.

To compare the deterministic and randomized approaches, we truncate the bonds of the resulting state (viewed as an MPS) and compute the expected value of the observable $Z_5Z_6$, the Pauli-$Z$ operators on the 5th and 6th qubits. Specifically, we once again perform Schmidt decompositions from left to right in sequence, truncating and normalizing each time. For the deterministic truncation, at each bond we keep the top $\chi$ Schmidt coefficients, where $\chi$ is the bond dimension cutoff. For the randomized method, at each bond we randomly sample the coefficients according to the optimal randomized truncation in trace distance with sparsity $\chi$, as in \Cref{thm:opt_td_sampling}. The result is a randomly sampled MPS with bond dimensions at most $\chi$, from which we can compute the expected value of $Z_5Z_6$. To compute the final estimate, we take the sample mean of 100 samples. We plot the results for different values of $\gamma$ in \Cref{fig:mps}.

\paragraph{Discussion.} \Cref{fig:mps} suggests that randomized truncation offers a significant improvement for heavier tails and at small bond dimension cutoffs, by an order of magnitude in the relative error in some cases. However, when the exponent for the power law $\gamma=0.8$, there is little improvement over the deterministic approach, and both methods result in fairly accurate estimates. This is consistent with the observation that for, say, the middle bond, more than $70\%$ of the entanglement spectrum is concentrated on just the top two squared Schmidt coefficients, so one should not expect more than a modest improvement by sampling smaller Schmidt coefficients.
\begin{figure}
  \centering
  \includegraphics[width=0.7\linewidth]{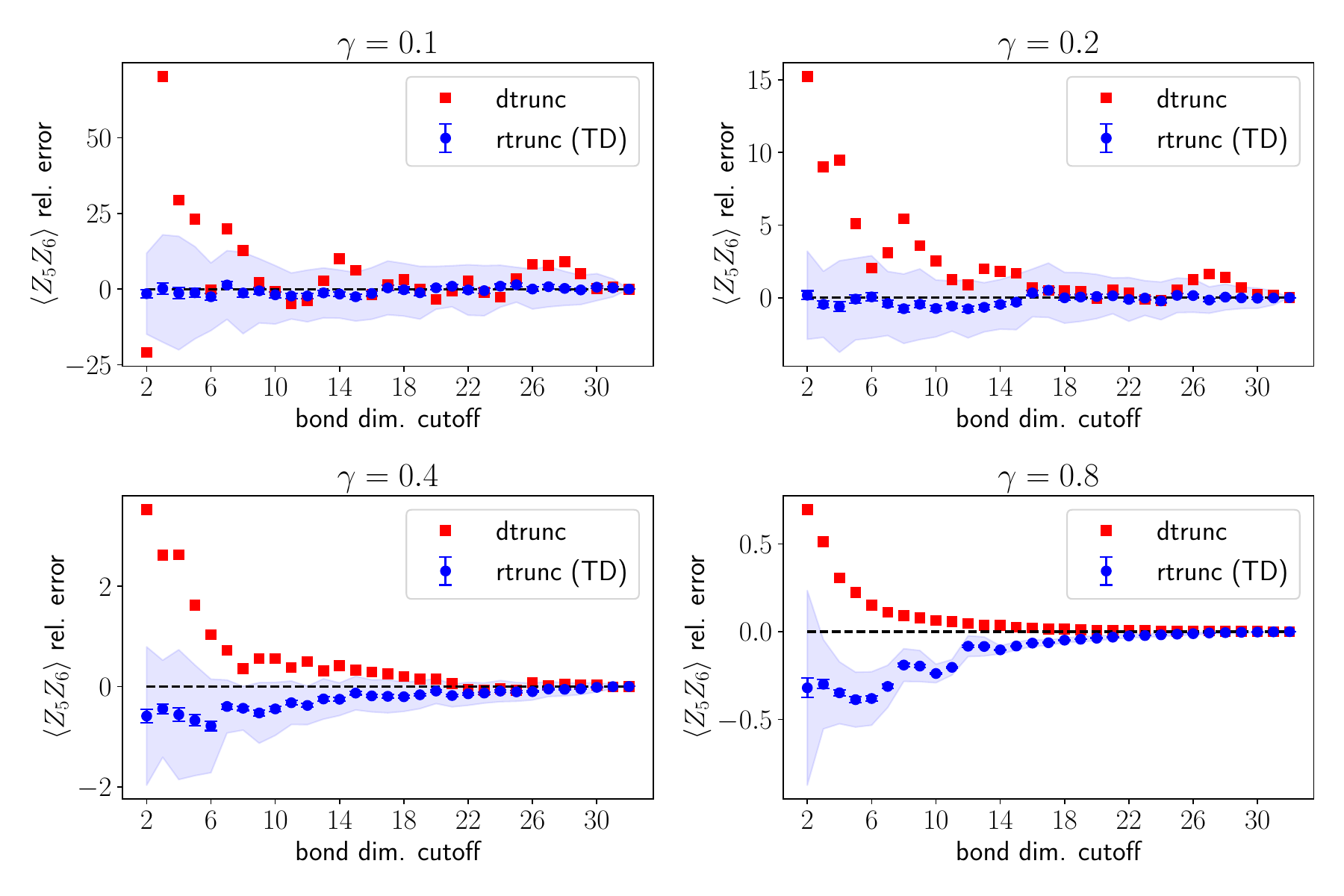}
  \caption{\label{fig:mps}Estimating the observable $\langle Z_5Z_6\rangle$ with respect to 4 random states on 10 qubits, which are generated using the method described in the main text. The parameter $\gamma$ determines the strength of the power law decay in the Schmidt coefficients across each bond. The \textsf{dtrunc} estimate is obtained from the standard, deterministic truncation of the bonds, while the \textsf{rtrunc (TD)} estimate is the sample mean of 100 samples obtained by applying the optimal (bipartite) randomized truncation in trace distance. The light blue region shows an estimate of the standard deviation $\sigma$, and error bars of length $\sigma/10$, corresponding to uncertainties in the empirical mean, are also displayed.}
\end{figure}
Randomized truncation can in principle be used as a subroutine in any MPS algorithm that requires bond dimension truncation, including TEBD and DMRG (see, e.g., \cite{SCHOLLWOCK201196}). Whether it is beneficial in practice will depend on the application at hand and how well trace distance or robustness correlate with the relevant metrics for success.
Crucially, the computational cost of the randomized truncation may be comparable to the standard method. Note that in either case we must perform Schmidt decompositions at each bond. However, for randomized truncation we must in addition apply the algorithms from \Cref{sec:optimal_density_matrix_td} or \Cref{sec:robustness}, which require computing the weights in a max-entropy distribution. As shown in \Cref{sec:sampling}, this can be done efficiently and in practice has cost $O(\chi d)$ per iteration when using the Newton method with the diagonal approximation outlined in \Cref{sec:fast_algorithms}. Upon solving for the weights, obtaining a single sample of the Schmidt coefficients at a given bond takes $O(\chi)$ time. Each of these times is subleading compared to the amount of time to perform the SVDs in the deterministic truncation. To fully evaluate the utility of the method, it remains to determine the number of iterations and samples needed, and then compare this with the cost of reaching the same precision with deterministic truncation and a larger bond dimension. We leave a detailed analysis of this tradeoff to future work.

\paragraph{Prior work on tensor network Monte Carlo.} The use of a randomized tensor network truncation algorithm has previously been proposed in Ref.~\cite{Ferris15}. There are two key differences between that work and the approach we have discussed above. First, we use explicitly normalized states for the randomized truncation, whereas the method from Ref.~\cite{Ferris15} uses an unbiased estimate of the state vectors, and so each state is only normalized in expectation. Iterating the latter procedure enough times will lead to an exponentially growing variance~\cite{todo2024markov}, while iterating approximations using normalized states will instead incur growing bias.

A more substantial difference is that our method may drop or keep a subset of the Schmidt coefficients deterministically, and chooses marginal inclusion probabilities for the remaining indices, which we have denoted by $q$ thoughout this work, in a way that minimizes trace distance or robustness.  In \cref{sec:sampling} we describe how to compute from $q$ the max-entropy weights $\mu$ which lead to the optimal ensembles of states. In contrast, Ref.~\cite{Ferris15} simply sets $\mu_i$ such that $v_i^2 = e^{-\mu_i}$ for all the indices $i$.  This is a reasonable heuristic, since larger values of $v_i$ then yield larger values of $q_i$, though it does not enjoy the theoretical guarantees established in the present work. How the method of~\cite{Ferris15} compares to our randomized approach is an interesting question which could be explored numerically in future work.

\section{Improved algorithms for conditional Poisson sampling}\label{sec:sampling}
We make use of efficient methods to compute and sample from distributions over $\ell$ distinct integers drawn from the set $[n]=\{1,2,\dots,n\}$, with $n\geq \ell$. Equivalently, we are interested in random $n$-bit strings of Hamming weight $\ell$, representing inclusion in the draw. We refer to this set as $H_\ell\coloneq\{x\in\{0,1\}^n: |x|=\ell\}$. The precise problem we consider is as follows. The inputs are the alphabet size $n$, the size of the set $\ell$, and the marginal probabilities $q_1,\ldots,q_n \in (0,1)$ that each element is included in the set\footnote{We exclude the possibility of $q_i\in \{0,1\}$  without loss of generality, since in those cases we can deterministically leave out or include that element and do not need to involve it in the sampling.}.    Normalization requires that $\sum_{j=1}^n q_j = \ell$. 
Given these inputs our goal is to efficiently compute a description of a distribution $p:H_\ell\to [0,1]$ that satisfies four desired criteria:
\begin{enumerate}
  \item single-element marginal (inclusion) probability constraints are satisfied, i.e., $q_j=\sum_{x\in H_\ell}p(x)x_j$ for every $j\in [n]$;
  \item two-element marginals
        $Q_{i,j}\coloneq \sum_{x\in H_\ell}p(x) x_i x_j$
        can be computed in time $\poly(n)$;
  \item the distribution can be sampled from in time $\poly(n)$; and
  \item negative correlation (of indicator random variables): $Q_{i,j}\leq q_iq_j$ for all $i,j\in [n]$.
\end{enumerate}
As we will see in \Cref{sec:max_entropy} below, this problem reduces to solving for the maximum-entropy distribution subject to Item 1, the single-element marginal probability constraints. Such distributions over $H_\ell$ are a special case of Fisher's multivariable noncentral hypergeometric distribution and are ubiquitous throughout physics and mathematics; for example, when studying systems of $\ell$ fermions restricted to occupy $n$ distinct modes, or in survey sampling, where the corresponding sampling task is known as conditional Poisson sampling~\cite{Tille06}. The latter terminology, which we adopt in this work, is due to the observation that (i) sampling \emph{independent} but non-identical Bernoulli random variables $z_1,\dots,z_n$ is called Poisson sampling, and (ii) if the resulting conditional distribution, conditioned on $\sum_i z_i = \ell$, satisfies the marginal constraints it is in fact the maximum-entropy distribution over $H_\ell$. (See, e.g., \cite{Chen94}.) For randomized truncation the problem is relevant because our methods, described in \Cref{sec:trace-dist,sec:robustness}, involve sampling random fixed-size subsets of indices from the vector to be truncated.

\subsection{Maximum-entropy distributions}\label{sec:max_entropy}
We consider the following optimization problem:
\begin{equation}\label{eq:max_ent_primal}
  \begin{alignedat}{2}
    \underset{p\in \bbR_{\geq 0}^{H_\ell}}{\textnormal{max.}} & \quad \rmH(p) \coloneq \sum_{x\in H_\ell} p(x)\log(1/p(x))                              \\[0.5em]
    \textnormal{s.t.}                                         & \quad \sum_{x\in H_\ell} p(x)x_j = q_j                     & \quad & \forall j \in [n], \\
                                                              & \quad \sum_{x\in H_\ell} p(x) = 1.
  \end{alignedat}
\end{equation}
Since there are exponentially many variables and only $n+1$ constraints, most work on the subject focuses on the dual optimization problem.  First we observe that the optimal solution to \cref{eq:max_ent_primal} takes the form $p^\star(x)\equiv p_{\mu^\star}(x)$ for some $\mu^\star\in\bbR^n$, where
\begin{align}
  p_\mu(x) &\coloneq \frac{\ee^{-\sum_{j=1}^nx_j \mu_j}}{Z}\ \forall x\in H_\ell, \quad Z_\mu\coloneq \sum_{x\in H_\ell}\ee^{-\sum_{j=1}^nx_j \mu_j}
  \label{eq:FMNHD}
\end{align}
for any set of parameters $(\mu_j\in\bbR$: $j\in [n])$. This is a standard fact from statistical physics, but 
for a self-contained proof of a slightly more general fact, see \cite[Lemma~A.1]{singh2014entropy}.
Due to \cref{eq:FMNHD}, the optimal distribution has a succinct description in terms of just $n$ parameters $\mu_1^\star,\dots,\mu_n^\star$, which correspond to the optimal dual variables for the dual problem:
\begin{equation}\label{eq:max_ent_dual}
  \begin{alignedat}{2}
    \textnormal{min.}  & \quad  g_q(\mu)\coloneq \sum_{i=1}^n\mu_iq_i + \log(\sum_{x\in H_\ell}\ee^{-\sum_{j=1}^nx_j\mu_j})                             \\[0.5em]
    \textnormal{s.t .} & \quad \mu_j\in\bbR \quad \forall j \in [n]
  \end{alignedat}
\end{equation}
where $q\coloneq (q_1,q_2,\dots,q_n)^\top$.

We introduce some additional notation before proceeding further.  In the above, we considered $q$ to be an input to the optimization problems \cref{eq:max_ent_dual,eq:max_ent_primal}.  However, we can also consider $q(\mu)$, a function of $\mu$, where we first determine a distribution $p_\mu$ over $H_\ell$ which has the form given in \cref{eq:FMNHD}, and then we use this to determine the single-element marginals $q(\mu)_i = \sum_x p_\mu(x) x_i$ and two-element marginals $Q(\mu)_{ij} = \sum_x p_\mu(x) x_i x_j$. In particular, $q(\mu^\star)=q$ for the dual optimal solution $\mu^\star$. The following is a useful fact about these marginals, recorded as Property 2 in \cite{Chen94}.
\begin{lem}\label{lem:negative_correlation}
  For any $\mu\in\bbR^n$ and $i,j\in [n]$ with $i\neq j$ we have $0 < {Q}(\mu)_{ij} < q(\mu)_{i}q(\mu)_j$.
\end{lem}
In other words, the $x_i$, $x_j$ are negatively correlated for any $i\neq j$.

\subsection{Sampling from the maximum-entropy distribution}\label{sec:faster sampling}

Given the single-element marginal probabilities $q\in\bbR^n$ as input, the algorithmic tasks involved in a maximum-entropy approach to the conditional Poisson sampling problem are:
\bit
\item
Compute the dual optimal solution $\mu^\star$ from $q$.
\item
Given $\mu\in\bbR^n$, compute $Z_\mu$, $q(\mu)$ and $Q(\mu)$.
\item
Given $\mu\in\bbR^n$, sample from $p_{\mu}$.
\eit
The first problem is referred to as \textit{inference} since we need to infer the parameters $\mu^\star$ from $q$, and we give algorithms for it in \Cref{sec:fast_algorithms} below.
In the remainder of this section we focus on the arithmetic complexity of the latter two problems, ignoring bit complexity. However, \Cref{thm:efficient_computation_of_weights} gives us a polynomial bound on the largest value of $\norm{\mu}_\infty$ which need be considered; hence taking bit complexity into account only adds a polylogarithmic overhead. We give algorithms that improve on those in \cite{Chen94, Tille06} for large enough $\ell$, which computed $Z_\mu$ in time $O(n\ell)$, $q(\mu)$ in time $O(n\ell^2)$ and $Q(\mu)$ in time $O(n^2\ell)$, and sampled from $p_\mu$ in time $O(n\ell)$.

\begin{lem}\label{lem:marginal_computation}
  Given $\mu\in\bbR^n$ one can compute $Z_\mu$ in time $O(n\log^2\ell\log\log\ell)$,  $q(\mu)$ in time $O(n\ell)$, and $Q(\mu)$  in time $O(n^2)$.
\end{lem}


\begin{proof} 
  The partition function $Z_\mu$ is equal to the coefficient of $t^\ell$ in the polynomial
  \be
  f(t) = \prod_{i=1}^n \qty(1 + t e^{-\mu_i})
  \ee
  We can evaluate $f(t)$ with a multiplication tree of depth $\log n$, while keeping track only of intermediate polynomials modulo $t^{\ell+1}$.  This involves $n/2^j$ degree-$2^{j-1}$ multiplies for $j=1,\ldots,\log \ell$ and then $O(n/\ell)$ degree-$\ell$ multiplies.  Since degree-$d$ polynomial multiplication can be done in time $O(d \log(d) \log\log(d))$ using Ref.~\cite{poly-mult87}, the final runtime is $O(n \log^2(\ell)\log\log(\ell))$.

  The previous computation of $Z_\mu$ in fact yielded $Z_\mu^{(m)}$ for $m=1,\ldots,\ell$ where we define
  \ba
    p^{(m)}_{\mu} (x)\coloneq \frac{\ee^{-\sum_{j\in [n]}x_j\mu_j}}{Z_\mu^{(m)}}\ \forall x\in H_m,\quad Z_\mu^{(m)} \coloneq \sum_{x\in H_m} \ee^{-\sum_{j\in [n]} x_j\mu_j}
  \ea
Similarly, set $q_j^{(m)}(\mu)=\sum_{i\in [n]}p_\mu^{(m)}(x)x_i$ to make the $m$-dependence of the marginal probabilities explicit.    Then
Result 22(i) from \cite{Tille06} states that
\be
q_j^{(m)}(\mu) = e^{-\mu_j} (1 - q_j^{(m-1)}(\mu)) \frac{Z_\mu^{(m-1)}}{Z_\mu^{(m)}} .
\ee
Note that $q_j^{(0)}=0$.    Then given knowledge of $Z_\mu^{(1)},\dots,Z_\mu^{(\ell)}$ from the previous step we can compute each $q_j^{(\ell)}(\mu)=q_j(\mu)$ in time $O(\ell)$.  Computing all of $q$ takes time $O(n\ell)$, which dominates the runtime compared to computing the required partition functions.
(Result 22(iii) from \cite{Tille06} is another option, yielding the same runtime.)

For $Q_{ij}(\mu)$, when $\mu_i\neq \mu_j$ there is a formula that can be evaluated in constant time.  In this case, Result 23(ii) from \cite{Tille06} states
\be
Q_{ij}(\mu) = \frac{q_i(\mu) e^{-\mu_j} - q_j(\mu) e^{-\mu_i}}{e^{-\mu_j} - e^{-\mu_i}}  .
\ee
The remaining entries of $Q(\mu)$ can be computed using the fact that $Q_{ii}(\mu) = q_i(\mu)$ and $\sum_j Q_{ij}(\mu) = q_i(\mu) n$.   Concretely,  
when $\mu_i=\mu_j$, define $C \coloneq  \abs{\{ j : j\neq i, \mu_j = \mu_i\}}$, so that
\be
Q_{ij}(\mu) = \frac 1C q_i(\mu)(n-1) - \frac 1C \sum_{k:\mu_k\neq \mu_i } Q_{ik}(\mu)
\ee
when $i\neq j$ and $\mu_i=\mu_j$. This also costs $O(1)$ time per entry, if we assume that the $\mu$ vector is sorted.  Thus $Q(\mu)$ can be computed in time $O(n^2)$.
\end{proof}

We now give an improved algorithm to sample from the distribution $p_\mu$. A different algorithm with the same runtime is given in \cite{HJJ20}, but it only yields approximate samples and the runtime bound is not fully proven.
\begin{lem}\label{lem:fast-sampling}
 Given $\mu\in\bbR^n$, a sample from $p_\mu$ can be obtained in expected time $O(n\log n)$.
\end{lem}

The first step of our algorithm involves a rescaling of the weights which we separate into its own lemma.

\begin{lem}\label{lem:rescale}
  Given $\mu\in\bbR^n$ and $\ell\leq n-1$, in time $O(n\log (n(1+\|\mu\|_\infty)/\eps))$ we can find $\Delta\in\bbR$ such that
  \be \sum_{i=1}^n (1+e^{\mu_i+\Delta})^{-1} = \ell,\label{eq:mu-centered}\ee
  up to error $\eps$.
  This implies that the product distribution
  $P_{\mu+\Delta} = \bigotimes_i \text{Ber}\qty(1+e^{\mu_i+\Delta})^{-1}$ has expected Hamming weight $\ell$ (up to the approximation error).  Additionally $P_{\mu+\Delta}|_{H_\ell}=p_\mu$.
\end{lem}

We say that the $\mu+\Delta$ satisfying \cref{eq:mu-centered} are ``centered.''

\begin{proof}[Proof of \cref{lem:rescale}]
  Let $g:\bbR\to [0, n]$ be defined through $g(\Delta) = \sum_{i=1}^n\frac{1}{1+e^{\mu_i+\Delta}}$.  Our goal is to find $\Delta$ such that $g(\Delta)=\ell$.  Since $g$ is monotonically decreasing, binary search will work here.  
Write $R = \norm{\mu}_\infty$. Then
\begin{align}
  \frac{n}{1+e^{R+\Delta}}\leq g(\Delta) \leq \frac{n}{1+e^{-R+\Delta}}
  \qand
  -n \leq g'(\Delta) \leq 0.
\end{align}
We can then set
\begin{align}
    a\coloneq -R + \log(\frac{n-\ell}{\ell})\quad\textnormal{and}\quad b\coloneq R+\log(\frac{n-\ell}{\ell})
\end{align}
to ensure that $g(a)\geq \ell$ while $g(b)\leq \ell$. Performing a binary search over the interval $[a,b]$ yields $\Delta$ in time $O(n\log (R/\eps))$, to accuracy $\eps$. which will suffice for the rest of the algorithm.

Finally, the last claim of the lemma is that $p_\mu$ is not affected by the choice of $\Delta$.  This is because each $x\in H_\ell$ has the same $e^{\ell \Delta}$ dependence on $\Delta$, so once we condition on $x\in H_\ell$, the effects of $\Delta$ go away.
\end{proof}

By \cref{lem:R-bound}, if $q_j \in [\eta,1-\eta]$, we can easily shift $\mu$ so that $\norm{\mu}_\infty \leq O(\log(1/\eta))$, and inverse-polynomial precision in $\eps$ will suffice, so this step can be done efficiently.
    
\begin{proof}[Proof of \cref{lem:fast-sampling}]
Begin by using the algorithm \cref{lem:rescale} so that we can use a vector $\mu\in\bbR^n$ such that $\sum_{i=1}^n(1+e^{\mu_i})^{-1}=\ell$. Define the product distribution $P_\mu$ where $x_i$ is distributed according to
$\text{Ber}(\frac{1}{1+e^{\mu_i}})$,
i.e.,~the Bernoulli random variable with
$\Pr[x_i=1]=\frac{1}{1+e^{\mu_i}}$ and
$\Pr[x_i=0]=\frac{e^{\mu_i}}{1+e^{\mu_i}}$.  Our choice of $\Delta$ ensures that $\bbE_{P_\mu}[|x|]=\ell$.  Furthermore, it is
straightforward to check that $p_\mu$ is equivalent to conditioning
$P_\mu$ on the event $H_\ell := \{x : |x|=\ell\}$.  Rather than naively sampling until the desired event occurs\footnote{In fact, this rejection sampling approach would achieve a runtime $O(n\sqrt\ell)$, which is already an improvement over the $O(n\ell)$-time draw-by-draw algorithm from~\cite{Chen94}, and has been described elsewhere in the survey sampling literature. However, as shown in the proof we can do better.}, out algorithm proceeds as follows:

\benum
\item
  Sample $x$ from $P_\mu$.
\item Repeat the following $n$ times.
  \bit
\item
  If $x\in H_\ell$ then output $x$ and stop.
\item
  Choose $i\in [n]$ uniformly randomly and replace $x_i$ with $\text{Ber}(\frac{1}{1+e^{\mu_i}})$.
  \eit
  \item Repeat Step 1.
  \eenum
  
  We claim that this procedure halts in Step 2 with probability $\Omega(1)$.  This means the expected runtime is $O(n)$, and the runtime is at most $O(n\log(1/\eps))$ with probability at least $1-\eps$.
  Furthermore, conditioned on halting, the algorithm outputs a sample from $p_\mu$. To prove this, let $x{(0)},x{(1)}, x{(2)}, \ldots$ denote the values of $x$ produced by the algorithm.  By construction, $x{(0)}$ is distributed according to $P_\mu$ and, since the update procedure is Glauber dynamics, the subsequent $x{(t)}$ are as well.  Finally, the algorithm halts when $x{(t)}\in H_\ell$, thus producing a sample from $P_\mu$ conditioned on $H_\ell$.  This yields a sample from $p_\mu$.

  It remains only to analyze the success probability in Step 2. We can get intuition for this from the case when each $q_i$ is equal to $1/2$, $\ell=n/2$ and $n$ is large.  Let $z = |x| - n/2$.  In that case $z$ starts typically around $\pm O(\sqrt{n})$ and then evolves according to $n$ steps of a simple random walk.  For concreteness, suppose $z$ starts at $c\sqrt{n}$ for $c>0$.   The probability that a random walk will hit $-c\sqrt{n}$ after $n$ steps is known to be $\approx \exp(-c^2/2)$, which is $\Omega(1)$ if $c$ is $O(1)$.  Extending this argument to nonequal $q_i$ is  more delicate but follows a similar intuition. We next analyze this part of the algorithm
  
  Define $s(x)$ to be $+1$ if $|x|>\ell$, $0$ if $|x|=\ell$, and $-1$ if $|x|<\ell$.     A sufficient, but not necessary, condition for the  algorithm to succeed is for $s(x(0))s(x(n))$ to equal $0$ or $-1$.  If the product is 0, then either $x(0)$ or $x(n)$ is in $H_\ell$.  If the product is $-1$ then the Hamming weight went from $>\ell$ to $<\ell$ (or vice versa) and so at some intermediate point we had $x(t)\in H_\ell$.  It is possible for the algorithm to succeed even if $s(x(0))s(x(n))=1$, but we will not analyze this case.  We now use the mixing properties of the walk to show that the signs of $s(x(0))$ and $s(x(n))$ have bounded correlations, and thus have a decent chance of either having opposite signs or of being zero.

  Let  $\pi \coloneq P_\mu(H_\ell)$. 
  For later use, we record that
  \ba\label{eq:pi large}
  \pi = \Omega(\ell^{-1/2}).
  \ea
  Indeed, this follows from Chebyshev's inequality together with
  \ba
  \bbE \abs{x} = \sum_i q_i = l, \qquad \Var[\abs{x}] = \sum_i q_i(1-q_i) \leq \sum_i q_i = l.
  \ea
  Theorem~3.2 in Ref.~\cite{JS68} states that the median for the Hamming weight of $x$ is equal to $\ell$, given our assumption that the mean is equal to $\ell$.  This implies the inequality:
  \begin{align}
    \abs{\bbE_{x\sim P_\mu}[s]} = \abs{\Pr_{x\sim P_\mu}[|x|>\ell]-\Pr_{x\sim P_\mu}[|x|<\ell]} \leq \pi.
  \end{align}
  Let $T$ denote the transition matrix of the update step in which we resample a random coordinate of $x$.  This process is reversible and has stationary distribution $P_\mu$. Furthermore, $T$ has a spectral gap equal to $1/n$\footnote{To see this, observe that the entire nontrivial spectrum of $T$ (not including $1$) can be generated as an eigenvalue with corresponding eigenvector $\phi_S:\{0,1\}^n\to \bbR$ defined through $\phi_S(x) = \prod_{i\in S}(x_i-\frac{1}{1+e^{\mu_i}})$ for any $S\subseteq [n]$. The largest such eigenvalue comes from, say, $S=\{1\}$, and is equal to $1-1/n$.}.
  
  For $f,g:\{0,1\}^n\to\bbC$, define $\langle f,g\rangle \coloneq \bbE_{x\sim P_\mu}[f(x)^*g(x)]$.  Reversibility means that $T$ is Hermitian with respect to this inner product and the gap means that
\be
\langle f,(I-T) f\rangle \geq \frac 1n \Var(f) = \frac1n \qty(\langle f,f\rangle - \langle f,1\rangle^2) \ .
\ee
(Here, the ``1" in the inner product denotes the constant function equal to one.) Let $f=s-\bbE[s]$ so that $\langle f,1\rangle=0$.  Rearranging the gap equation  we obtain
\ba
\langle f, T^nf \rangle 
& \leq \qty(1-\frac 1n)^n \langle f,f\rangle
\\ & \leq e^{-1} (\langle s,s\rangle - \bbE[s]^2)
\\ & \leq e^{-1} (1-\bbE[s]^2) \ .
\ea
On the other hand, $\langle f, T^nf \rangle = \langle s,T^n s\rangle - \bbE[s]^2$.
Thus 
\be
\langle s,T^n s\rangle \leq e^{-1} (1-\bbE[s]^2) + \bbE[s]^2 \leq e^{-1} + (1-e^{-1})\pi^2 \ .
\ee
The left-hand side can also be related to the underlying walk as
\be
\langle s,T^n s\rangle = 
\bbE[s(x(0))s(x(n))]
\geq 2\Pr[s(x(0))s(x(n))=1]-1.\ee
Combining these we find that
\be
\Pr[s(x(0))s(x(n))=1] \leq \frac{1}{2} (1+e^{-1})(1+\pi^2) \ .
\ee
Now suppose that $\pi \leq 1/2$.  Then the right-hand side is strictly less than $7/8$ and there is a constant probability of success.  On the other hand with probability $\pi$ the first sample $x(0)$ is already in $H_\ell$, so if $\pi>1/2$ then we also obtain a constant probability of success.  Either way, each iteration of the sampling algorithm succeeds with probability $\Omega(1)$, so the expected number of repetitions is $O(1)$.
\end{proof}

\subsection{Inference}\label{sec:fast_algorithms}
Finally we turn to the inverse problem of finding $\mu^\star\in\bbR^n$ such that $q(\mu^\star)\equiv \sum_{x\in H_\ell}p_{\mu^\star}(x)x = q$ given some input $q$, which we refer to as {\em inference}. This turns out to be the slowest part of our overall randomized truncation algorithm, and is accomplished by solving the dual optimization problem, \cref{eq:max_ent_dual}. We discuss several algorithms for the task.

First we sketch why we should expect a polynomial-time solution in the first place.  
Naively \cref{eq:max_ent_primal} is a convex program with $|H_\ell| = \binom{n}{\ell}$ variables, while \cref{eq:max_ent_dual} has only $n$ variables. This observation on its own, however, does not substantially reduce the complexity of the problem since the dual objective requires summing over $\binom{n}{\ell}$ terms. Fortunately, the exponentially large sum can be evaluated efficiently, due to \cref{lem:marginal_computation}, and this is the basis for efficient algorithms. One also needs to show that an accurate solution of \cref{eq:max_ent_dual} yields a good approximation of $p^\star$.
For this, note that the dual objective function can be rewritten as
\be\label{eq:objective vs divergence}
g_q(\mu)= \mathrm{H}(p^\star) + \rmD_{\textnormal{KL}}\infdivx{p^\star}{p_\mu}.
\ee
Then by Pinsker's inequality, an $\eps^2$-accurate solution $\mu$ to the dual problem gives a solution $p_{\mu}$ to the primal problem which is $O({\eps})$ in total variation distance from the optimum $p^\star$.  For this reason we express our goal as $\eps^2$-accuracy.

As in \Cref{sec:faster sampling}, we focus on arithmetic complexity. However, for the inference problem we also keep track of the complexity in terms of a boundary distance parameter $\eta$ which satisfies
\begin{align}
  q_i\in [\eta, 1-\eta]\quad \textnormal{for every}\ i\in [n].
\end{align}
We will see later that the eigenvalue of the Hessian are related to $q_i(1-q_i)$, so we might expect that $\eta$ affects the convergence of naive gradient descent.
At the same time, we would ideally like the run time to scale polylogarithmically in $\eta$, since this is polynomial in the bit length of the input. We can also assume that $\eta\geq \eps/n$, because one can round each $q_i$ to be in $[\eta,1-\eta]$ by incurring error $\leq n\eta$.

In \cref{tab:inference}, we summarize several algorithms for the task of minimizing \cref{eq:objective vs divergence}. An algorithm with rigorous guarantees, running in time polynomial in $n$, $\log(1/\eps)$, and $\log(1/\eta)$ can be obtained through the ellipsoid method, as shown by Straszak and Vishnoi~\cite{SV19}. This leads to the following, which we state without proof.
\begin{thm}[Implied by {\cite[{Theorem~15}]{SV19}}]\label{thm:efficient_computation_of_weights}
  Assume $q_j\in [\eta, 1-\eta]$ for each $j\in [n]$, $\eps\in (0,1]$, and $\sum_jq_j=\ell$.  A variant of the ellipsoid algorithm runs in time $\poly(n,\log(1/\eps),\log(1/\eta))$ and computes an $\eps^2$-accurate solution $\mu\in \bbR^n$ to the dual problem in \cref{eq:max_ent_dual} such that $\norm{\mu}_{2}\leq \poly(n,\log(1/\eps),\log(1/\eta))$.
\end{thm}
With this theorem in hand, we have enough to conclude the argument that a maximum-entropy construction satisfies the four desired criteria listed at the beginning of this section, and that a polynomial-size representation of this distribution, in the form of the parameters $\mu$, can be computed in polynomial time. Nonetheless, one might hope for alternative methods with an explicit and favorable dependence on $n$, which we now provide.

In fact, we will present three such algorithms, along with performance guarantees.  Each of these three algorithms outperforms the other two for some choice of parameters.  We summarize their performance in \Cref{thm:inference-algs} below. In particular, the box-constrained Newton algorithm (BN) runs asymptotically at least as quickly as the SVD used in MPS truncation, assuming that the SVD is done in time $O(n^3)$ instead of $O(n^\omega)$.  Even if the SVD uses the  $O(n^{2.807})$-time Strassen method, the runtime will be comparable, especially if $\ell$ is significantly smaller than $n$.
  
To get intuition for the optimization problem, observe that $g_q(\mu)$ is convex with gradient and Hessian given by
\be
\nabla_\mu g_q(\mu) = q - q(\mu) \qand
\nabla_\mu^2 g_q(\mu) \eqcolon \Sigma(\mu) = Q(\mu) - q(\mu)q(\mu)^T,
\label{eq:g-grad-Hess} \ee
which we can compute in time $O(n\ell)$ and $O(n^2)$ by
\cref{lem:marginal_computation}.  Letting $D(\mu)$ be the diagonal of
$\Sigma(\mu)$, we will prove in \cref{lem:hessian_bound} that
$\Sigma(\mu) \preceq 2 D(\mu)$.  Since the entries of $D(\mu)$ are
$q_i(\mu)(1-q_i(\mu))$, this means that the Hessian has eigenvalues
that are always $\leq 1/2$ but approach zero as $q_i(\mu)$ approaches 0 or 1.
If the Hessian were in fact $D(\mu)$ and not $\Sigma(\mu)$ then this would mean, in the language of convex optimization, that  $g_q$ would be $1/4$-smooth but only $O(\eta)$-strongly convex.    (The comparison between $\Sigma(\mu)$ and $D(\mu)$ can be made more rigorous by \cref{lem:hessian_bound,conj:diag}.)

As a result we can expect gradient descent to work but with some slowdowns depending
on $\eta$ (or equivalently $\norm{\mu}_\infty$).  These slowdowns can be dealt with using either Newton's method or a variant of gradient descent that does not require strong convexity.

\paragraph{APGD.} The simplest algorithm we can use is gradient descent.  A somewhat improved  variant is the FISTA algorithm with constant step size from Section 4 of \cite{BT09}, also known as accelerated projected gradient descent (APGD). The algorithm starts with $\mu^{(0)} = \nu^{(1)}=0$, $\beta_1=2$ and for $t=1,\ldots,T$ sets
\ba
\mu^{(t)} & = \nu^{(t)} + \alpha (q - q(\nu^{(t)})) \\
\alpha & \text{ is the largest value }\leq 2 \text{ s.t. } \norm{\mu^{(t)}}_2 \leq R_2 \\
\beta_{t+1} & = 1 + \sqrt{1+\beta_t^2} \\
\nu^{(t+1)} &= \mu^{(t)} + \frac{\beta_t - 2}{\beta_{t+1}}(\mu^{(t)} - \mu^{(t-1)})
\ea
We will choose $R_2$ and $T$ later.  This iteration is dominated by the cost of computing $q(\cdot)$ which takes time $O(n\ell)$ by \cref{lem:marginal_computation}.  The runtime guarantee requires that an approximate solution exists within a 2-norm ball of radius $R_2$.  \Cref{lem:R-bound} below implies $R_2$ can be taken to be $O(\sqrt{n}\log(1/\eta))$, improving on the $R_2\leq \poly(n)$ bound of \cite{SV19}. Finally, Thm 4.4 of \cite{BT09} states that
\begin{align}
  g_q(\mu^{(t)}) - g_q(\mu^*)\leq \frac{2\beta R_2^2}{(t+1)^2},
\end{align}
implying that it suffices to take $T = \lceil \sqrt{2} R_2/\eps\rceil \lesssim \sqrt n\log(1/\eta)\eps^{-1}$ to achieve $\eps^2$-accuracy.
We believe that a bound of the form $T \leq \tilde O(\eta^{-1/2}\log(1/\eps))$ is also likely, assuming a constrained version of Thm 3.18 of \cite{bubeck2015convex} exists, but we do not explore this question further here.

\paragraph{BN.}  The cautious step sizes of APGD can be improved using Newton's method; specifically, the box-constrained Newton method from~\cite{CMTV17}, for which the basic iterate is: given a function $f\colon \bbR^n\to \bbR$ to minimize and $k> 0$, set $x\to x + \frac{1}{k}\tilde{y}$ where $\tilde{y}\in \bbR^n$ is any vector such that $\norm{\tilde{y}}_{\infty}\leq k$ and
\begin{align}
  \frac{1}{2}\tilde{y}^T A \tilde{y} + b^T\tilde{y} \leq \frac{1}{2}\cdot \min_{\norm{y}_\infty\leq 1}\left(\frac{1}{2}y^T A y + b^Ty\right)
\end{align}
for $A=\frac{e^2}{k}\nabla^2 f(x)$ and $b=\frac{1}{k}\nabla f(x)$. In the case where $A$ is symmetric and diagonally dominant (SDD), $\tilde{y}$ can be computed efficiently with $k=O(\log n)$, as shown in \cite{CMTV17}. If, in addition, we have the {\em second-order robustness} condition
\begin{align}
 \frac{1}{c}\nabla^2f(x)\preceq\nabla^2 f(x+h)\preceq c\nabla^2 f(x)\quad \textnormal{for every}\ \norm{h}_\infty\leq 1,
\end{align}
where $c > 0$ is some constant, then we can apply Theorem~3.2 in \cite{CMTV17} to bound the runtime of the solver. Namely, an $\eps^2$-approximate solution can be computed in time
\begin{align}
  \widetilde{O}\left(n^2(1+R_\infty)\log\left(\frac{f(x_0)-f(x^\star)}{\eps}\right)\right)
\end{align}
where $x_0\in \bbR^n$ is some initial point, $x^\star\in\bbR^n$ is a minimizer of $f$, and $R_\infty$ is a bound on the diameter of the level set at $x_0$. We will now show that if we define $f$ through
\begin{align}
  f(x) = g_q\left(\frac{x}{4\ell}\right) + \frac{\eps^2}{8n}\sum_{i=1}^n \cosh\left(\frac{x_i}{M}\right)\quad \textnormal{where}\quad M\coloneq 8\ell \log(1/\eta),
\end{align}
then the above SDD condition and second-order robustness condition are satisfied, and $R_\infty = O(\ell\log(1/\eta)\log(n/\eps))$ with the intial point $x_0=0$. This is similar to the regularization approach in Section 4 of \cite{CMTV17}. Moreover, the choice of $M$ allows us to conclude that an $\eps^2/4$-minimizer $x$ of $f$ yields an $\eps^2/2$-minimizer $\mu = x/(4\ell)$ of $g_q$, through
\begin{align}
g_q(\mu)\leq f(x)\leq f(4\ell\mu^\star)+\frac{\eps^2}{4} = g_q\left(\mu^\star\right) + \frac{\eps^2}{8n}\sum_{i=1}^n\cosh\left(\frac{4\ell\mu^\star_i}{M}\right) + \frac{\eps^2}{4}\leq g_q(\mu^\star)+\frac{\eps^2}{2}
\end{align}
where $\mu^\star$ is the minimizer of $g_q$, and we used \Cref{lem:R-bound} to assume $\norm{\mu^\star}_\infty\leq 2\log(1/\eta)$.
In summary, this gives an approximate solution to our original optimization problem with the desired accuracy in time
\begin{align}
  \widetilde{O}\left(n^2\ell\log(1/\eta)\log(1/\eps)\right).
\end{align}
First, to see that $\nabla^2f$ is SDD, write
\be
  \nabla^2 f(x) = \frac{1}{16\ell^2}\,\Sigma\!\left(\frac{x}{4\ell}\right)
  + \frac{\eps^2}{8nM^2}\,\mathrm{diag}\!\left(\cosh\!\left(\frac{x_i}{M}\right)\right),
  \label{eq:hessian-of-f}
\ee
using the chain rule together with \cref{eq:g-grad-Hess}. The second term is a
strictly positive diagonal matrix, and $\Sigma(\mu)$ is also
SDD for every $\mu\in\bbR^n$, by \Cref{lem:hessian_bound}.

Next, the fact that $f$ is second-order robust with $c=\ee$, follows from \Cref{lem:second_order_robustness}. In particular, this bounds the first term, while the second term can be bounded using hyperbolic trigonometric identities.

Finally, for the bound on the diameter of the level set of $f$ at the initial point $x_0=0$, note that
\be
  f(0) = g_q(0) + \frac{\eps^2}{8n}\sum_{i=1}^n\cosh(0)
  = \log\binom{n}{\ell} + \frac{\eps^2}{8}
\ee
while, for any $x$ such that $f(x)\leq f(0)$ and any $i\in[n]$,
\be
  \frac{\eps^2}{8n}\cosh\!\left(\frac{x_i}{M}\right)
  \leq \frac{\eps^2}{8n}\sum_{j=1}^n\cosh\!\left(\frac{x_j}{M}\right)
  \leq f(0) - g_q\!\left(\frac{x}{4\ell}\right)
  \leq \log\binom{n}{\ell} + \frac{\eps^2}{8}.
\ee
Using $\cosh(t)\geq \ee^{|t|}/2$ and $\log\binom{n}{\ell}\leq n$, this implies
$
  |x_i| \leq M\cdot O\left(\log(n/\eps)\right)
$
for every $i$, giving
$R_\infty = O(\ell\log(1/\eta)\log(n/\eps))$ as claimed.
    
\paragraph{IPFF.}  This third algorithm and its analysis are entirely due to \cite{Chen94}, except for our faster algorithm for computing $q(\mu)$.  IPFF improves a 1940 algorithm called Iterative Proportional Fitting~\cite{deming1940least} with the simple modification of holding the largest weight fixed.
The algorithm assumes $q_n = \max_i q_i$ and initializes $w_i^{(0)} := q_i$. Let $z_{-i}(t)$ denote the partition function $Z^{(\ell-1)}_\mu$ where $e^{-\mu}$ is the vector length $n-1$ obtained by deleting the $i$-th entry of $w_i^{(t-1)}$.
Then for $t=1,2,\ldots,T$ and $i\in [n-1]$ we set
\be
w_i^{(t)} = q_i \frac{z_{-i}}{z_{-n}}
\label{eq:IPFF}\ee
We set all $w_n^{(t)}=q_n$.
The intuition behind \cref{eq:IPFF} is that $w_i=e^{-\mu^\star_i}$ is the unique fixed point of the map and Theorem 2 of \cite{Chen94} proves that the distance to this fixed point shrinks by at least $q_n$, implying that $T=O((1-q_n)^{-1}\log(1/\eps)) \leq O(\eta^{-1} \log(1/\eps))$ steps suffice.
To obtain the total runtime, our \cref{lem:marginal_computation} shows how to carry out each step in time $O(n^2\ell)$.

We summarize the above discussion in the following theorem which compares the performance of all three algorithms. We also present these results with prior ones in \cref{tab:inference}.
\begin{thm}\label{thm:inference-algs}
  The following three algorithms each compute an $\eps^2$-accurate solution $\mu\in \bbR^n$ to the dual problem in \cref{eq:max_ent_dual}.
  \benum
\item {\em Accelerated projected gradient descent (APGD)} in time $O(n^{1.5}\ell\eps^{-1}\log(\ell/\eta))$.
\item {\em Box-constrained Newton's method (BN) } in time $\widetilde{O}(n^2\ell\log(1/\eta)\log(1/\eps))$. 
  \item {\em Iterative Proportional Fitting with $\mu_n$ fixed (IPFF)} in time $O(n\ell^2 \eta^{-1} \log(1/\eps))$.
\eenum
\end{thm}

\begin{table}
  \centering
  \begin{tabular}{lll}
    \toprule
    Method & Upper Bound on Runtime & References\\
    \midrule
    Ellipsoid & $\poly(n, \log(1/\eps), \log(1/\eta))$ &
     {\cite{SV19}} \\
     Box-constrained (BC) Newton & $n^2\ell \log(1/\eta)\log(1/\eps)$
     &   Thm.~\ref{thm:inference-algs}, \cite{CMTV17,CKV20} \\
     \parbox[t]{15em}{Diagonal BC Newton} & $n\ell^2 \log(1/\eta)\log(1/\eps)$
     & \parbox[t]{15em}{assuming \cref{conj:diag}}\\
    IPF & unknown &  {\cite{deming1940least}} \\
    IPF with $\mu_n$ fixed & $n\ell^2 \eta^{-1} \log(1/\eps)$ &  {\cite{Chen94}} \\
    \parbox[t]{15em}{Accelerated PGD} & $n^{1.5}\ell\eps^{-1}\log(1/\eta)$ &
           Thm.~\ref{thm:inference-algs}, \cite{BT09} \\
    \bottomrule
  \end{tabular}
  \caption{
    Comparison of different algorithms for the inference problem of finding $\mu^\star \in\bbR^n$ such that $q(\mu^\star)=q$, given $q\in\bbR^n$.  Runtimes are written omitting the $\tilde O()$ and depend on $n$ (length of bit string), $\ell$ (Hamming weight), $\eps$ (trace-distance accuracy), and $\eta = \min_i \{q_i, 1-q_i\}$.  
    \label{tab:inference}
    }
\end{table}

The proof of \cref{thm:inference-algs} relied on the following lemmas.

\begin{lem}\label{lem:hessian_bound}
  The covariance matrix $\Sigma(\mu) = Q(\mu) - q(\mu)q(\mu)^T\in \bbR^{n\times n}$ is symmetric and diagonally dominant, and satisfies
  \begin{align}
    0\preceq \Sigma(\mu)\preceq 2D(\mu)
  \end{align}
  where $D(\mu)$ is the diagonal approximation to $\Sigma(\mu)$ with entries $D(\mu)_{ii}=\Sigma(\mu)_{ii}$ and $D(\mu)_{ij}=0$ for $i\neq j$.
\end{lem}
\begin{proof}
  The operator $\Sigma(\mu)$ is positive semidefinite because it is a covariance matrix: for any $v\in \bbR^n$ we have $\langle v, \Sigma v\rangle = \bbE[\langle v, q-x\rangle^2] \geq 0$.  For the upper bound, note that $\sum_j Q(\mu)_{ij} = \ell q(\mu)_i$ and $\sum_j q(\mu)_j = \ell$, so $\sum_j \Sigma(\mu)_{ij}=0$.  Along with the fact that $\Sigma(\mu)_{ii}>0$ and $\Sigma(\mu)_{ij}<0$ for $i\neq j$, this implies that $\Sigma(\mu)-D(\mu) \preceq D(\mu)$.
\end{proof}

\begin{lem}\label{lem:R-bound}
  Given $\mu\in\bbR^n$ with $q = q(\mu)$, assume that $q_j\in [\eta,1-\eta]$ for each $j\in [n]$,
  Then
  $
  \max_i \mu_i - \min_i \mu_i \leq 2\log(1/\eta).
  $
\end{lem}
In particular, we can always rescale $\mu$ (like in \cref{lem:rescale}), so this implies that there exists an optimal $\mu^*$ with $\norm{\mu^*}_\infty \leq 2\log(1/\eta)$.

\begin{proof}
  Consider the ratio of probabilities (under the distribution $p_\mu$)
    \ba
        \Delta_{ij} = \frac{\Pr[x_i=1,x_j=0]}{\Pr[x_i=0,x_j=1]} = \frac{q_i - Q_{ij}}{q_j - Q_{ij}}
    \ea
    using that $\Pr[x_i=1,x_j=0] = \Pr[x_i=1] - \Pr[x_i=1, x_j=1]$ and similar for $i$ and $j$ reversed.
    On the other hand, for any $x$ with $x_i=1$ and $x_j=0$, we have $p_\mu(x-e_i+e_j)  = p_\mu(x)\cdot e^{-\mu_i+\mu_j}$.  Summing over all such strings we find
    \ba
        \Delta_{ij} = e^{-\mu_i + \mu_j}.
    \ea
    Next, suppose $\mu_j \geq \mu_i$, so $q_j \leq q_i$.
    Recall that that we have negative correlation, so $Q_{ij} \leq q_i q_j$.
    Then we get
    \ba
    e^{\mu_j - \mu_i} = \Delta_{ij} 
    = \frac{q_i - Q_{ij}}{q_j - Q_{ij}}
    \leq \frac{q_i - q_i q_j}{q_j - q_i q_j} 
    = \frac{q_i(1 - q_j)}{q_j(1 - q_i)} \leq \frac{(1 - \eta)^2}{\eta^2}
    \ea
    where we use that $t \mapsto \frac{q_i - t}{q_j - t}$ is an increasing function in $t$ for $t < q_j$, and our assumption on $q$.
    We conclude that $\mu_j - \mu_i \leq 2 \log(1/\eta)$.
\end{proof}
Finally, we have the following statement pertaining to the second-order robustness of the objective function $g_q$. The special case where $\ell=n$ appears in~\cite{CKV20}.
\begin{lem}\label{lem:second_order_robustness}
  Let $\mu,h\in \bbR^n$ and $\Sigma(\mu)$ be the Hessian defined in \cref{eq:g-grad-Hess}. It holds that
  \ba
    \exp(-4 \ell \norm{h}_\infty)\Sigma(\mu) \preceq \Sigma(\mu + h) \preceq \exp(4 \ell \norm{h}_\infty)\Sigma(\mu).
    \label{eq:self-concordance}
  \ea
\end{lem}
\begin{proof}
Let $P=p_\mu$ and $P^\prime=p_{\mu+h}$, so $P^\prime(x)=P(x)e^{-\langle h,x\rangle}/\E_{y\sim P}[e^{-\langle h, y\rangle}]$ and define $\kappa\coloneq \max_{x\in H_\ell}\frac{P^\prime(x)}{P(x)}$. Note that we have
\begin{align}
  \kappa\leq e^{\max_{x,y\in H_\ell}\langle h, x-y\rangle}\leq e^{2\ell\norm{h}_\infty}\label{eq:kappa_ineq_1}.
\end{align}
Also, for any function $F\colon \{0,1\}^n\to \bbR$ we have
\begin{align}
  \underset{x\sim P^\prime}{\Var}[F(x)] &= \frac{1}{2}\E_{x\sim P^\prime}\E_{y\sim P^\prime}(F(x)-F(y))^2\\
  &\leq \kappa^2 \E_{x\sim P}\E_{y\sim P}(F(x)-F(y))^2\\
  &=\kappa^2\underset{x\sim P}{\Var}[F(x)]\label{eq:variance_ineq_1}.
\end{align}
If $F(x)=\langle v,x\rangle$ for some $v\in \bbR^n$ then $\underset{x\sim P}{\Var}[F(x)]=v^T\Sigma(\mu)v$ and $\underset{x\sim P^\prime}{\Var}[F(x)]=v^T\Sigma(\mu+h)v$. Since this holds for arbitrary $v\in \bbR^n$ we have $\Sigma(\mu+h)\preceq e^{4\ell\norm{h}_\infty}\Sigma(\mu)$ by \cref{eq:variance_ineq_1,eq:kappa_ineq_1}, as required. The reverse inequality follows by applying the above with $\mu+h$ in place of $\mu$
and $-h$ in place of $h$, since $\norm{-h}_\infty = \norm{h}_\infty$.
\end{proof}

\subsection{Heuristics and conjecture for faster inference}
A natural approach that has been proposed~\cite{Chen00} is to simply use Newton's method, which has the iterates
\begin{align}
  \mu^{(t+1)} = \mu^{(t)} - \alpha \Sigma(\mu^{(t)})^{-1}(q-q(\mu^{(t)})).
\end{align}
The parameter $\alpha$ can be set to 1 or  found by line search.  Finally $\Sigma$ is technically noninvertible, since $\Sigma\mathbf{1}=0$.  So we need to take the inverse in the subspace orthogonal to $\mathbf{1}$.  This is essentially the same as the BN method, with the $\ell_\infty$ box, and is a good heuristic algorithm even though it currently lacks a provable convergence guarantee.

There is a yet faster method that avoids the need to compute all of $Q$ or to solve a linear system of equations.  We can approximate the covariance matrix $\Sigma(\mu)$ with its diagonal $D(\mu) = \text{diag}(q_i(\mu)(1-q_i(\mu)))$.  This means we only need to compute $q(\mu)$ at each step instead of $Q(\mu)$, so the time per iteration becomes $O(n \ell)$. In practice, this method appears to converge quickly to the optimal solution, and we use it in our numerics.

Finally we state a plausible conjecture which would imply a provable bound on this heuristic method.  In \cref{lem:hessian_bound} we proved that $0\preceq \Sigma(\mu) \preceq 2 D(\mu)$.  We believe that the lower bound on $\Sigma$ can be significantly improved. We conjecture that $D(\mu) \preceq \Sigma(\mu)$ on the subspace orthogonal to ${\bf 1}$. Formally,
\begin{con}\label{conj:diag}
$\lambda_2(D(\mu)^{-1/2}\Sigma(\mu)D(\mu)^{-1/2}) \geq 1$ where $\lambda_2$ denotes the second smallest eigenvalue.
\end{con}
This conjecture was independently proposed by \cite{AHM26}. We have run numerical tests and not found any counterexamples to this conjecture.   In the difficult case where some $q_i$ is very close to 0 or 1, $\lambda_2$ is only slightly larger than 1, suggesting that 1 is the exact numerical value of the minimum of $\lambda_2(D(\mu)^{-1/2}\Sigma(\mu)D(\mu)^{-1/2})$ over all $\mu$.  However, for algorithmic purposes it would be enough to prove
that $\lambda_2(D(\mu)^{-1/2}\Sigma(\mu)D(\mu)^{-1/2}) \geq c$ for some universal constant $c>0$.

If \cref{conj:diag} holds then we can replace $\Sigma(\mu)$ with $D(\mu)$ in the BN algorithm and will need only $O(1)$ more iterations.  Now each iteration runs in time $O(n\ell)$ instead of $O(n^2)$, so the total runtime becomes $O(n\ell^2 \log(\eta/\eps))$.

\paragraph{Note for arXiv v3:}  As noted in the introduction, we are posting this paper now to comply with the QIP 2027 requirement that submissions be available on arxiv.org.  However, we expect to further update the paper in the near future.  Recently [arXiv:2607.13990v2], a proof of \cref{conj:diag}.  With a little work, this implies that a Newton-type algorithm can find $\mu$ from $q$ in time $\tilde O(n\ell)$.  We will expand on this point in a future version.  Even without this change, BC Newton is proven to run faster than SVD in \cref{thm:inference-algs}.

\section*{Acknowledgments}

This material is based upon work supported by the U.S. Department of
Energy, Office of Science, National Quantum Information Science
Research Centers, Quantum Systems Accelerator (Award
No.~DE-SCL0000121).  This work was partially funded by the Wellcome
Leap as part of the Quantum for Bio (Q4Bio) Program.  AWH was
supported by the Simons Foundation (MP-SIP-00001553, AWH), the NSF
(PHY-2325080) and a Fulbright Scholar Grant.  FW was supported by the
European Union (ERC, ASC-Q, 101040624).  We thank Garnet Chan, Jakob
G\"unther, Allen Liu, John Martyn, Ankur Moitra, Norbert Schuch, and
Ewin Tang for helpful conversations.

We used Claude Opus 5.0 to generate proof ideas for some of the algorithms in \Cref{sec:fast_algorithms}, and unsuccessfully used ChatGPT to attempt to prove \Cref{conj:diag}.

\appendix

\section{Background material}
For the convenience of the reader, we collect a few standard definitions and facts from convex analysis, nonlinear optimization and matrix analysis used in the proofs in the main text.

\subsection{Minimax and zero-sum games}\label{sec:minimax}
A minimax inequality of the form
\begin{align}\label{eq:minimax_ineq}
  \max_{y\in \cY}\min_{x\in \cX}f(x,y)\leq \min_{x\in\cX}\max_{y\in\cY}f(x,y)
\end{align}
holds for arbitrary choices of sets $\cX,\cY$ and a function $f:\cX\times \cY\to\bbR$. Here and throughout when we take maxima and minima of functions, it is implied that these values are attained. Such an inequality admits a game-theoretic interpretation as follows. Consider a zero-sum, two-player game where the payoff function for the $x$-player is $-f(x,y)$ and that for the $y$-player is $+f(x,y)$. Then the right-hand side of \cref{eq:minimax_ineq} represents the worst-case payoff for the $y$-player when they ``go second'' in the game, i.e., the strategy for the $y$-player may depend on $x$. Through this lens, the minimax inequality amounts to the statement that going second in the game is a privilege. We encounter a specific example of such a two-player game in \Cref{sec:trace-dist}.

A minimax theorem claims that the minimax inequality holds with equality under certain conditions. One example which suffices for our purposes is Sion's Minimax Theorem, provided below. We refer the interested reader to \cite{simons1995minimax} for a history of other minimax theorems.
\begin{thm}[{Implied by Thm.~3 in \cite{fan1952fixed}}]\label{thm:fan_minimax}
  Let $\cX$ and $\cY$ be compact, convex subsets of a Hilbert space and $f: \cX\times\cY\to \bbR$ be a real-valued function such that:
  \begin{enumerate}
    \item $f(\cdot,y)$ is continuous and convex on $\cX$ for every fixed $y\in \cY$, and
    \item $f(x,\cdot)$ is continuous and concave on $\cY$ for every fixed $x\in\cX$.
  \end{enumerate}
  Then
  \begin{align}\label{eq:minimax_eq}
    \max_{y\in\cY}\min_{x\in \cX} f(x,y) = \min_{x\in\cX}\max_{y\in\cY} f(x,y).
  \end{align}
\end{thm}

A related concept is the existence of a saddle point, also referred as an equilbrium point in the context of two-player games.
\begin{dfn}[Saddle point]
  Given arbitrary sets $\cX$ and $\cY$ and a function $f:\cX\times\cY\to\bbR$, a point $(x^\star,y^\star)\in\cX\times\cY$ is a saddle point if
  \begin{align}\label{eq:saddle_point_def}
    f(x^\star,y)\leq f(x^\star,y^\star)\leq f(x,y^\star)
  \end{align}
  for every $x\in\cX$ and $y\in\cY$.
\end{dfn}
The following lemma states that the existence of a saddle point is equivalent to a minimax equality where both the minimum and the maximum are attained.
\begin{lem}[Well-known]\label{lem:saddle_point_existence}
  A saddle point $(x^\star,y^\star)\in\cX\times \cY$ for the function $f:\cX\times\cY\to\bbR$ exists if and only if
  \begin{align}
    \max_{y\in\cY} \min_{x\in\cX} f(x,y)=\min_{x\in\cX}\max_{y\in\cY}f(x,y)\eqcolon f^\star.
  \end{align}
  Furthermore, for such a saddle point we have $f(x^\star, y^\star)=f^\star$.
\end{lem}
\begin{proof}
  Suppose \cref{eq:minimax_eq} holds. Then we may choose
  \begin{align*}
    x^\star\in\argmin_{x\in\cX}\max_{y\in\cY}f(x,y),\qquad y^\star\in \argmax_{y\in\cY}\min_{x\in\cX}f(x,y)
  \end{align*}
  in which case
  \begin{align*}
    \max_{y\in\cY} \min_{x\in\cX} f(x,y) & = \min_{x\in\cX}f(x,y^\star)
    \leq f(x^\star,y^\star)
    \leq \max_{y\in\cY}f(x^\star, y)
    =\min_{x\in\cX}\max_{y\in\cY}f(x,y)
  \end{align*}
  and, by the minimax equality, the two inequalities hold with equality. Hence, the relations in \cref{eq:saddle_point_def} are satisfied.

  On the other hand, suppose there is a saddle point $(x^\star, y^\star)\in\cX\times\cY$. Then
  \begin{align}
    \min_{x\in\cX}\max_{y\in\cY}f(x,y) & \leq \max_{y\in\cY}f(x^\star,y) \leq f(x^\star,y^\star)\leq \min_{x\in\cX}f(x,y^\star)\leq \max_{y\in\cY}\min_{x\in\cX}f(x,y)
  \end{align}
  where the second and third inequalities use the relations in \cref{eq:saddle_point_def}. But this is an inequality in the opposite direction of the minimax inequality, which always holds.
\end{proof}

\subsection{Convex analysis}\label{sec:convex}
We now list some standard facts from convex analysis. See \cite{rockafellar1970convexanalysis,hl2001fundamentalsofconvex} for an introduction to these concepts. Throughout, we use the extended reals and let $\dom f$ denote the set of all values $x$ for which $f(x)$ is finite.
\begin{dfn}[Subgradient and subdifferential]
  Let $f:\bbR^d\to\bbR\cup\{+\infty\}$ be a convex function. A vector $g\in\bbR^d$ is a \emph{subgradient} of $f$ at $x\in\bbR^d$ if
  \begin{align}
    f(y)-f(x)\;\ge\;\langle g,\,y-x\rangle \qquad \text{for all }y\in\bbR^d.
  \end{align}
  The set of all such subgradients is the \emph{subdifferential} of $f$ at $x$, denoted $\partial f(x)$.
\end{dfn}
Some rules familiar from calculus are applicable in many cases.
\begin{lem}[Implied by Theorem~23.8 in \cite{rockafellar1970convexanalysis}]\label{lem:subdifferential_addition}
Let $f,g:\bbR^d\to \bbR\cup \{+\infty\}$ be convex functions and suppose that the relative interiors of $\dom f$ and $\dom g$ have a point in common. Then for any $x\in\bbR^d$ it holds that $\partial (f+g)(x) = \partial f(x) + \partial g(x)$.
\end{lem}
\begin{dfn}[Normal cone]\label{def:normal_cone}
  For $\Omega\subseteq\bbR^d$ and $x\in\Omega$, the \emph{normal cone} of $\Omega$ at $x$ is
  \begin{align}
    \cN_\Omega(x) = \{g\in\bbR^d:\ \langle g,\,y-x\rangle\le 0\ \text{for all }y\in\Omega\}.
  \end{align}
\end{dfn}
The above definition allows us to state the following first-order optimality condition, which is valid even when the objective function is not differentiable.
\begin{lem}\label{lem:fooc}
  Let $f:\bbR^d\to\bbR$ be a convex function and $\Omega\subseteq\bbR^d$ be a convex set with nonempty relative interior. A point $x\in\Omega$ minimizes $f$ over $\Omega$ if and only if
  \begin{align}
    0\in \partial f(x) + \cN_\Omega(x).
  \end{align}
\end{lem}
\begin{proof}
    It is straightforward to verify that for any convex function $h:\bbR^d\to\bbR\cup \{+\infty\}$ it holds that $x$ minimizes $h$ if and only if $0\in\partial h(x)$. This is sometimes referred to as \textit{Fermat's Rule}. Now, the set of optimal solutions for $\min \{f(x): x\in \Omega\}$ is identical to that for the optimization $\min\{f(x)+\delta_\Omega(x): x\in \bbR^d\}$ where $\delta_\Omega:\bbR^d\to\bbR\cup\{+\infty\}$ satisfies $\delta_\Omega(x)=0$ if $x\in\Omega$ and $\delta_\Omega(x)=+\infty$ otherwise. Furthermore, making use of \Cref{def:normal_cone} we see that $\partial \delta_\Omega(x) = \cN_\Omega(x)$ for any $x\in \Omega$. Putting $h=f+\delta_\Omega$ and applying \Cref{lem:subdifferential_addition} completes the proof.
\end{proof}
If $f$ is differentiable and the feasible set is specified using equality constraints only then the condition above gives back the usual method of Lagrange multipliers.

Given a convex set $\cX$ in a finite-dimensional real inner-product space and $f:\cX\to\bbR$, the \emph{Fenchel conjugate} $f^*:\cX\to\bbR$ is given by
\begin{align}
  f^*(y) = \sup_{x\in\cX} \bigl\{ \langle x,y\rangle - f(x) \bigr\}.
\end{align}
When $f$ is a squared norm we have the following lemma. We omit the proof, which is a simple application of the Cauchy-Schwarz inequality.
\begin{lem}\label{lem:fenchel_conjugate_of_norm}
  Let $\norm{\cdot}$ be a norm on $\bbR^d$ and set $f(x)=\tfrac12\norm{x}^2$. Then
  \begin{align}
    f^*(y) = \tfrac12\norm{y}_*^2,
  \end{align}
  where $\norm{\cdot}_*$ is the dual norm, defined as $\norm{y}_* = \sup_{x\in\cX} \Bqty{\langle x,y\rangle :  \norm{x}\leq 1}$.
\end{lem}
\begin{cor}\label{cor:dual_norm_fenchel_conj}
  Let $\norm{\cdot}$ be a norm on $\bbR^d$ with dual $\norm{\cdot}_*$. For any $x,y\in\bbR^d$, the following are equivalent:
  \begin{enumerate}
    \item $y\in\partial(\tfrac12\norm{\cdot}^2)(x)$, and $x$ optimizes $\max_{x'}\{\langle x',y\rangle-\tfrac12\norm{x'}^2\}$,
    \item $x\in\partial(\tfrac12\norm{\cdot}_*^2)(y)$, and $y$ optimizes $\max_{y'}\{\langle x,y'\rangle-\tfrac12\norm{y'}_*^2\}$, and
    \item $\norm{x}^2+\norm{y}_*^2=2\langle x,y\rangle$.
  \end{enumerate}
\end{cor}
\begin{proof}
    To see that $x$ optimizes $\max_{x^\prime}\{\langle x^\prime, y\rangle -\frac{1}{2}\norm{x^\prime}^2\}$ whenever $y\in \partial(\frac{1}{2}\norm{\cdot}^2)(x)$, or the analogous statement for $y$, use the definition of the subgradient. The equivalence between the first and third statements then follows from \Cref{lem:fenchel_conjugate_of_norm}, and similarly for the equivalence between the second and third statements. 
\end{proof}

\subsection{Spectral stability}
The following is a well-known consequence of Weyl's Inequality concerning the spectral stability of Hermitian matrices. We let $\lambda_1(A), \lambda_2(A),\dots $ denote the eigenvalues of a Hermitian matrix $A$ sorted in nondecreasing order.
\begin{fact}[Spectral stability]\label{fact:spectral_stability}
  Let $k, d\in\bbZ_{>0}$ be such that $1\leq k\leq d$ and $A, B\in\bbC^{d\times d}$ be Hermitian matrices. It holds that
  \begin{align*}
    \lambda_k(A+B)-\lambda_k(A)\in [\lambda_1(B),\lambda_d(B)].
  \end{align*}
\end{fact}

\section{Proof of \Cref{lem:optimal_soln_to_k_supp}}\label{sec:k_supp_lem_proof}
    We make use of the following optimization fact, the proof of which we omit.
    \begin{lem}\label{lem:norm_opt_fact}
        Let $a,b,w_1,w_2\in\mathbb{R}_{>0}$ such that $b> a$ and $w_1,w_2\in [a,b]$. The unique optimal solution to the optimization problem
        \begin{equation}
        \begin{aligned}
            \underset{x\in \bbR^2}{\textnormal{max.}}\quad & w_1x_1 + w_2x_2 - \frac{1}{2}(x_1^2 + x_2^2)\\
           \textnormal{s.t.} \quad & b\geq x_1\geq x_2\geq a
        \end{aligned}
        \end{equation}
        is given by $x_1=x_2=\frac{1}{2}(w_1+w_2)$ if $w_1 < w_2$ and $x = (w_1,w_2)^\top$ otherwise.
    \end{lem}
    We now prove \Cref{lem:optimal_soln_to_k_supp}. The sufficiency of the condition in \cref{eq:k_supp_opt_soln_condn} is explicitly shown in the proof of Proposition~2.1 in \cite{AFS12}, so we restrict ourselves to the argument for necessity. We may assume without loss of generality that $v$ is nonnegative and sorted in nonincreasing order. We will first take $x$ to be sorted as well before justifying this assumption at the end of the proof. Clearly, an optimal solution $x$ which is sorted such that $x_1\geq x_2\geq\dots\geq x_{\ell-1}$ exists.
    
    First, similar to the proof of Proposition~2.1 in \cite{AFS12}, we give necessary conditions on the form of such an optimal solution. Consider the optimization
    \begin{equation}
    \begin{aligned}
        \textnormal{max.}\quad & f(x)\coloneq \sum_{i=1}^{\ell-1}v_ix_i - \frac{1}{2} \sum_{i=1}^k x_i^2 \\
        \textnormal{s.t.}\quad & x\in\bbR^d\\
        \quad & x_1\geq x_2\geq \dots \geq x_{\ell-1}> 0.
    \end{aligned}
    \end{equation}
    By the definition of $\ell$, a sorted solution $x$ to the maximization in \cref{eq:k_supp_opt_problem} is optimal only if its first $\ell-1$ entries are optimal for the problem above. Additionally, an optimal solution must have $x_{k}=x_{k+1}=\dots=x_{\ell-1}$. We can rewrite the objective function for such an optimal solution as 
    \begin{align}
        f(x)=\sum_{i=1}^{k-2}v_ix_i - \frac{1}{2}\sum_{i=1}^{k-2}x_i^2 + v_{k-1}x_{k-1} + s^v_k x_k - \frac{1}{2}x_{k-1}^2 - \frac{1}{2}x_k^2.
    \end{align}
    If $s^v_k < v_{k-1}$ then the optimal solution must satisfy $x_i=v_i$ for $i=1,\dots, k-1$ and $x_k=s^v_k$. (The optimization just becomes the Fenchel conjugate of half the 2-norm squared in this case.) Otherwise, by \Cref{lem:norm_opt_fact} we must have $x_{k-1} = x_{k}$ and therefore
    \begin{align}
        f(x)=\sum_{i=1}^{k-3}v_ix_i - \frac{1}{2}\sum_{i=1}^{k-3}x_i^2 + v_{k-2}x_{k-2}+s^v_{k-1} x_{k-1} - \frac{1}{2}x_{k-2}^2 - x_{k-1}^2
    \end{align}
    at the optimum. Continuing in this manner, we arrive at precisely the condition in the statement of the lemma.
    
    Now, for \emph{any} optimal solution $x$ (not necessarily sorted), we must have $x_j\geq x_{j+1}$ for all $j\in [\ell-2]$ such that $v_j\neq v_{j+1}$; otherwise, the objective function is less than that obtained by permuting the indices $j$ and $j+1$ in the vector $x$. On the other hand, if $v_j=v_{j+1}$ for some $j\in [\ell-2]$, then the objective function is unchanged by permuting the indices $j$ and $j+1$ in the vector $x$. Therefore, if there exists an optimal solution with $x_j < x_{j+1}$ for some $j\in [\ell-2]$, this contradicts the above necessary condition on the form of the sorted optimal solution. In particular, we cannot have simultaneously that $x_j \neq x_{j+1}$ and $v_j=v_{j+1}$ for entries $j$ and $j+1$ of $x$ which match $v$. For the case where $j=k-r-1$, we cannot have that the condition in \cref{eq:condition k support norm} is satisfied when $v_{j}=v_{j+1}$.

\bibliographystyle{hyperalpha}
\bibliography{references}

\newcommand{\etalchar}[1]{$^{#1}$}
\begin{thebibliography}{CRPW12}

\bibitem[ABW15]{ABW15}
Amir Abboud, Arturs Backurs, and Virginia~Vassilevska Williams.
\newblock If the current clique algorithms are optimal, so is {V}aliant's parser.
\newblock In {\em 2015 IEEE 56th Annual Symposium on Foundations of Computer Science}, pages 98--117, 2015,  \href{http://arxiv.org/abs/1504.01431}{{\ttfamily arXiv:1504.01431}}.

\bibitem[AFS12]{AFS12}
Andreas Argyriou, Rina Foygel, and Nathan Srebro.
\newblock Sparse prediction with the $ k $-support norm.
\newblock {\em Advances in Neural Information Processing Systems}, 25, 2012,  \href{http://arxiv.org/abs/1204.5043}{{\ttfamily arXiv:1204.5043}}.

\bibitem[AHM26]{AHM26}
Nima Anari, Alireza Haqi, and Eric Ma.
\newblock On rounding on the hypersimplex, 2026,  \href{http://arxiv.org/abs/2606.00996}{{\ttfamily arXiv:2606.00996}}.

\bibitem[AKT22]{AKT22}
Seiseki Akibue, Go~Kato, and Seiichiro Tani.
\newblock Quadratic improvement on accuracy of approximating pure quantum states and unitary gates by probabilistic implementation, 2022,  \href{http://arxiv.org/abs/2111.05531}{{\ttfamily arXiv:2111.05531}}.

\bibitem[AKT24a]{AKT24-state}
Seiseki Akibue, Go~Kato, and Seiichiro Tani.
\newblock Probabilistic state synthesis based on optimal convex approximation.
\newblock {\em npj Quantum Information}, 10(1):3, 2024,  \href{http://arxiv.org/abs/2303.10860}{{\ttfamily arXiv:2303.10860}}.

\bibitem[AKT24b]{AKT24-unitary}
Seiseki Akibue, Go~Kato, and Seiichiro Tani.
\newblock Probabilistic unitary synthesis with optimal accuracy.
\newblock {\em ACM Transactions on Quantum Computing}, 5(3):1–27, August 2024,  \href{http://arxiv.org/abs/2301.06307}{{\ttfamily arXiv:2301.06307}}.

\bibitem[BBPS96]{BBPS96}
C.~H. Bennett, H.~J. Bernstein, S.~Popescu, and B.~Schumacher.
\newblock Concentrating partial entanglement by local operations.
\newblock 53:2046--2052, 1996,  \href{http://arxiv.org/abs/quant-ph/9511030}{{\ttfamily arXiv:quant-ph/9511030}}.

\bibitem[Bha97]{Bhatia1997}
Rajendra Bhatia.
\newblock {\em Matrix Analysis}.
\newblock Springer New York, 1997.

\bibitem[BT09]{BT09}
Amir Beck and Marc Teboulle.
\newblock A fast iterative shrinkage-thresholding algorithm for linear inverse problems.
\newblock {\em SIAM Journal on Imaging Sciences}, 2(1):183--202, 2009.

\bibitem[Bub15]{bubeck2015convex}
Sébastien Bubeck.
\newblock Convex optimization: Algorithms and complexity.
\newblock {\em Foundations and Trends in Machine Learning}, 8(3-4):231--357, 11 2015,  \href{http://arxiv.org/abs/1405.4980}{{\ttfamily arXiv:1405.4980}}.

\bibitem[Cam17]{Campbell17}
Earl Campbell.
\newblock Shorter gate sequences for quantum computing by mixing unitaries.
\newblock {\em Physical Review A}, 95(4), April 2017,  \href{http://arxiv.org/abs/1612.02689}{{\ttfamily arXiv:1612.02689}}.

\bibitem[Cam19]{qDRIFT}
Earl Campbell.
\newblock Random compiler for fast hamiltonian simulation.
\newblock {\em Physical Review Letters}, 123(7), August 2019,  \href{http://arxiv.org/abs/1811.08017}{{\ttfamily arXiv:1811.08017}}.

\bibitem[CDL94]{Chen94}
Xiang-Hui Chen, Arthur~P. Dempster, and Jun~S. Liu.
\newblock Weighted finite population sampling to maximize entropy.
\newblock {\em Biometrika}, 81(3):457--469, 09 1994.

\bibitem[CG19]{CG19}
Eric Chitambar and Gilad Gour.
\newblock Quantum resource theories.
\newblock {\em Reviews of modern physics}, 91(2):025001, 2019,  \href{http://arxiv.org/abs/1806.06107}{{\ttfamily arXiv:1806.06107}}.

\bibitem[CGJ{\etalchar{+}}16]{Chen_2016_quantifying}
Jianxin Chen, Shane Grogan, Nathaniel Johnston, Chi-Kwong Li, and Sarah Plosker.
\newblock Quantifying the coherence of pure quantum states.
\newblock {\em Phys. Rev. A}, 94:042313, Oct 2016,  \href{http://arxiv.org/abs/1601.06269}{{\ttfamily arXiv:1601.06269}}.

\bibitem[Che00]{Chen00}
Sean~X Chen.
\newblock General properties and estimation of conditional bernoulli models.
\newblock {\em Journal of Multivariate Analysis}, 74(1):69--87, 2000.

\bibitem[CHKT21]{CHKT21}
Chi-Fang Chen, Hsin-Yuan Huang, Richard Kueng, and Joel~A Tropp.
\newblock Concentration for random product formulas.
\newblock {\em PRX Quantum}, 2(4):040305, 2021,  \href{http://arxiv.org/abs/2008.11751}{{\ttfamily arXiv:2008.11751}}.

\bibitem[CK87]{poly-mult87}
David~G Cantor and Erich Kaltofen.
\newblock Fast multiplication of polynomials over arbitrary rings.
\newblock {\em Acta Inf}, 28, 1987.

\bibitem[CKV20]{CKV20}
L.~Elisa Celis, Vijay Keswani, and Nisheeth Vishnoi.
\newblock Data preprocessing to mitigate bias: A maximum entropy based approach.
\newblock In Hal~Daumé III and Aarti Singh, editors, {\em Proceedings of the 37th International Conference on Machine Learning}, volume 119 of {\em Proceedings of Machine Learning Research}, pages 1349--1359. PMLR, 13--18 Jul 2020,  \href{http://arxiv.org/abs/1906.02164}{{\ttfamily arXiv:1906.02164}}.

\bibitem[CMTV17]{CMTV17}
Michael~B Cohen, Aleksander Madry, Dimitris Tsipras, and Adrian Vladu.
\newblock Matrix scaling and balancing via box constrained newton's method and interior point methods.
\newblock In {\em 2017 IEEE 58th Annual Symposium on Foundations of Computer Science (FOCS)}, pages 902--913. IEEE, 2017,  \href{http://arxiv.org/abs/1704.02310}{{\ttfamily arXiv:1704.02310}}.

\bibitem[COS19]{COS19}
Andrew~M. Childs, Aaron Ostrander, and Yuan Su.
\newblock Faster quantum simulation by randomization.
\newblock {\em {Quantum}}, 3:182, September 2019,  \href{http://arxiv.org/abs/1805.08385}{{\ttfamily arXiv:1805.08385}}.

\bibitem[CRPW12]{convex-inverse12}
Venkat Chandrasekaran, Benjamin Recht, Pablo~A. Parrilo, and Alan~S. Willsky.
\newblock The convex geometry of linear inverse problems.
\newblock {\em Foundations of Computational Mathematics}, 12(6):805–849, October 2012,  \href{http://arxiv.org/abs/1012.0621}{{\ttfamily arXiv:1012.0621}}.

\bibitem[CRS10]{cvetkovic2010introduction}
Drago{\v{s}}~M Cvetkovi{\'c}, Peter Rowlinson, and Slobodan Simi{\'c}.
\newblock {\em An introduction to the theory of graph spectra}, volume~75.
\newblock Cambridge university press Cambridge, 2010.

\bibitem[DS40]{deming1940least}
W~Edwards Deming and Frederick~F Stephan.
\newblock On a least squares adjustment of a sampled frequency table when the expected marginal totals are known.
\newblock {\em The Annals of Mathematical Statistics}, 11(4):427--444, 1940.

\bibitem[EY36]{Eckart1936}
Carl Eckart and Gale Young.
\newblock The approximation of one matrix by another of lower rank.
\newblock {\em Psychometrika}, 1(3):211--218, September 1936.

\bibitem[Fan52]{fan1952fixed}
Ky~Fan.
\newblock Fixed-point and minimax theorems in locally convex topological linear spaces.
\newblock {\em Proceedings of the National Academy of Sciences of the United States of America}, 38(2):121--126, 1952.

\bibitem[Fer15]{Ferris15}
Andrew~J Ferris.
\newblock Unbiased {M}onte {C}arlo for the age of tensor networks, 2015,  \href{http://arxiv.org/abs/1507.00767}{{\ttfamily arXiv:1507.00767}}.

\bibitem[Gha08]{gharibian2008strong}
Sevag Gharibian.
\newblock Strong {NP}-hardness of the quantum separability problem, 2008,  \href{http://arxiv.org/abs/0810.4507}{{\ttfamily arXiv:0810.4507}}.

\bibitem[GKZ25]{GKZ25}
David Gosset, Robin Kothari, and Chenyi Zhang.
\newblock Multi-qubit {T}offoli with exponentially fewer {T} gates, 2025,  \href{http://arxiv.org/abs/2510.07223}{{\ttfamily arXiv:2510.07223}}.

\bibitem[GLS12]{GLS12}
Martin Gr{\"o}tschel, L{\'a}szl{\'o} Lov{\'a}sz, and Alexander Schrijver.
\newblock {\em Geometric algorithms and combinatorial optimization}, volume~2.
\newblock Springer Science \& Business Media, 2012.

\bibitem[Gur03]{gurvits2003classical}
Leonid Gurvits.
\newblock Classical deterministic complexity of {E}dmonds' problem and quantum entanglement.
\newblock In {\em Proceedings of the thirty-fifth annual ACM symposium on Theory of computing}, pages 10--19, 2003,  \href{http://arxiv.org/abs/quant-ph/0303055}{{\ttfamily arXiv:quant-ph/0303055}}.

\bibitem[GWS{\etalchar{+}}25]{Q4Bio-PhaseEstimation}
Jakob Günther, Freek Witteveen, Alexander Schmidhuber, Marek Miller, Matthias Christandl, and Aram Harrow.
\newblock {Phase estimation with partially randomized time evolution}, 2025,  \href{http://arxiv.org/abs/2503.05647}{{\ttfamily arXiv:2503.05647}}.

\bibitem[Has16]{Hastings16}
M.~B. Hastings.
\newblock Turning gate synthesis errors into incoherent errors, 2016,  \href{http://arxiv.org/abs/1612.01011}{{\ttfamily arXiv:1612.01011}}.

\bibitem[HJJ20]{HJJ20}
Jeremy Heng, Pierre~E. Jacob, and Nianqiao Ju.
\newblock A simple {M}arkov chain for independent {B}ernoulli variables conditioned on their sum, 2020,  \href{http://arxiv.org/abs/2012.03103}{{\ttfamily arXiv:2012.03103}}.

\bibitem[HL13]{hillar2013most}
Christopher~J Hillar and Lek-Heng Lim.
\newblock Most tensor problems are {NP}-hard.
\newblock {\em Journal of the ACM (JACM)}, 60(6):1--39, 2013,  \href{http://arxiv.org/abs/0911.1393}{{\ttfamily arXiv:0911.1393}}.

\bibitem[HUL01]{hl2001fundamentalsofconvex}
Jean-Baptiste Hiriart-Urruty and Claude Lemaréchal.
\newblock {\em Fundamentals of Convex Analysis}.
\newblock Springer Berlin Heidelberg, 2001.

\bibitem[HW23]{HW23}
Matthew Hagan and Nathan Wiebe.
\newblock Composite quantum simulations.
\newblock {\em Quantum}, 7:1181, November 2023,  \href{http://arxiv.org/abs/2206.06409}{{\ttfamily arXiv:2206.06409}}.

\bibitem[JLP{\etalchar{+}}18]{kcoh18}
Nathaniel Johnston, Chi-Kwong Li, Sarah Plosker, Yiu-Tung Poon, and Bartosz Regula.
\newblock Evaluating the robustness of k-coherence and k-entanglement.
\newblock {\em Phys. Rev. A}, 98(2):022328, 2018,  \href{http://arxiv.org/abs/1806.00653}{{\ttfamily arXiv:1806.00653}}.

\bibitem[JS68]{JS68}
Kumar Jogdeo and Stephen~M Samuels.
\newblock Monotone convergence of binomial probabilities and a generalization of {R}amanujan's equation.
\newblock {\em The Annals of Mathematical Statistics}, 39(4):1191--1195, 1968.

\bibitem[KMB24]{koczor2024probabilistic}
B{\'a}lint Koczor, John~JL Morton, and Simon~C Benjamin.
\newblock Probabilistic interpolation of quantum rotation angles.
\newblock {\em Physical Review Letters}, 132(13):130602, 2024,  \href{http://arxiv.org/abs/2305.19881}{{\ttfamily arXiv:2305.19881}}.

\bibitem[KT24]{KT24}
Anastasia Kireeva and Joel~A. Tropp.
\newblock Randomized matrix computations: themes and variations, 2024,  \href{http://arxiv.org/abs/2402.17873}{{\ttfamily arXiv:2402.17873}}.

\bibitem[LBT19]{LBT19}
Zi-Wen Liu, Kaifeng Bu, and Ryuji Takagi.
\newblock One-shot operational quantum resource theory.
\newblock {\em Phys. Rev. Lett.}, 123:020401, Jul 2019,  \href{http://arxiv.org/abs/1904.05840}{{\ttfamily arXiv:1904.05840}}.

\bibitem[Low21]{Low21}
Guang~Hao Low.
\newblock Halving the cost of quantum multiplexed rotations, 2021,  \href{http://arxiv.org/abs/2110.13439}{{\ttfamily arXiv:2110.13439}}.

\bibitem[MF85]{govind1985unitarily}
Govind~S. Mudholkar and Marshall Freimer.
\newblock A structure theorem for the polars of unitarily invariant norms.
\newblock {\em Proceedings of the American Mathematical Society}, 95(3):331--337, 1985.

\bibitem[Mir60]{mirsky1960symmetric}
Leon Mirsky.
\newblock Symmetric gauge functions and unitarily invariant norms.
\newblock {\em The Quarterly Journal of Mathematics}, 11(1):50--59, 1960.

\bibitem[MR25]{MR25}
John~M. Martyn and Patrick Rall.
\newblock Halving the cost of quantum algorithms with randomization.
\newblock {\em npj Quantum Information}, 11(1), March 2025,  \href{http://arxiv.org/abs/2409.03744}{{\ttfamily arXiv:2409.03744}}.

\bibitem[NC12]{nielsen_chuang}
Michael~A. Nielsen and Isaac~L. Chuang.
\newblock {\em Quantum Computation and Quantum Information: 10th Anniversary Edition}.
\newblock Cambridge University Press, 2012.

\bibitem[OWC20]{OWC20}
Yingkai Ouyang, David~R. White, and Earl~T. Campbell.
\newblock Compilation by stochastic {H}amiltonian sparsification.
\newblock {\em {Quantum}}, 4:235, February 2020,  \href{http://arxiv.org/abs/1910.06255}{{\ttfamily arXiv:1910.06255}}.

\bibitem[Reg17]{Regula18}
Bartosz Regula.
\newblock Convex geometry of quantum resource quantification.
\newblock {\em Journal of Physics A: Mathematical and Theoretical}, 51(4):045303, Dec 2017,  \href{http://arxiv.org/abs/1707.06298}{{\ttfamily arXiv:1707.06298}}.

\bibitem[Roc70]{rockafellar1970convexanalysis}
R.~Tyrrell Rockafellar.
\newblock {\em Convex Analysis}.
\newblock Princeton University Press, 1970.

\bibitem[SAP17]{SAP17}
Alexander Streltsov, Gerardo Adesso, and Martin~B. Plenio.
\newblock Colloquium: Quantum coherence as a resource.
\newblock {\em Reviews of Modern Physics}, 89(4), October 2017,  \href{http://arxiv.org/abs/1609.02439}{{\ttfamily arXiv:1609.02439}}.

\bibitem[Sch11]{SCHOLLWOCK201196}
Ulrich Schollw\"ock.
\newblock The density-matrix renormalization group in the age of matrix product states.
\newblock {\em Annals of Physics}, 326(1):96--192, 2011.
\newblock January 2011 Special Issue.

\bibitem[Sim95]{simons1995minimax}
Stephen Simons.
\newblock Minimax theorems and their proofs.
\newblock In Ding-Zhu Du and Panos~M. Pardalos, editors, {\em Minimax and Applications}, pages 1--23. Springer US, Boston, MA, 1995.

\bibitem[SV14]{singh2014entropy}
Mohit Singh and Nisheeth~K. Vishnoi.
\newblock Entropy, optimization and counting.
\newblock In {\em Proceedings of the Forty-Sixth Annual ACM Symposium on Theory of Computing}, STOC '14, page 50–59, New York, NY, USA, 2014. Association for Computing Machinery,  \href{http://arxiv.org/abs/1304.8108}{{\ttfamily arXiv:1304.8108}}.

\bibitem[SV19]{SV19}
Damian Straszak and Nisheeth~K. Vishnoi.
\newblock Maximum entropy distributions: Bit complexity and stability.
\newblock In Alina Beygelzimer and Daniel Hsu, editors, {\em Proceedings of the Thirty-Second Conference on Learning Theory}, volume~99 of {\em Proceedings of Machine Learning Research}, pages 2861--2891. PMLR, 25--28 Jun 2019,  \href{http://arxiv.org/abs/1711.02036}{{\ttfamily arXiv:1711.02036}}.

\bibitem[TH00]{terhal2000schmidtnumber}
Barbara~M. Terhal and Pawe\l{} Horodecki.
\newblock Schmidt number for density matrices.
\newblock {\em Phys. Rev. A}, 61:040301, Mar 2000,  \href{http://arxiv.org/abs/quant-ph/9911117}{{\ttfamily arXiv:quant-ph/9911117}}.

\bibitem[Til06]{Tille06}
Yves Till\'e.
\newblock {\em Sampling Algorithms}.
\newblock Springer New York, NY, 2006.

\bibitem[Tod24]{todo2024markov}
Synge Todo.
\newblock {Markov Chain Monte Carlo} in tensor network representation, 2024,  \href{http://arxiv.org/abs/2412.02974}{{\ttfamily arXiv:2412.02974}}.

\bibitem[VT99]{VT99}
Guifre Vidal and Rolf Tarrach.
\newblock Robustness of entanglement.
\newblock {\em Phys. Rev. A}, 59:141--155, 1999,  \href{http://arxiv.org/abs/quant-ph/9806094}{{\ttfamily arXiv:quant-ph/9806094}}.

\bibitem[WZYZ25]{wang2025faster}
Yue Wang, Xiao-Ming Zhang, Xiao Yuan, and Qi~Zhao.
\newblock Faster state preparation with randomization, 2025,  \href{http://arxiv.org/abs/2510.12247}{{\ttfamily arXiv:2510.12247}}.

\end{thebibliography}

\end{document}